\documentclass[a4paper, 12pt]{article}

\usepackage[top=2cm, bottom=2cm, left=2.5cm, right=2.5cm]{geometry}
\usepackage{amssymb}
\usepackage{amsmath}
\usepackage{amssymb,amsfonts,amsthm}
\usepackage{bm}
\usepackage{hyperref}
\usepackage{apacite}
\usepackage{setspace}
\usepackage[justification=centering]{caption}
\usepackage{float}
\usepackage{xcolor}
\usepackage{changes}
\usepackage{multirow}
\usepackage{longtable}
\usepackage{supertabular} 
\usepackage{dcolumn}
\usepackage{pdflscape}
\usepackage{float}
\usepackage[shortlabels]{enumitem}
\usepackage{longtable}
\usepackage{tabularx,calc}
\usepackage{booktabs}
\usepackage{acro}
\usepackage{setspace}
\usepackage{afterpage}
\usepackage{lscape}
\usepackage{mathtools}
\usepackage{pst-plot, pstricks} 
\usepackage{fancybox,amssymb,color}
\usepackage{subfigure}
\usepackage{framed}
\usepackage{scalerel}[2016-12-29]
\usepackage{titling}
\usepackage[bottom]{footmisc} 
\usepackage[]{authblk}
\usepackage[title]{appendix}

\newcommand{\thetahat}{\hat{\theta}}

\newcommand{\thetastar}{\theta^*}

\newcommand{\thetastarT}{\theta^{*T}}

\newcommand{\imat}{\text{I}}

\newcommand{\prob}{\mathop{\mathbb{P}}}

\newcommand{\normL}[1]{\big\lVert#1\big\rVert}

\newcommand{\norm}[1]{\left\lVert#1\right\rVert}
\newcommand{\normB}[1]{\Bigg\lVert#1\Bigg\rVert}
\newcommand{\normb}[1]{\Big\lVert#1\Big\rVert}

\newcommand{\ytilde}{\tilde{y}}
\newcommand{\ztilde}{\tilde{z}}
\newcommand{\Ztilde}{\tilde{Z}}

\newcommand{\Vtilde}{\tilde{V}}

\newtheorem{remark}{Remark}
\newtheorem{definition}{Definition} 
\newtheorem{assumption}{Assumption} 
 
\newtheorem{theorem}{Theorem}
\newtheorem{corollary}{Corollary}
\newtheorem*{EIC}{Nonnegative Elastic Irrepresentable Condition (NEIC)}
\newtheorem*{conditionRates}{Rate Condition on Density of Grid (RCDG)}
\DeclareMathOperator*{\argmin}{arg\,min}
\newtheorem{lemma}{Lemma}

\makeatletter
\newcommand{\subalign}[1]{%
	\vcenter{%
		\Let@ \restore@math@cr \default@tag
		\baselineskip\fontdimen10 \scriptfont\tw@
		\advance\baselineskip\fontdimen12 \scriptfont\tw@
		\lineskip\thr@@\fontdimen8 \scriptfont\thr@@
		\lineskiplimit\lineskip
		\ialign{\hfil$\m@th\scriptstyle##$&$\m@th\scriptstyle{}##$\crcr
			#1\crcr
		}%
	}
}
\makeatother

\definecolor{Diceblue}{RGB}{52,116,181}

\hypersetup{
	colorlinks = true,
	linkcolor=Diceblue,
	citecolor=Diceblue,
	urlcolor=Diceblue
}

\definechangesauthor[name={Max}, color=red]{Max}
\definechangesauthor[name={Stephan}, color=orange]{Stephan}

\newcolumntype{C}{@{\extracolsep{0.35cm}}c@{\extracolsep{0pt}}}

\newcolumntype{z}{@{\extracolsep{0.15cm}}c@{\extracolsep{0pt}}}

\begin{document}

	
\setlength{\abovedisplayskip}{3pt}
\setlength{\belowdisplayskip}{12pt}
	
	
\setlength{\jot}{0.5cm}

\begin{titlingpage}

\title{Nonparametric Estimation of the Random Coefficients Model:  An Elastic Net Approach}

\author[$\dagger$]{Florian Heiss}
\affil[$\dagger\ast$]{Heinrich Heine University D\"usseldorf }
\author[$\ddagger$]{Stephan Hetzenecker}
\affil[$\ddagger$]{Ruhr Graduate School in Economics \& \protect\\ University of Duisburg-Essen}
\author[$\ast$]{Maximilian Osterhaus}

\date{September 2019}
	

\maketitle

\begin{abstract}
		
This paper investigates and extends the computationally attractive nonparametric random coefficients estimator of \citeA{fox2011}. We show that their estimator is a special case of the nonnegative LASSO, explaining its sparse nature observed in many applications. Recognizing this link, we extend the estimator, transforming it to a special case of the nonnegative elastic net. The extension improves the estimator's recovery of the true support and allows for more accurate estimates of the random coefficients' distribution. Our estimator is a generalization of the original estimator and therefore, is guaranteed to have a model fit at least as good as the original one. 
A theoretical analysis of both estimators' properties shows that, under conditions, our generalized estimator approximates the true distribution more accurately. 
Two Monte Carlo experiments and an application to a travel mode data set illustrate the improved performance of the generalized estimator. 
		
		

		

\vspace{0.4cm}
		
 \vspace{0.1cm}\newline

\end{abstract}

\vfill
\noindent\hrulefill \\
{\footnotesize{
Financial support by the Deutsche Forschungsgemeinschaft (DFG) project 235577387/GRK1974 and the Ruhr Graduate School in Economics is gratefully acknowledged. \\
$\dagger$ Heinrich Heine University D\"usseldorf, Universit\"atsstr. 1, 40225 D\"usseldorf, Germany. Email: \href{florian.heiss@hhu.de}{florian.heiss@hhu.de} \\
$\ddagger$ Ruhr Graduate School in Economics, RWI - Leibniz Institute for Economic Research, Hohenzollernstr. 1-3, 45128 Essen, Germany and University of Duisburg-Essen, Universit\"atsstr. 12, 45117 Essen, Germany.   Email: \href{mailto:stephan.hetzenecker@rgs-econ.de}{stephan.hetzenecker@rgs-econ.de} \\
$\ast$ Heinrich Heine University D\"usseldorf, Universit\"atsstr. 1, 40225 D\"usseldorf, Germany. Email: \href{mailto:osterhaus@dice.hhu.de}{osterhaus@dice.hhu.de}
}}

\end{titlingpage}


\onehalfspacing

\section{Introduction}

Adequately modeling unobserved heterogeneity across agents is a common challenge in many empirical economic studies. A popular approach to address  unobserved heterogeneity are random coefficient models, which allow the coefficients of the economic model to vary across agents. The aim of the researcher is to estimate the distribution of the random coefficients.  

\citeA{fox2011}, hereafter FKRB, propose a simple and computationally fast estimator that can approximate distributions of any shape. 
The estimator uses a fixed grid where every grid point is a prespecified vector of random coefficients. The distribution function is obtained from the probability weights at the grid points, which are estimated with constrained least squares. In principle, the approach can approximate any distribution arbitrarily closely if the grid of random coefficients is sufficiently dense \cite{mcfadden2000}. 

Applications of the estimator indicate, however, that it tends to estimate only few positive weights and  that it sets the weights at many grid points to zero. As a consequence, the estimator lacks the ability to estimate smooth distribution functions but instead approximates potentially continuous distributions through step functions with only few steps.
Our first contribution is to show that the estimator of FKRB is  Nonnegative LASSO \cite{wu2014lasso} (NNL) with a fixed tuning parameter to explain its sparse nature. 

NNL, which was first mentioned in the seminal work of \citeA{efron2004} as positive LASSO, is a popular model selection method typically used in applications with supposedly sparse models. It is applied in various research fields, e.g., in vaccine design \cite{hu2015}, nuclear material detection  \cite{kump2012}, document classification \cite{arini2013}, and index tracking in stock markets \cite{wu2014lasso}. 
NNL shares the  property of LASSO \cite{tibshirani1996}  that it regularizes the coefficients of the model and shrinks some to zero. This property is observed for the FKRB estimator in different Monte Carlo studies (e.g., \citeNP{fox2011} and \citeNP{fox2016}) and applications to real data (e.g., \citeNP{nevo2016}, \citeNP{illanes2018}, \citeNP{blundell2018} and \citeNP{houde2019}). \citeA{nevo2016} study the demand for residential broadband and estimate that there are only 53 out of 8,626 potentially heterogeneous consumer types. \citeA{illanes2018} use the approach to estimate the demand for private pension plans in Chile and assign positive weights to only 194 of 83,251 grid points. \citeA{blundell2018} analyze firms' reaction to the regulation of air pollution and recover no more than five of the 10,001 potential points. 

In addition to its sparse nature, the connection of the FKRB estimator to NNL reveals the estimator's potentially incorrect selection of grid points under strong correlation. The estimator ``randomly'' selects one out of a group of highly correlated points and sets the remaining weights to zero (see \citeNP{hastie2005}, and \citeNP{hastie09}, for the random behavior of LASSO).  

The estimator's sparsity and ``random'' selection behavior can cause inaccurate approximations of the true distribution through non-smooth distributions with the estimated support possibly deviating from the true distribution's support. The latter can lead to misleading conclusions with respect to the heterogeneity of agents in the population. 
\citeA{fox2016} prove that the estimator identifies the true distribution if the grid of random coefficients  becomes sufficiently dense. However, in practice, the correlation tends to increases with the density of the grid and can become so strong that the optimization problem to the FKRB estimator cannot be solved due to singularity (\citeNP{nevo2016}, Online Supplement). Therefore, the high correlation of a dense grid in combination with the incorrect grid point selection of the estimator under strong correlation can have a drastic impact on the identification of the model. 

Our second contribution is to provide a generalization of the FKRB estimator that is able to accurately approximate continuous distributions even under strong correlation. Recognizing the link to NNL, we add a quadratic constraint on the probability weights. The constraint transforms the estimator to a special case of nonnegative elastic net \cite{wu2014elastic}. The extension mitigates the sparsity and improves the selection of the grid points. Due to the additional flexibility that is introduced with the extension, the estimator adjusts to the degree of correlation among grid points. Note that our generalization always includes the FKRB estimator as a special case such that the model fit cannot be worse for our estimator than the FKRB estimator.

We analyze theoretically, under conditions, that our estimator provides more accurate estimates of the true underlying distribution. For that purpose, we show the selection consistency and derive an error bound on the estimated distributions. 
The analysis of the selection consistency examines the estimator's ability to estimate positive probability weights at grid points that lie inside the true distributions support, and zero weights at points outside the true support. The selection consistency is necessary to approximate the true distribution as accurately as possible. Since the estimated distribution reveals the existing heterogeneity in the population, i.e., agents' varying preferences, recovering the true support points is also important for the correct interpretation of the model.   

The analysis shows that our generalized estimator correctly selects the grid points under less restrictive conditions than the FKRB estimator.
The error bounds on the estimated distribution functions illustrate the positive impact of our extension on the overall approximation accuracy.  
Two Monte Carlo experiments in which we estimate a random coefficients logit model confirm the superior properties of our generalized estimator.

Other nonparametric estimators for the random coefficients model include \citeA{train2008}, \citeA{train2016}, \citeA{burda2008} and \citeA{rossi2005}.  \citeA{train2008} introduces three different estimators that are, in principle, similar to the general approach of FKRB but employ a log-likelihood criterion instead of constrained least squares. \citeA{train2016} suggests approximating the random coefficients' distribution with polynomials, splines or step functions instead of with a fixed grid of preference vectors. The approach substantially reduces the number of required fixed points if the researcher specifies overlapping splines and step functions. Due to the lower number of required fixed points, the approach reduces the curse of dimensionality, which is a shortcoming of the fixed grid approach if the economic model includes a large number of random coefficients. However, both \citeA{train2008} and \citeA{train2016} estimate the respective model with the EM algorithm, which is sensitive to its starting values and is not guaranteed to converge to a global optimum.
\citeA{burda2008} and \citeA{rossi2005} employ a Bayesian hierarchical model to approximate the random coefficients' distribution with a mixture of Normal distributions. Even though the estimator potentially has better finite sample properties, it uses a Markov Chain Monte Carlo technique with a multivariate Dirichlet Process prior on the coefficients, which is computationally more demanding.

The remainder of the paper is organized as follows. Section \ref{sec:estimator} describes the FKRB estimator  and introduces our generalized version. Section \ref{sec:errorbound} derives the condition on the estimators' sign consistency and an error bound on the estimated distribution functions. We present two Monte Carlo experiments in Section \ref{sec:montecarlo} that investigate the performance of our generalized estimator in comparison to the FKRB estimator. Section \ref{sec:application} applies the estimators to the \textit{Mode Canada} data set from the R package \textit{mlogit} \cite{mlogit}. \ref{sec:conclusion} concludes.

\section{Fixed Grid Estimators}\label{sec:estimator}

For the introduction of our estimator, we consider the framework of a random coefficient discrete choice model. The approach, however, is not restricted to discrete choice models but can be applied to any model with unobserved heterogeneous parameters.
Let there be an i.i.d. sample of $N$ observations, each confronted with a set of $J$ mutually exclusive potential outcomes. The researcher observes a $K$-dimensional real-valued vector of explanatory variables $x_{i,j}$ for every observation unit $i$ and potential outcome $j$, and a binary vector $y_i$ whose entries are equal to one whenever she observes outcome $j$ for the $i$th observation, and zero otherwise. 
The goal is to estimate the unknown distribution of heterogeneous parameters $F_0(\beta)$ in the model 

\begin{equation}\label{eq1}
P_{i,j}\left(x\right) = \int g\left(x_{i,j}, \beta\right)dF_0\left(\beta\right)
\end{equation}

where $g\left(x_{i,j}, \beta\right)$ denotes the probability of outcome $j$ conditional on the random coefficients $\beta$ and covariates $x_{i,j}$. The researcher specifies the functional form of $g\left(x_{i,j}, \beta\right)$. A prominent example of Equation (\ref{eq1}) is the multinomial mixed logit model, the state-of-the-art model for demand estimation. In this model, consumer $i$ realizes utility $u_{i,j} = x_{i,j}^T\beta_i + \omega_{i,j}$ from alternative $j$, given product characteristics $x_{i,j}$ and unobserved consumer-specific preferences $\beta_i$. $\omega_{i,j}$ denotes an additive, consumer- and  choice-specific error term. Consumer $i$ chooses alternative $j$ of $J$ alternatives (and an outside good with utility $u_{i,0} = \omega_{i,0}$) if $u_{i,j} > u_{i,l}$ for all $l \neq j$. Under the assumption that $\omega_{i,j}$ follows a type I extreme value distribution, the unconditional choice probabilities, $P_{i,j}(x)$, are of the form 
\begin{equation}\label{logit}
P_{i,j}(x) = \vcenter{\hbox{\scaleto[6ex]{\displaystyle\int}{9ex}}} \frac{\exp\left(x_{i,j}^T\beta\right)}{1 + \sum\limits_{l = 1}^J\exp\left(x_{i,l}^T\beta\right)}d F_0\left(\beta\right).
\end{equation}
$F_0(\beta)$ represents the distribution of heterogeneous consumer preferences in the population and is to be estimated. In most applications, researchers place restrictive assumptions on the functional form of $F_0(\beta)$ in advance, and estimate its parameters from the data. 

\subsection{Fixed Grid Estimator by FKRB}

FKRB propose a simple and fast mixture approach to estimate the underlying random coefficients' distribution without restrictive assumptions on its shape. 
The estimator uses a finite grid of fixed random coefficient vectors as mixture components to construct the distribution from the estimated probability weight of every component. 
The underlying idea of this fixed grid estimator is the transformation of the unconditional choice probabilities in Equation (\ref{eq1}) into a probability model in which $F_0(\beta)$ enters linearly. FKRB derive the linear probability model in two steps: they transform Equation (\ref{eq1}) into a regression model with the random coefficients' distribution as the only unknown term. Adding $y_{i,j}$ to both sides and moving $P_{i,j}$ to the right results in the probability model

\begin{equation}\label{eq2}
y_{i,j} = \int g\left(x_{i,j}, \beta\right)d F_0\left(\beta\right) + \left(y_{i,j} - P_{i,j}\left(x\right)\right). 
\end{equation}

To exploit linearity in parameters, they use a sieve space approximation to the infinite-dimensional parameter $F_0(\beta)$. The sieve space approximation divides the support of the random coefficients $\beta$ into $R$ fixed vectors. Each vector has length $K$, the number of random coefficients included in the model. The location of these vectors is specified by the researcher. With the sieve space approximation, Equation (\ref{eq2}) becomes a simple linear probability model with unknown parameters $\theta = (\theta_1, \ldots, \theta_R)^T$

\begin{equation}\label{eq3}
y_{i,j} \approx \sum\limits_{r = 1}^R g\left(x_{i,j}, \beta_r\right)\theta_r + \left(y_{i,j} - P_{i,j}\left(x\right)\right) 
\end{equation}

where $g(x_{i,j}, \beta_r)$ denotes the conditional choice probability evaluated at grid point $r$. Given the fixed grid of random coefficients, $\mathcal{B}_R = (\beta_1, \ldots, \beta_R)$, the researcher estimates the probability weight $\theta_r$ at every point $r = 1, \ldots, R$.
The linear relationship between the outcome variable and the unknown parameters $\theta$ allows to estimate the mixture weights with the least squares estimator. The linear regression, which regresses the binary dependent variable $y_{i,j}$ on the choice probabilities evaluated at $\mathcal{B}_R$, in total has $NJ$ observations, $J$ ``regression observations'' for every statistical observation unit $i = 1, \ldots, N$ and $R$ covariates $z_{i,j} = (g(x_{i,j}, \beta_1), \ldots, g(x_{i,j}, \beta_R))$.  
By the definition of choice probabilities, the expected value of the composite error term $y_{i,j} - P_{i,j}(x_{i,j})$ conditional on $x_{i,j}$ is zero. Thus, the regression model satisfies the mean-independence assumption of the least squares approach.\\

The estimator of the random coefficients' joint distribution is constructed from the estimated weights
\begin{equation}
\hat{F}\left(\beta\right) = \sum\limits_{r = 1}^R \thetahat_r \ 1\left[\beta_r\leq \beta\right]
\end{equation}
where $\beta$ is an evaluation point chosen by the researcher and   the indicator function $1[\beta_r\leq\beta]$ is equal to one whenever $\beta_r\leq\beta$, and zero otherwise.

To ensure that $\hat{F}(\beta)$ is a valid distribution function, FKRB suggest estimating the weights with the least squares estimator subject to the constraints that the weights are greater than or equal to zero, and sum to one

\begin{align}\label{eq4}
\begin{split}
&\thetahat^{FKRB} = \argmin\limits_{\theta}\frac{1}{NJ}\sum\limits_{i = 1}^N\sum\limits_{j = 1}^J\left(y_{i,j} - \sum\limits_{r = 1}^R \theta_r z_{i,j}^r\ \right)^2 \\
&\text{s.t.}\quad \theta_r\geq0 \quad \forall r \quad\text{and} \quad \sum\limits_{r = 1}^R\theta_r = 1 .
\end{split}
\end{align}
Key to an accurate approximation of $F_0(\beta)$ is the precise estimation of the probability weights at every grid point. Basis to a precise estimation of the probability weights is the consistent selection of the relevant grid points. This requires the constrained least squares estimator to estimate positive weights at all grid points at which $F_0(\beta)$ has a positive probability mass, and zero weights otherwise. While zero weights at grid points inside $F_0(\beta)$'s support cause inaccurate approximations through step functions with only few steps, positive estimates at grid points outside $F_0(\beta)$'s support lead to unreliable estimates of the random coefficients' distribution.

\subsection{Nonnegative LASSO vs. Nonnegative Elastic Net}

To provide a more accurate non-parametric estimator with similar computational advantages, we suggest a simple generalization of the FKRB estimator. Our adjusted version includes the baseline estimator as a special case but allows for smoother estimates of $F_0(\beta)$ when necessary. To derive our estimator, we extend the optimization problem formulated in Equation (\ref{eq4}) by a constraint on the sum of the squared probability weights. 
This additional constraint provides a straightforward way to mitigate the estimator's sparse nature. Our generalized estimator is still simple and computationally fast. 
\subsubsection{Connection to Nonnegative LASSO}

We first illustrate the source of the FKRB estimator's sparsity, which helps to understand its behavior and the intuition behind our extension.

One explanation of the potential sparsity of the estimates is the effect of the nonnegativity constraint. 
\citeA{slawski2013} show that nonnegative least squares estimators exhibit a self-regularizing property that yields sparse solutions. The FKRB estimator restricts the weights  not only to be the nonnegative but also to sum up to one.    

Taking both constraints into account, we recognize that the FKRB estimator is a special case of the nonnegative LASSO (NNL) \cite{wu2014lasso}. 

To show the relation of the FKRB estimator to NNL, we transform the equality constrained problem formulated in Equation (\ref{eq4}) into its inequality constrained form. The constraint that the probability weights sum to one allows us to reparametrize the optimization problem in terms of $R-1$ instead of $R$ unknown parameters. Without loss of generality, one can rewrite the $R$th weight as $\theta_R = 1 -  \sum_{r=1}^{R-1}\theta_r$. Substituting $\theta_R$ in Equation (\ref{eq3}) with $1 - \sum_{r = 1}^{R-1}\theta_r$ gives the inequality constrained optimization problem 

\begin{align}\label{eq5}
\begin{split}
&\thetahat^{\text{FKRB}} = \argmin\limits_{\theta}\frac{1}{NJ}\sum\limits_{i = 1}^N\sum\limits_{j = 1}^J\bigg(\ytilde_{i,j} - \sum\limits_{r = 1}^{R - 1}\theta_r\ztilde_{i,j}^r\bigg)^2 \\
&\quad\quad\text{s.t.}\quad  \theta_r\geq0 \quad \forall r \quad \text{and}\quad \sum\limits_{r = 1}^{R - 1}\theta_r \leq 1 
\end{split}
\end{align}

where $\ytilde_{i,j} = y_{i,j} - z_{i,j}^R$ and $\ztilde_{i,j}^r = z_{i,j}^r - z_{i,j}^R$ for every $r = 1, \ldots, R-1$. Because Equation (\ref{eq5}) is an equivalent form of the optimization problem in Equation (\ref{eq4}), the objective functions are minimized by the same vector of probability weights. The only difference in the inequality constrained problem is the estimation of the $R$th weight, which is calculated after optimization as $\theta_R = 1-\sum_{r=1}^{R-1}\theta_r$, and is not explicitly part of the optimization. By the constraints $\theta_r\geq0 \quad \forall r $ and  $ \sum_{r = 1}^{R - 1}\theta_r \leq 1$, the $R$th weight satisfies the property of a probability weight, $ 1 \geq \theta_R \geq 0 $.

Comparing the FKRB estimator's transformed optimization problem with that of the NNL applied to the linear probability model formulated in Equation (\ref{eq3}), 

\begin{align}\label{eq6}
\begin{split}
&\hat{\theta}^{\text{NNL}} = \argmin\limits_{\theta}\frac{1}{NJ}\sum\limits_{i = 1}^N\sum\limits_{j = 1}^J\bigg(\ytilde_{i,j} - \sum\limits_{r = 1}^{R - 1}\theta_r\ztilde_{i,j}^r\bigg)^2 \\
&\quad\quad\text{s.t.}\quad  \theta_r\geq0 \quad \forall r \quad \text{and}\quad \sum\limits_{r = 1}^{R - 1}\theta_r \leq s,
\end{split}
\end{align}

reveals that the baseline estimator is a special case of NNL with fixed tuning parameter $s = 1$. The constraint that the probability weights sum to one resembles an  $\ell_1$ penalty that regularizes the parameter estimates and shrinks some weights to zero  if the sum of unrestricted weights exceeds one. 

The amount of regularization depends on the size of the unrestricted estimates. The more the sum of the $R-1$ unconstrained weights in Equation (\ref{eq5}) exceeds one, the stronger the shrinkage imposed by the constraint, and the larger the number of potential zero weights. According to \citeA{wu2014lasso}, NNL can result in very sparse models if the constraint is too restrictive. If the sum of the $R-1$ unconstrained weights is less than or equal to one, the constraint has no effect, and the estimated coefficients correspond to the nonnegative least squares solution. 

In addition to its sparse nature, the relation to NNL reveals that the FKRB estimator exhibits a ``random'' selection behavior among grid points. Just like NNL, the estimator has no unique solution when the correlation among choice probabilities evaluated at $\mathcal{B}_R$ is strong. It tends to select one out of a group of highly correlated grid points at random and estimates the remaining to zero (see \citeNP{hastie2005}, and \citeNP{hastie09}, for the random behavior of LASSO). Because the correlation is particularly strong in a dense grid among neighboring grid points, the random selection behavior is especially severe for dense random coefficient grids. This property conflicts with the requirement of a sufficiently fine grid for accurate approximations of $F_0(\beta)$. 


\subsubsection{Elastic Net Estimator}

Extending the FKRB estimator's optimization problem formulated in Equation (\ref{eq5}) by a quadratic constraint on the probability weights alleviates the sparse nature and random selection behavior. The additional constraint is known from ridge regression \cite{hoerl1970} and transforms the FKRB estimator into the nonnegative elastic net \cite{wu2014elastic} with fixed constraint on the $\ell_1$-penalty. Thus, our adjusted estimator minimizes

\begin{align}\label{eq7}
\begin{split}
&\thetahat^{\text{ENET}} = \argmin\limits_{\theta}\frac{1}{NJ}\sum\limits_{i = 1}^N\sum\limits_{j = 1}^J\bigg(\ytilde_{i,j} - \sum\limits_{r = 1}^{R - 1}\theta_r\ztilde_{i,j}^r\bigg)^2 \\
\quad\quad\text{s.t.}&\quad  \theta_r\geq0 \quad \forall r \quad \text{and}\quad \sum\limits_{r = 1}^{R - 1}\theta_r \leq 1 \quad \text{and}\quad \sum\limits_{r = 1}^{R - 1}\theta_r^{\,2}\leq t 
\end{split} 
\end{align} 

where $t$ is a nonnegative tuning parameter specified by the researcher. Having a linear and quadratic constraint on the probability weights ensures a more reliable selection of grid points: 
the quadratic constraint encourages a grouping effect, which allows us to recover highly correlated points inside the true support of $F(\beta)$ together and, hence, reduces the estimator's sparsity. The linear constraint, in turn, retains the LASSO property, which makes it possible to select  weights  inside the support of the true distribution and to estimate zero weights at points outside the true support. 

In addition to the improved selection consistency, the quadratic constraint has the desirable property that it allows the specification of a substantially finer grid of random coefficients. While the FKRB estimator runs into almost perfect collinearity problems  if the grid becomes finer \cite{fox2016}, the quadratic constraint ensures that the optimization problem for our adjusted estimator always has a solution. 
The non-sparse solutions together with the possibility of a finer grid endow our estimator with the ability to provide more accurate and reliable estimated distribution functions.

The specification of the tuning parameter allows adjusting the estimator to the level of correlation among grid points. Smaller values of $t$ give more weight to the quadratic constraint, which enables the joint recovery of grid points if the correlation is strong and, hence, reduces the sparsity of the estimator. 
For decreasing values of $t$, the estimator shrinks the probability weights of highly correlated grid points toward each other and induces an averaging of the estimated weights. 
For any $t \geq 1$, the quadratic constraint does not bind, such that the adjusted estimator simplifies to the baseline estimator. Therefore, our estimator is a generalization of the FKRB estimator given in Equation \eqref{eq5}, including it as a special case.
We recommend choosing the tuning parameter with cross-validation and the one standard error rule based on the mean squared error (MSE) criterion. This approach ensures that our estimator achieves a model fit that is at least as high as the FKRB estimator. If the model fit is highest for  $t \geq 1$, the outcome of our adjusted estimator is the same as that for the estimator by FKRB, while it performs better if the model fit is lowest for some $t<1$. 

Loosely speaking, the improved selection consistency of our generalized estimator leads to  more precise estimates of the probability weights. 
We argue that the FKRB estimator  can lead to potentially biased estimates if the linear constraint is binding. In that case, the estimator shrinks the weights at some grid points to zero despite the positive probability mass of $F_0(\beta)$ at these points. Due to the constraint that the estimated weights sum to one, the incorrect zero weights lead to downward biased estimates at points with positive weights. The FKRB estimator reallocates the probability mass from the points with incorrect zero weights to other points, which imposes an upward bias at these points.
The quadratic constraint potentially reduces the described distortions through its improved selection consistency. As a result of more correct positive probability weights, the quadratic constraint diminishes the reallocation of probability caused by the linear constraint and, therefore, reduces the bias both at points with incorrect zero weights and positive weights.\\

The results of the two Monte Carlo studies presented in Section \ref{sec:montecarlo} demonstrate that the quadratic constraint reduces both the sparsity and the bias. 
Moreover, we derive an  error bound on the estimated probability weights in Section \ref{sec:errorbound} which, under certain conditions, is tighter for our generalized estimator than for the FKRB estimator if the correlation between grid points is strong.


\newpage

\section{Theoretical Analysis of the Estimators’ Properties}\label{sec:errorbound}

The requirement of a sufficiently fine grid, which potentially includes points outside the true support, transforms the fixed grid estimator into a high dimensional regression problem with potentially sparse solutions and highly correlated covariates. Recall that in such a context, an important element of an accurate estimation of $F_0(\beta)$ is the consistent selection of grid points. It guarantees the correct recovery of $F_0(\beta)$'s support, and is fundamental to an undistorted estimation of the probability weights. 
Subsection \ref{subsec:recovery} aims to analyze both estimators' ability to select the correct weights.  To evaluate the overall approximation accuracy of the estimators presented in Section \ref{sec:estimator}, we derive an error bound for the estimated probability weights in Subsection \ref{subsec:bound}. 


Suppose $\thetastar = (\theta_1^{*}, \ldots, \theta_{R-1}^{*})^T$ specifies the vector of probability weights that yields the most accurate discrete approximation, $F^*(\beta) = \sum_{r=1}^R\theta_r^{*}\bm{1}[\beta_r\leq\beta]$ with $\theta_R^{*} = 1 - \sum_{r=1}^{R-1}\theta_r^{*}$, of $F_0(\beta)$ which can be obtained with the estimators for a given fixed grid $\mathcal{B}_R$. 
Furthermore, assume that $F^*(\beta)$ converges to $F_0(\beta)$ for $R$ going to infinity. We use $F^*(\beta)$ as a benchmark to compare the estimated distribution function, $\hat{F}(\beta) = \sum_{r=1}^R\hat{\theta}_r \bm{1  }[\beta_r\leq\beta]$ with $\hat{\theta}_R = 1 - \sum_{r=1}^{R-1}\hat{\theta}_r$, to the true underlying distribution.
The introduction of $F^*(\beta)$ allows us to study the selection consistency and the distance between $\hat{\theta}$ and $\thetastar$.\\

The focus of our analysis is on the impact of the correlation among the grid points on the estimators. 
We show that our generalized estimator is selection consistent under less restrictive conditions on the design matrix. 


Due to the relation of the estimators to the NNL and nonnegative elastic net, respectively, we build on the literature on regularized regression. Our proof of the selection consistency mainly follows  \citeA{jia2010}, who analyze selection consistency of the elastic net under i.i.d. Gaussian errors. Similarly to  \citeA{jia2010}, \citeA{wu2014lasso} and \citeA{wu2014elastic} derive selection consistency of the nonnegative LASSO, and the nonnegative elastic net for i.i.d. Gaussian errors. 

We extend their proof to sub-Gaussian errors and allow for correlation among the $J$ errors that belong to the same observation unit $i$. 
Thereby, we contribute to the literature on the nonnegative elastic net.
Neither \citeA{jia2010} nor \citeA{wu2014elastic} calculate error bounds on the deviation between the estimated and the true coefficients.
Our proof of the error bound on the estimated weights draws from \citeA{takada2017}, who analyze a generalization of the elastic net. We adjust their proof such that it is in line with the probability model in Section \ref{sec:estimator}.


For any $\mathcal{B}_R$, denote the linear probability model corresponding to $F^*(\beta)$  by

\begin{equation}\label{eq:true}
y_{i,j} =  \sum_{r=1}^R \theta_r^{*} z_{i,j}^r  + \epsilon_{i,j}
\end{equation}

where $\epsilon_{i,j} $ is the linear probability error. 
For our analysis of the selection consistency and for the error bound on the estimated weights, we make the following assumptions on the linear probability model in Equation (\ref{eq:true}), and on the data generating process. 

\begin{assumption} \label{ass:Exogen} $ $
	\begin{enumerate}[(i)]
	    \item  $\Big( \epsilon_i = (\epsilon_{i,1},...,\epsilon_{i,J}) \Big)_{i=1}^N $ are independent.
	    \item  $\epsilon_{i,j}$ is sub-Gaussian: $ \mathbb{E}\left[\exp\left(t \epsilon_{i,j} \right) \right] \leq \exp\Big(\frac{\sigma^2 t^2}{2} \Big) \quad (\forall t \in \mathbb{R} ) $ for  $\sigma >0 $.	
	    \item $\big( \Ztilde_i \big)_{i=1}^N $ are $i.i.d.$ with a density bounded from above and each $ \ztilde_{i,j}^r \in [-1,1]$.
		\item $\mathbb{E}\left[ \epsilon_i | \Ztilde_1,...,\Ztilde_N \right] = 0$.
		
	\end{enumerate}
\end{assumption}
$\Ztilde$ refers to the regressor matrix of the transformed model in Equation (\ref{eq5}) and $\tilde{Z}_i$ to the corresponding $J\times R-1$ regressor matrix for observation unit $i$.
Assumption \ref{ass:Exogen} (i) imposes independence across the vectors of errors for each observation unit. It does not assume independence of elements within each vector of errors.  
 Assumption \ref{ass:Exogen} (ii) assumes that the errors are sub-Gaussian with variance proxy $\sigma$. The variance proxy $\sigma$ serves as an upper bound of the variance of the errors and  allows for (conditional) heteroscedasticity.  Note that the error term  in the linear probability model in Equation \eqref{eq:true} is sub-Gaussian with variance proxy $\sigma \leq 1$. This follows from the fact that the error term in the linear probability model is bounded between -1 and 1 since $y_{i,j}$ is either  $0$ or $1$, the weights $\theta_r$ are nonnegative and by  Assumption \ref{ass:Exogen} (iii) $\ztilde_{i,j}^r $  is also bounded between  -1 and 1.  $ \ztilde_{i,j}^r \in [-1,1]$ is satisfied by the logit kernel in Equation \eqref{logit} and other examples such as the kernel of binary choice and of multinomial choice without logit errors (e.g., see \citeNP{fox2016}).   Assumption \ref{ass:Exogen} (iv) holds by the definition of linear probability models. 


\subsection{Selection Consistency} \label{subsec:recovery}

For our analysis of the selection consistency, we adapt the definition of \citeA{zhao2006}. An estimator is defined as equal in sign if $\thetahat_r$ and $\thetastar_r$ have the same sign for every $r = 1, \ldots, R-1$. Due to the nonnegativity of the estimates, the definition implies that $\thetahat$ must be positive at all points in $\mathcal{B}_R$ for which ${\theta}_r^{*} > 0$, and zero at those where ${\theta}_r^{*} = 0$. Therefore, the estimation of the correct signs is equivalent to the correct selection of grid points. If an estimate $\thetahat$ of the true weights $\theta$ is equal in sign, we write $ \thetahat =_s \theta$. 

Our definition only includes $R-1$ points of the transformed model in Equation (\ref{eq7}). That is, we  only identify whether the $R-1$ weights included in Equation (\ref{eq7}) have the correct sign but not whether the last weight $\theta_R = 1-\sum_{r=1}^{R-1}\theta_r$ has the correct sign.  

\begin{definition}\label{def:Consistency}
	An estimate $\thetahat$ is \textbf{sign consistent} if 
	\begin{equation}
	\lim\limits_{N\rightarrow\infty} P\left( \thetahat =_s \thetastar \right) = 1.\notag
	\end{equation}	
\end{definition}

According to Definition \ref{def:Consistency}, an estimator is sign consistent if it estimates a positive weight at every grid point at which $\thetastar > 0$, and zero weights otherwise with probability approaching one as the number of observation units $N$ goes to infinity.

To derive the condition under which our generalized estimator is sign consistent, we make the following notations on the design matrix and probability weights. 
We assume that $\mathcal{B}_R$ includes both grid points inside the support of $F_0(\beta)$, i.e., points at which $\thetastar>0$, and points outside the true support, i.e., at which $\thetastar = 0$. Let $S = \{r\in\{1, \ldots, R-1\} \vert\theta_r^{*}>0\}$ define the index set of grid points at which $\thetastar>0$, and let $S^C = \{r\in\{1, \ldots, R-1\} \vert\theta_r^{*} = 0\}$ denote its complement. The corresponding cardinalities are defined as $s := |S|$ and $ s^C :=|S^C| $. We refer to grid points in $S$ as active grid points and to grid points in $S^C$ as inactive grid points.
$\Ztilde_S$ and $\Ztilde_{S^C}$ denote the sub-matrices of all columns of $\Ztilde$ that are in $S$ and $S^C$, respectively. 

Let $\lambda$ denote the fixed LASSO parameter which corresponds to the Lagrange parameter for $s=1$ in Equation (\ref{eq7}) and $\mu$ the Lagrange version of the ridge tuning parameter $t$ in Equation (\ref{eq7}).\\


Following \citeA{wu2014elastic}, we then obtain the subsequent  condition for the sign consistency of the generalized estimator:

\begin{EIC} \label{cond:eic}
	There exists a positive constant  $\eta > 0$ (independent of $N$) such that
	\begin{equation} \label{NEIC} \nonumber
	\max\limits_{r\in S^C} \frac{1}{NJ}\Ztilde_{S^C}^T\Ztilde_S \left( \frac{1}{NJ}\Ztilde_S^T\Ztilde_S + \mu \imat_S \right)^{-1} \left(\iota_S + \frac{\mu}{\lambda}\thetastar_S  \right)  \leq 1 - \eta 
	\end{equation}
	where $\iota_S$ is a vector of $s$ ones and $I_S$ is the identity matrix. 
\end{EIC}

The \hyperref[cond:eic]{NEIC} is a condition for the correct recovery of support points through our generalized estimator.
The term $\Ztilde_{S^C}^T\Ztilde_S $  restricts the linear dependency between active and inactive grid points. The term $\Ztilde_S^T\Ztilde_S$  measures the linear dependency among active grid points.   In addition to the linear dependence of the regressor matrix, the magnitude of the fixed LASSO parameter and the tuning parameter $\mu$ is taken into account by the \hyperref[cond:eic]{NEIC}. 
For $\mu = 0$, the \hyperref[cond:eic]{NEIC} reverts to the Nonnegative Irrepresentable Condition (NIC), the corresponding condition for selection consistency through the estimator proposed by FKRB.  
In contrast to the \hyperref[cond:eic]{NEIC}, the NIC requires that the inverse of $\Ztilde_S^T\Ztilde_S$ exists, which is not a necessary condition for the \hyperref[cond:eic]{NEIC} to hold.  



We exploit the special structure of our data by incorporating the fact that all $\theta$ are between zero and one, and all elements of $\Ztilde$ between minus one and one.

In line with \citeA{fox2016}, we allow $R(N)$ to depend on the sample size $N$. That is, the larger $N$, the more grid points $R(N)$ can be included into the grid. If $R(N)$ increases, we typically expect the number of positive weights $s(N)$ to increase  if the true distribution $F_0(\beta)$ is sufficiently smooth. The next condition restricts the rate at which $s(N)$ and $R(N)$ can increase with $N$. 
For convenience, we write $s$ and $R$ instead of $s(N)$ and $R(N)$ in the subsequent analyses.  

\newpage

\begin{conditionRates}  $ $
	\begin{enumerate} 
		\item $ \lim\limits_{N\rightarrow\infty} 2 s J \exp\left(-\frac{N   {\xi_{\min}^S(\mu)}^2  \rho^2  }{2  s}\right) = 0 $.
		\item $ \lim\limits_{N\rightarrow\infty} 2 (R-1) J \exp\left(-N \eta^2 \lambda^2 \left( \frac{\xi_{\min}^S(\mu)}{s \sqrt{s} +\xi_{\min}^S(\mu)}  \right)^2\big/2\right) = 0, $
	\end{enumerate} 
	where  $\xi_{\min}^S(\mu) $ denotes the (unrestricted) minimal eigenvalue of $1/(NJ)\Ztilde_S^T \Ztilde_S + \mu\imat_S$ and  $\rho := \min  \limits_{i \in S} \bigg\vert \left( 1/(NJ)\Ztilde_S^T\Ztilde_S + \mu \imat_S \right)^{-1} \left( 1/(NJ)\Ztilde_S^T\Ztilde_S \thetastar_S  - \lambda\iota_S \right) \bigg\vert$. \label{ass:ConditionSelectionLASSO}
\end{conditionRates}

\hyperref[ass:ConditionSelectionLASSO]{RCDG}  requires that $\xi_{\min}^S(\mu)>0$. Otherwise, the condition can never be satisfied. 
This is only restrictive for the FKRB estimator and always holds for its generalization as long as $\mu > 0$ since  $1/(NJ)\Ztilde_S^T \Ztilde_S + \mu\imat_S$ is positive  definite for $\mu >0$ and only positive semidefinite for $\mu=0$. The assumption $\xi_{\min}^S(\mu)>0$ excludes the possibility of perfect collinearity to ensure that the  solution to the FKRB estimator exists. \\

\begin{theorem} \label{theo:ConditionSelectionLASSO}
	Suppose Assumption \ref{ass:Exogen}  holds. Suppose further that \hyperref[cond:eic]{NEIC} and \hyperref[ass:ConditionSelectionLASSO]{RCDG} hold. Then 
	\begin{equation} \nonumber
	 \lim\limits_{N\rightarrow\infty} \prob\left( \thetahat =_s \thetastar \right) = 1.
	\end{equation}
	
\end{theorem}
\begin{proof}
See Appendix \ref{app:ProofOfSelectionConsistency}.
\end{proof}
Theorem \ref{theo:ConditionSelectionLASSO} relies on sufficient conditions for the  estimators to select the true weights. These conditions are more restrictive for the FKRB estimator than for our generalization. 
Since $\xi_{\min}^S(\mu) = \xi_{\min}^S(0)+\mu$, the minimal eigenvalue $\xi_{\min}^S(\mu)$ is  higher for the elastic net than for the LASSO estimator. Furthermore, the \hyperref[cond:eic]{NEIC} holds whenever the NIC is satisfied. This implies that our estimator consistently selects the true support whenever the FKRB estimator does.

The converse is not true since 
the \hyperref[cond:eic]{NEIC} might hold even though NIC does not. 
Thus, Theorem \ref{theo:ConditionSelectionLASSO} reveals that our estimator can select the true weights in cases in which the FKRB estimator cannot.



\subsection{Error Bounds}\label{subsec:bound}

A key requirement for an accurate estimation of $F_0(\beta)$ - in addition to the correct support recovery discussed in Subsection \ref{subsec:recovery} - is the precise estimation of the probability weights. In this section, we derive the error bound for the estimated probability weights and the weights that yield the best discrete approximation of $F_0(\beta)$. 

Let $\mathcal{H}$ denote the set of vectors of length $R$ in  $[-1,1]^R$ for which the $\ell_1$-norm is no greater than $2$
\begin{equation}
\mathcal{H} := \left\{x\in [-1,1]^R \ \Big\vert \ \normL{x}_1  \leq  2  \right\} .    \notag
\end{equation}
The set  $\mathcal{H}$ contains all possible values of $\Delta \thetahat := \thetahat - \thetastar$ since $\thetahat$ and $\thetastar$ are vectors of weights which sum up to $1$. Therefore, it is sufficient to consider elements in  $\mathcal{H}$ when analyzing the potential error $\Delta \thetahat$.

Define the restricted minimum eigenvalue of the real symmetric $R\times R$ matrix\\$1/(NJ)\Ztilde^T\Ztilde+\mu\imat_R$ 
over the set of vectors $\mathcal{H}$ as

\begin{equation}
\xi_{\min}(\mu) := \  \inf\limits_{v\in\mathcal{B}} \ \frac{v^T\big[\frac{1}{NJ}\Ztilde^T\Ztilde + \mu\imat_R\big]v}{\normL{v}_2^2}. \notag
\end{equation}

Because the restricted minimal eigenvalue is greater than or equal to the unrestricted minimal eigenvalue, we use the restricted eigenvalue to derive a tighter error bound. We still assume $\xi_{\min}(\mu) > 0$ which rules out perfect collinearity. 
By the same arguments as in Subsection \ref{subsec:recovery},  $\xi_{\min}(\mu) > 0$  is  always satisfied for our generalized estimator 
with $\mu>0$ and $\xi_{\min}(\mu)>0$ is only restrictive for the FKRB estimator. 

%



Following the proof in \citeA{takada2017}, we obtain an error bound on the $R-1$ estimated probability weights.

\begin{theorem} \label{theo:l2normweights}
	Let $0<\delta\leq 1$. Define $\gamma(N,\delta) := \sqrt{2\log\left(\frac{2(R-1)J}{\delta}\right)\big/N}$. Suppose Assumption \ref{ass:Exogen} holds, and that $\xi_{\min}(\mu) > 0$ for $\mu \geq 0$. Then, for any positive  $k$ such that $\gamma(N,\delta) \leq k\lambda$, it holds with probability $1-\delta$ that 
	\begin{equation} \nonumber
	\normL{\thetahat - \thetastar}_2 \leq \frac{2\sqrt{R-1} \ k\lambda + 2\mu\sqrt{s} \normL{\thetastar_S}_\infty}{\xi_{\min}(\mu)}.
	\end{equation}
\end{theorem}
\begin{proof}
See Appendix \ref{app:ProofErrorBoundWeights}.
\end{proof}

Theorem \ref{theo:l2normweights} holds with probability approaching one as $\delta \to 0$. Because $\gamma(N,\delta)$ decreases in $N$, the error bound becomes tighter if the number of observation units increases. This can be seen from the condition $\gamma(N,\delta)\leq k\lambda$ which requires a smaller constant $k$ for a larger $N$ (and fixed $\lambda$).

The number of grid points leads to a direct increase of the error bound, both through $R$ and $s$, which is  expected to increase with $R$, e.g., if the true distribution is continuous.  The number of grid points also has an indirect effect attributable to the stronger correlation typically associated with an increase in the number of grid points. This effect is captured through the restricted minimum eigenvalue $\xi_{\min}(\mu)$, which decreases if the correlation increases. Hence, an increase in the number of grid points typically leads to a wider error bound on the estimated weights (for a fixed $\mu$). \\

For $\mu = 0$, the bound in Theorem \ref{theo:l2normweights} simplifies to the error bound for the FKRB estimator. A comparison of the bound for $\mu = 0$ and $\mu > 0$ reveals that the extension has two opposing effects on the estimator's precision. First, a direct increasing effect that is captured through the tuning parameter in the numerator of Theorem \ref{theo:l2normweights} and, second, an indirect decreasing effect  via the restricted minimum eigenvalue.  

While the direct effect becomes stronger with the number of true support points $s$, the indirect effect is especially relevant if the correlation among grid points is strong. In that case, the extension leads to an increase of $\xi_{\min}(\mu)$ and hence, to a tighter error bound. The indirect effect becomes particularly important if the design matrix tends to be almost singular, in which case the restricted minimum eigenvalue of the FKRB estimator approaches zero (and the error bound its maximum possible value 2).
Also note that the estimation error for the weight $\theta_R$, which is not included in the bound in Theorem \eqref{theo:l2normweights} and calculated as $\theta_R = 1-\sum_{r=1}^{R-1}\theta_r$, will approach zero whenever $\normL{\thetahat - \thetastar}_2$ is close to zero.

Corollary \ref{core:l2normweights} establishes the condition under which our extension provides a tighter error bound on the estimated weights than the FKRB estimator.

\begin{corollary} \label{core:l2normweights}
	When $ \sqrt{s} \normL{\thetastar_S}_\infty \xi_{\min}(0) <  \sqrt{R-1} \ k \lambda$,
	then the error bound for $\normL{\thetahat - \thetastar}_2 $ in Theorem \ref{theo:l2normweights} is tighter for the generalized estimator than for the FKRB estimator.
\end{corollary}
\begin{proof}
See Appendix \ref{app:ProofErrorBoundCorrollary}.
\end{proof}
Using the error bound on the estimated and true probability weights in Theorem \ref{theo:l2normweights}, we derive a bound on the 
error of the estimated distribution function $\hat{F}(\beta)$ and the best discrete distribution $F^*(\beta)$.

\begin{theorem}  \label{theo:boundDistribution}
Under the assumptions and conditions in Theorem \ref{theo:l2normweights}, it holds at any point $\beta \in \mathbb{R}^K$ with probability $1-\delta$ that
	\begin{equation}
	\vert \hat{F}\left(\beta\right) - F^*(\beta) \vert \leq    \; \frac{4 (R-1) \ k\lambda + 4\mu\sqrt{s(R-1)} \normL{\thetastar_S}_\infty}{\xi_{\min}(\mu)}.  \notag \\
	\end{equation}
\end{theorem}
\begin{proof}
See Appendix \ref{app:ProofErrorBoundDistribution}.
\end{proof}
The bound on the difference between the estimated distribution and the best discrete approximation of $F_0(\beta)$ increases in $R$ and decreases in $\xi_{\min}(\mu)$.
Similarly to Theorem \ref{theo:l2normweights}, the difference in the distributions decreases in $N$ since $k$ may decrease when $N$ increases.

Additionally, \citeA{fox2016} show that, under some regularity conditions, it holds that $\vert F_0(\beta) - F^*(\beta) \vert = O(R^{- \bar{s}/K})$ where $\bar{s}\geq 0$  measures the degree of  smoothness of $F_0(\beta)$\footnote{The density function of $\beta$ is assumed to be  $\bar{s}$-times continuously differentiable.} and $K$ refers to the number of random coefficients. This explains the relevance of Theorem \ref{theo:boundDistribution}  since the difference of $F_0(\beta)$ and $ F^*(\beta) $ becomes negligibly small as $R$ increases and the estimation error can then be well captured by $\vert \hat{F}\left(\beta\right) - F^*(\beta) \vert $.

\section{Monte Carlo Simulation }\label{sec:montecarlo}

We conduct two Monte Carlo experiments to examine the selection consistency and the approximation accuracy of our generalized estimator. The Monte Carlo simulation on the selection consistency uses a discrete distribution with a subset of grid points as support points. 

The second experiment generates the random coefficients from a mixture of two normal distributions. This allows us to study the estimators' ability to estimate smooth distributions.
We use a random coefficients logit model as the true data generating process to generate individual-level discrete choice data. Each observational unit $i$ chooses among $J = 4$ mutually exclusive alternatives and an outside option. For every alternative $j$ and observation unit $i$, we draw the two-dimensional covariate vector $x_{i,j} = (x_{i,j,1}, x_{i,j,2})$ from $\mathcal{U}\left(0,5\right)$ and $\mathcal{U}\left(-3,1\right)$, respectively.
To study the effect of the fixed grid and the number of observation units on the estimators' performance, we run every experiment for different sample sizes, and numbers of grid points. We repeat the experiment for every combination of $R$ and $N$ 200 times to compare the performance of our estimator with the FKRB estimator in terms of selection consistency and accuracy for every setup. All calculations are conducted with the statistical software R \cite{RCore}.

\subsection{Discrete Distribution}\label{subsec:MonteCarloDiscrete}

To study the estimators' selection consistency, we generate the random coefficients $\beta$ from a discrete probability mass function. The estimator successfully recovers the true support from the data if it estimates a positive weight at every support point of $F_0(\beta)$, and zero weights at all points outside its support.

For the support points of $F_0(\beta)$, we select a subset of the grid points from the fixed grid we use for the estimation.
The grid  covers the range $[-4.5, 3.5] \times [-4.5, 3.5]$ with $R = \{25, 81, 289\}$ uniformly allocated grid points. We specify the support of our discrete data generating distribution on $[-4.5, -0.5] \times [-4.5, 0.5]$, and $[-0.5, 4.5] \times [-0.5, 3.5]$, whereby the number of support points varies due to the varying number of grid points. That is, we draw the random coefficients $\beta$ from a discrete mass function with $S = \{17, 49, 161\}$ support points, each drawn with uniform probability weight $\theta_s = 1/S$.\\

In this setup, the data generating process exactly matches the underlying probability model of the fixed grid estimator. This way, we abstract from any approximation errors that can arise from the sieve space approximation of the true underlying distribution. Therefore,  the experiment studies the estimators' selection consistency in the most simple framework possible. 

The two areas of the discrete distribution with positive probability mass simulate two heterogeneous groups of preferences in the population. We estimate every distribution for sample sizes $N = \{1000, 10000\}$.

Figure \ref{fig:MCdiscrete} illustrates the setup of the Monte Carlo experiment for the three data generating distributions. The blue shaded area indicates the support of the discrete mass functions, and the filled blue points inside this area the active grid points. The hollow black points outside the blue shared areas are the inactive grid points that are not used for data generation.  
\\
\captionsetup{position=top}
\begin{figure}[h]%
	\centering
	\caption{ Grid of Monte Carlo Study with Discrete Mass Points}
	\subfigure[$R = 25 $, $S = 17$]{\includegraphics[width=0.32\textwidth]{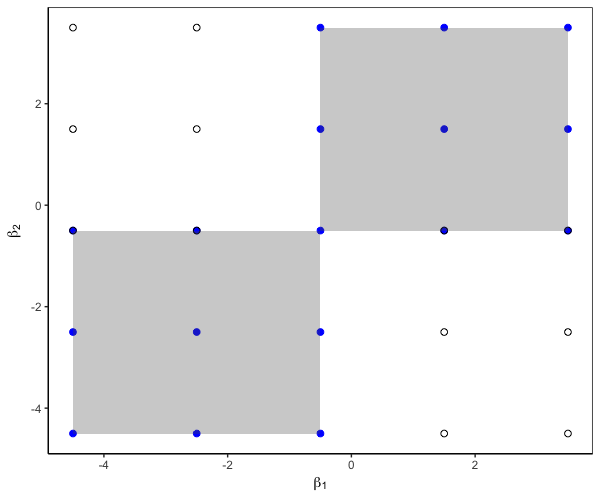}} 
	\subfigure[$R = 81 $, $S = 49 $]{\includegraphics[width=0.32\textwidth]{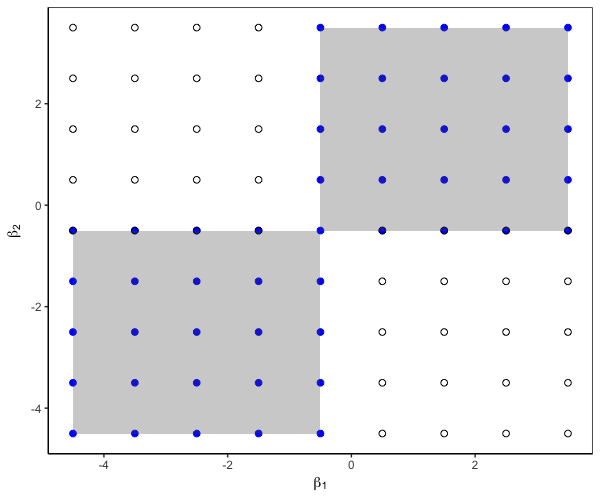}}
	\subfigure[$R = 289 $, $S = 161 $]{\includegraphics[width=0.32\textwidth]{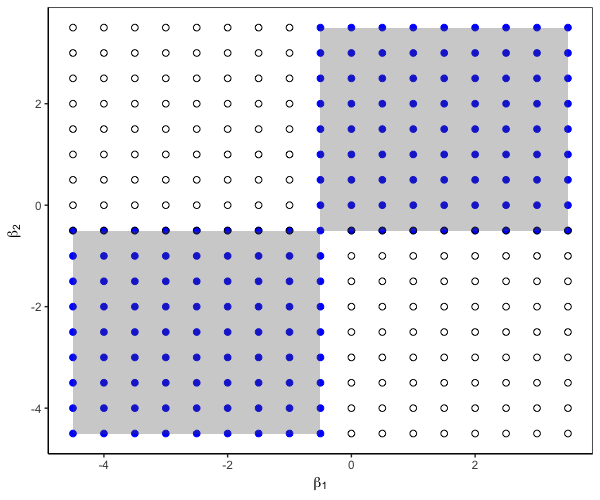}}%
	
	\label{fig:MCdiscrete}
\end{figure}

We choose the optimal tuning parameter $\mu $ for the generalized estimator with 10-fold cross-validation from a sequence of $101$ potential values. For $100$ of these values, we use the sequence suggested by the R package \textit{glmnet} for ridge regression with nonnegative coefficients. We also include $\mu = 0$ in the range of possible values to allow our estimator to simplify to the FKRB estimator if the model fit in the cross-validation is highest for $\mu = 0$. 
The selection of the optimal tuning parameter is based on the mean squared error (MSE) criterion. In addition to the tuning parameter with the lowest MSE, we report the tuning parameter that follows from the one-standard-error rule (OneSe).\footnote{We observe that the curve of the MSE in dependency of $\mu$ tends to be flat and that the $\mu$ chosen by OneSe often corresponds to the largest element of the sequence of tuning parameters suggested by the \textit{glmnet} package. Therefore,  a possible strategy is to choose the largest $\mu$ given by the \textit{glmnet} package to obtain the $\mu$ of OneSe if one wants to avoid cross-validation.}

As  robustness-checks,  we consider  the prediction accuracy of the predicted choice of every observation  and the log-likelihood   as a measure of fit in the cross-validation. We choose the $\mu$ based on the smallest average out-of-sample prediction error  and based on the highest log-likelihood, respectively. The results of the Monte Carlo study for the log-likelihood and predicted choices as selection criteria can be found in Appendix \hyperref[tab:MCdiscreteLong]{A}. They indicate that the MSE and the one-standard-error rule give the best results.

To evaluate the estimators' selection consistency, we calculate the average share of sign consistent estimates. An estimate is sign consistent if it is positive at active grid points, and zero otherwise. A weight is defined as positive if it is greater than $10^{-3}$. To illustrate the sparsity of the estimators' solutions, we report the average number of positive weights and the average share of true positive weights.

Beyond selection consistency, the discrete setup of the Monte Carlo experiment allows us to study the bias of the estimated probability weights.  Denote the estimated weight at grid point $r$ in Monte Carlo run $m$ by $\hat{\theta}_{r,m}$. We calculate the $L_1$ norm 

\begin{equation} \label{eq:l1norm}
L_1 = \frac{1}{M}\sum_{m=1}^M  \frac{1}{R} \sum_{r=1}^R \left\vert\theta_r - \hat{\theta}_{r,m}\right\vert 
\end{equation}

to measure the average absolute bias of $\hat{\theta}$ in comparison to the true weights $\theta$ over all Monte Carlo runs $M$. In addition, we adopt the root mean integrated squared error (RMISE) from \citeA{fox2011} to provide a metric on the approximation accuracy of the estimated distribution. The RMISE averages the squared difference between the true and estimated distribution at a fixed set of grid points across all Monte Carlo runs 

\begin{equation} \label{eq:RMISE}
\text{RMISE}=\sqrt{\frac{1}{M} \sum_{m=1}^M \left[ \frac{1}{E} \sum_{e=1}^E \left(\widehat{F}_m(\beta_e)-F_0(\beta_e)\right)^2\right]} \,\, ,
\end{equation}

where $\widehat{F}_m(\beta_e)$ denotes the estimated distribution function in Monte Carlo run $m$ evaluated at grid point $\beta_e$. For the evaluation, we use $E = 10,000$ points uniformly distributed over the range $[-4.5, 3.5] \times [-4.5, 3.5]$.  \\

Table \ref{tab:MCdiscrete} summarizes the results of the Monte Carlo experiment. The first three columns report the sample size $N$, the number of grid points $R$, and the number of true support points $S$. The upper part of the table presents the measures on the accuracy of the estimated weights, and the lower part the shares of positive, true positive, and sign consistent estimated weights. The final column in the upper part reports the third quantile of the absolute values of the correlation $\rho$ among grid points.\footnote{In addition, we also considered the mean and median to summarize the absolute correlation among grid points. We focus on the third quantile since it best illustrates the strong correlation in this setup.} 

\bigskip

\begin{table}[h]
\centering
\caption{Summary Statistics of 200 Monte Carlo Runs with Discrete Distribution.}  \label{tab:MCdiscrete}
\tabcolsep=0.1cm
\begin{tabular}{lllccccccccc} 
\toprule \toprule\noalign{\smallskip} 
		
& & & \multicolumn{3}{c}{RMISE} & \multicolumn{3}{c}{$L_1$} & \multicolumn{2}{c}{$\mu$} & \multicolumn{1}{c}{$\rho$} \\ 
\cmidrule(l{0.5pt}r{9pt}){4-6}\cmidrule(l{0.5pt}r{9pt}){7-9} \cmidrule(l{0.5pt}r{9pt}){10-11}\cmidrule(l{0.5pt}r{9pt}){12-12} 
$N$ & $R$ & $S$ & FKRB & MSE & OneSe & FKRB & MSE & OneSe & MSE & OneSe & 3rd Qu.  \\ 
		
\hline \\ [-2ex]

1000 & 25 & 17 & 0.067 & 0.04 & 0.034 & 0.035 & 0.017 & 0.014 & 56.05 & 67.95 & 0.808 \\ [0.3ex]    
1000 & 81 & 49 & 0.08 & 0.046 & 0.038 & 0.019 & 0.008 & 0.007 & 58.90 & 70.06 & 0.819 \\ [0.3ex]      
1000 & 289 & 161 & 0.088 & 0.057 & 0.045 & 0.006 & 0.004 & 0.003 & 54.87 & 71.20 & 0.822 \\ [0.8ex] 
10000 & 25 & 17 & 0.042 & 0.026 & 0.023 & 0.02 & 0.012 & 0.011 & 61.15 & 66.78 & 0.809 \\ [0.3ex]    
10000 & 81 & 49 & 0.05 & 0.031 & 0.027 & 0.015 & 0.008 & 0.007 & 59.31 & 69.16 & 0.818 \\ [0.3ex]    
10000 & 289 & 161 & 0.057 & 0.037 & 0.033 & 0.006 & 0.004 & 0.003 & 61.90 & 70.50 & 0.822 \\ [0.3ex]
		
\hline \\[-0.5ex] 
		
& & & \multicolumn{3}{c}{Pos.} & \multicolumn{3}{c}{$\%$ True Pos.} & \multicolumn{3}{c}{$\%$ Sign}   \\ 
\cmidrule(l{0.5pt}r{9pt}){4-6}\cmidrule(l{0.5pt}r{9pt}){7-9} \cmidrule(l{0.5pt}r{9pt}){10-12}  
$N$ & $R$ & $S$ & FKRB & MSE & OneSe & FKRB & MSE & OneSe & FKRB & MSE & OneSe  \\ 
		
\hline \\[-2ex] 
		
1000 & 25 & 17 & 13.3 & 20.82 & 22.36 & 68.44 & 94.53 & 99.71 & 71.88 & 77.28 & 78.14 \\[0.3ex]      
1000 & 81 & 49 & 15.47 & 49.58 & 54.67 & 27.02 & 82.12 & 90.4 & 53.1 & 77.65 & 81.38 \\[0.3ex]       
1000 & 289 & 161 & 16.24 & 103.13 & 123.8 & 8.62 & 55.31 & 66.39 & 48.27 & 70.24 & 75.42 \\[0.8ex]   
10000 & 25 & 17 & 17.17 & 19.39 & 19.73 & 90.32 & 98.12 & 99.53 & 86.16 & 87.86 & 88.46 \\[0.3ex]    
10000 & 81 & 49 & 23.32 & 44.84 & 48.26 & 42.29 & 81.07 & 87.14 & 61.88 & 82.24 & 85.36 \\[0.3ex]    
10000 & 289 & 161 & 24.88 & 97.39 & 105.84 & 13.53 & 55.07 & 59.94 & 50.76 & 71.96 & 74.46 \\[0.3ex] 
		
\bottomrule\bottomrule  \\[-1.8ex] 
\end{tabular}
\captionsetup{justification=justified,singlelinecheck=true, width=0.95\textwidth, font={footnotesize,stretch=0.8}}
\caption*{\textit{Note:} The table reports the average summary statistics over all Monte Carlo replicates for the FKRB estimator (FKRB), and for our generalized estimator with tuning parameter $\mu$ from a 10-fold cross-validation and the $MSE$ criterion (MSE) and the one-standard-error rule (OneSe).}
\end{table}



The results show that our generalized estimator outperforms the FKRB estimator for every combination of $N$ and $R$, in particular when the tuning parameter $\mu$ is chosen based on the one-standard-error rule. 
With respect to the selection consistency, the generalized estimator recovers more true positive and sign consistent probability weights from the data than the FKRB estimator. While the decrease in these shares is moderate for the generalized estimator when the discrete distribution becomes more complex, the correct recovery through the FKRB estimator becomes significantly worse.

This is best illustrated by the small number of positive weights, which changes only slightly alongside the increasing complexity. In the extreme case of $R = 289$, the FKRB estimator estimates positive weights at no more than $16/25$ of the grid points for $N = 1,000/10,000$ (in comparison to $124/106$ for the generalized estimator).

In addition to its improved selection consistency, all measures on the estimated weights indicate that our generalized version provides substantially more accurate estimates of the probability weights than the FKRB estimator. The bias reduction persists for small and large sample sizes. 

\bigskip

\captionsetup{position=top}
\begin{figure}[h]
	\centering
	\caption{Correlation Matrix for $N=10,000$ and $R = 81$}
	\subfigure{\includegraphics[width=0.55\textwidth]{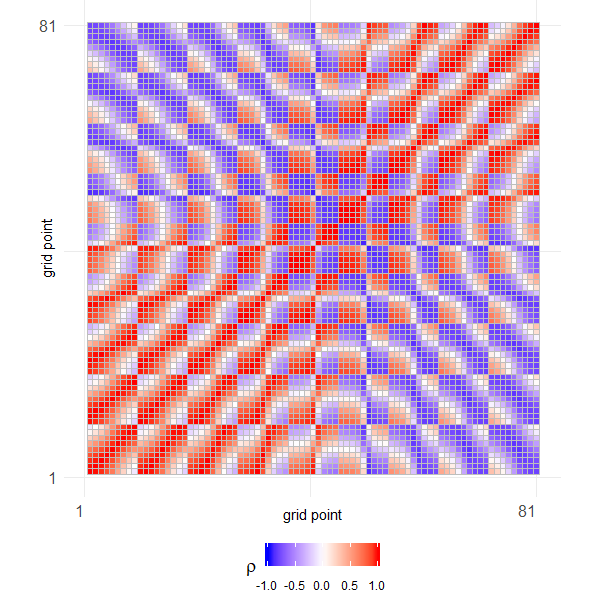}}
	\label{fig:CorQuant}
\end{figure}

The plot of the correlation matrix in Figure \ref{fig:CorQuant} and the third quantile of the values of absolute correlation in Table \ref{tab:MCdiscrete} both illustrate that correlation among many grid points is strong.

\bigskip

\subsection{Continuous Distribution }

The second Monte Carlo experiment considers a mixture of two bivariate normal distributions for $F_0(\beta)$  to analyze how our generalized estimator accommodates more complex continuous distributions. This way, we can assess its ability to recover distributions that cannot be estimated with parametric techniques. 

For the estimation, we use a fixed grid with points spread on $[-4.5, 3.5] \times [-4.5, 3.5]$. The fixed grid covers the support of the true distribution with probability close to one (0.993). We keep the correlation among grid points as low as possible and generate the grid points with a Halton sequence. To study the convergence of the estimated distribution to $F_0(\beta)$ for an increasing number of grid points, we estimate the model with $R = \{25, 50, 100, 250\}$. The number of observation units $N$ varies from 1,000 to 10,000.\\

The variance-covariance matrices of the two normals are $\Sigma_1 = \Sigma_2 = \big[\begin{smallmatrix} 0.8& 0.15\\ 0.15 & 0.8 \end{smallmatrix}\big]$. 
We generate the random coefficient vectors $\beta$ from the following two-component bivariate mixture

\begin{equation}
0.5 \ \mathcal{N}\bigg([-2.2, -2.2], \Sigma_1\bigg) + 0.5 \ \mathcal{N}\bigg([ 1.3, 1.3], \Sigma_2\bigg) \notag
\end{equation}

The left panel in Figure \ref{fig:MCcontinuous} displays the bimodal joint density of the mixture of the two normals, and the right panel the joint distribution function. 

\captionsetup{position=top}
\begin{figure}[H]%
	\centering
	\caption{True Density and Distribution Function of Mixture of two Normals}
	\subfigure[PDF]{\includegraphics[width=0.495\textwidth]{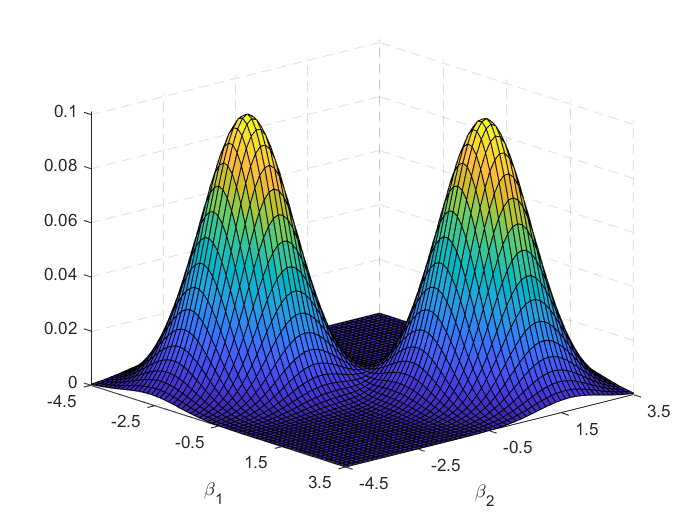}} 
	\subfigure[CDF]{\includegraphics[width=0.495\textwidth]{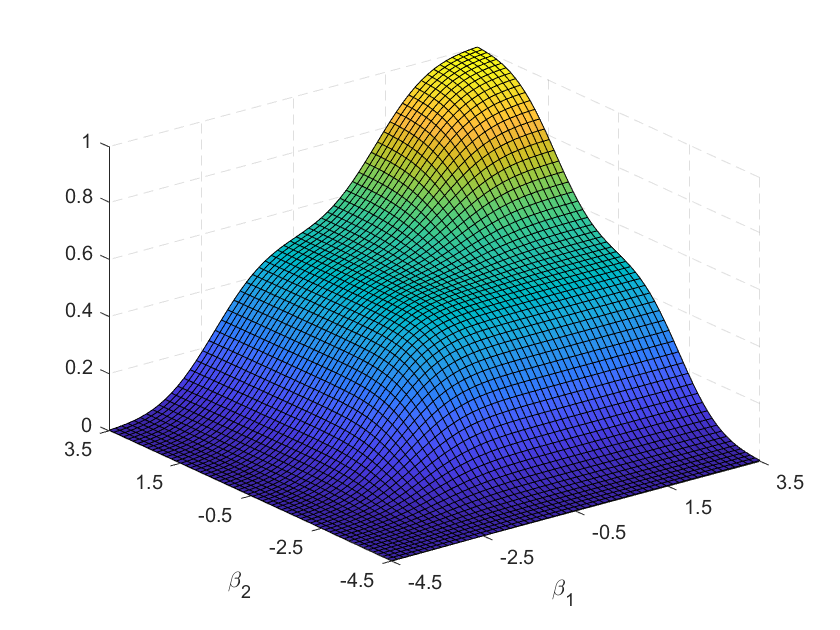}} 
	\label{fig:MCcontinuous}
\end{figure}

%

For the calculation of the RMISE, we use $E = 10,000$ evaluation points uniformly distributed over the range of the fixed grid.
In addition, we report the average number of positive, true positive, and sign consistent estimated weights. For the number of true positive and sign consistent weights, we calculate the true density at every grid point and define a true weight as positive if the density is greater $10^{-3}$.\\

Table \ref{tab:MCcontinuous} summarizes the average results over the $M = 200$ Monte Carlo replicates for the FKRB estimator and our generalized estimator when $\mu$ is chosen with 10-fold cross-validation and the MSE and one-standard error rule, respectively. 
Results for the prediction accuracy of the predicted choices and the log-likelihood as criteria are reported in Appendix \hyperref[tab:MCcontinuouslong]{A}.

\bigskip


\begin{table}[H]
	\caption{Summary Statistics of 200 Monte Carlo Runs with Mixture of Two Bivariate Normals.}  \label{tab:MCcontinuous}
	\centering
	\tabcolsep=0.1cm
	\begin{tabular}{lllccccccccc} 
		\toprule\toprule \noalign{\smallskip} 
		
& & & \multicolumn{3}{c}{RMISE} & \multicolumn{3}{c}{Pos.} & \multicolumn{2}{c}{$\mu$} & \multicolumn{1}{c}{$\rho$} \\ 
\cmidrule(l{0.5pt}r{9pt}){4-6} \cmidrule(l{0.5pt}r{9pt}){7-9}\cmidrule(l{0.5pt}r{9pt}){10-11}\cmidrule(l{0.5pt}r{9pt}){12-12} 
$N$ & $R$ & $S$ & FKRB & MSE & OneSe & FKRB & MSE & OneSe & MSE & OneSe & 3rd Qu.    \\[0.2ex]

\hline \\ [-2ex]

1000 & 25 & 17 & 0.085 & 0.072 & 0.055 & 9.65 & 12.94 & 17.78 & 21.2 & 74.01 & 0.823 \\ [0.3ex]     
1000 & 50 & 33 & 0.09 & 0.068 & 0.058 & 12.57 & 26.82 & 32.4 & 48.09 & 74 & 0.82 \\ [0.3ex]         
1000 & 100 & 67 & 0.095 & 0.07 & 0.061 & 13.65 & 47.09 & 54.85 & 58.6 & 74.37 & 0.822 \\ [0.3ex]    
1000 & 250 & 163 & 0.102 & 0.077 & 0.063 & 14.2 & 79.28 & 103.98 & 50.5 & 74.52 & 0.824 \\ [0.8ex]  
10000 & 25 & 17 & 0.063 & 0.061 & 0.057 & 11.65 & 12.57 & 14.89 & 18.3 & 73.94 & 0.823 \\ [0.3ex]   
10000 & 50 & 33 & 0.058 & 0.051 & 0.047 & 17.52 & 24.94 & 28.3 & 48.71 & 73.94 & 0.82 \\ [0.3ex]    
10000 & 100 & 67 & 0.06 & 0.048 & 0.043 & 19.87 & 39.59 & 46.7 & 51.63 & 74.04 & 0.823 \\ [0.3ex]   
10000 & 250 & 163 & 0.063 & 0.045 & 0.04 & 21.21 & 76.43 & 87.92 & 59.98 & 74.68 & 0.824 \\ [0.3ex]

\hline \\[-0.5ex] 

& & & \multicolumn{3}{c}{$\%$ True Pos.} & \multicolumn{3}{c}{$\%$ Sign} & & &  \\ 
\cmidrule(l{0.5pt}r{9pt}){4-6}\cmidrule(l{0.5pt}r{9pt}){7-9}   
$N$ & $R$ & $S$ & FKRB & MSE & OneSe & FKRB & MSE & OneSe & & &   \\[0.2ex]

\hline \\ [-2ex]

1000 & 25 & 17 & 49.06 & 66.35 & 88.68 & 60.12 & 70.48 & 81.48 & & & \\ [0.3ex]   
1000 & 50 & 33 & 33.39 & 70.33 & 84.48 & 52.93 & 73.19 & 80.72 & & &\\ [0.3ex]   
1000 & 100 & 67 & 18.01 & 63.46 & 74.1 & 43.48 & 70.95 & 77.44 & & &\\ [0.3ex]   
1000 & 250 & 163 & 7.37 & 44.38 & 58.17 & 38.73 & 60.96 & 69.06& & & \\ [0.8ex]  
10000 & 25 & 17 & 57.94 & 63.35 & 77.24 & 64.2 & 67.86 & 77.48 & & &\\ [0.3ex]   
10000 & 50 & 33 & 47.24 & 68.59 & 78.05 & 61.32 & 74.66 & 80.42& & & \\ [0.3ex]  
10000 & 100 & 67 & 26.57 & 54.72 & 64.84 & 48.74 & 66.73 & 73.19& & & \\ [0.3ex] 
10000 & 250 & 163 & 11.39 & 43.78 & 50.56 & 41.17 & 61.32 & 65.56 & & &\\ [0.3ex]

\bottomrule\bottomrule 
		\\[-1.8ex] 
		
	\end{tabular}
	\captionsetup{justification=justified,singlelinecheck=true, width=0.93\textwidth, font={footnotesize,stretch=0.8}}
	\caption*{\textit{Note:} The table reports the average summary statistics over all Monte Carlo replicates for the FKRB estimator (FKRB), and for our generalized estimator with tuning parameter $\mu$ from a 10-fold cross-validation and the $MSE$ criterion (MSE) and the one-standard-error rule (OneSe).}
\end{table}

The RMISE  shows that our generalized estimator provides more accurate estimates of the true underlying random coefficients' distribution than the FKRB estimator for every combination of $N$ and $R$. For $N = 10,000$ the generalized version becomes more accurate with increasing number of grid points and approximates $F_0(\beta)$ quite well for $R = 250$. However, the FKRB estimator does not result in a lower RMISE for $N=10,000$ when $R$ increases.

The improved performance of our estimator for every combination of $N$ and $R$ can be explained with the larger number of true positive and sign consistent estimated probability weights. Independently of the number of (relevant) grid points, the FKRB estimator estimates only a small number of positive weights and, hence, recovers only few relevant grid points. The share of true positive and sign consistent estimated weights is substantially higher for our estimator.  
Figure \ref{fig:MCEstCDFs} plots an example of the joint distribution functions estimated with the FKRB estimator (Panel (a)) and our generalized estimator (Panel (b)). Figure \ref{fig:MCMarginalCDFs} shows the corresponding estimated and true marginal distributions of $\beta_1$ and $\beta_2$. The distribution functions are estimated for $N = 10,000$ and $R = 250$. 

\captionsetup{position=top}
\begin{figure}[ht]%
	\centering
	\caption{Estimated Joint Distribution Functions for $N = 10,000$ and $R = 250$}
	\subfigure[FKRB]{\includegraphics[width=0.495\textwidth]{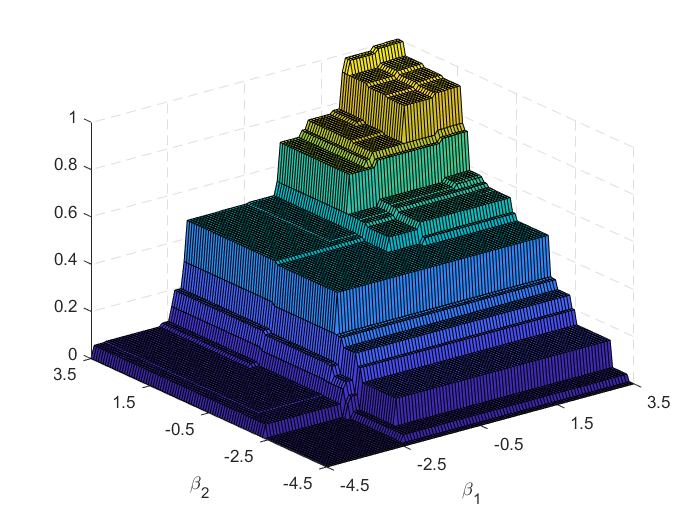}} 
	\subfigure[Generalized with OneSe]{\includegraphics[width=0.495\textwidth]{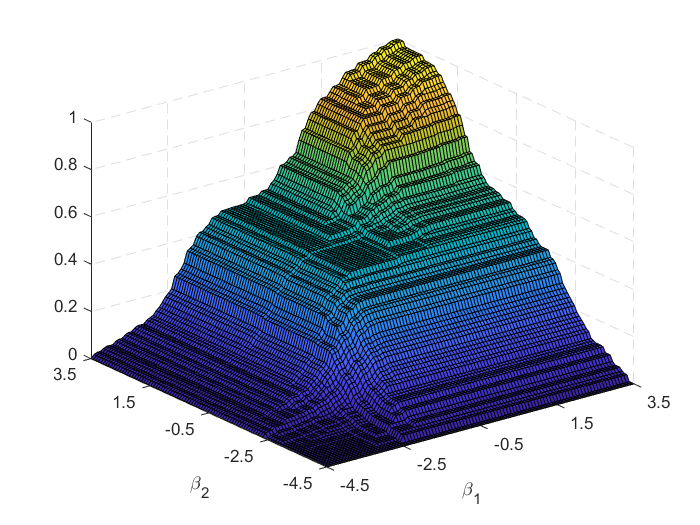}} 
	\label{fig:MCEstCDFs}
\end{figure}

\captionsetup{position=top}
\begin{figure}[ht]%
	\centering
	\caption{True and Estimated Marginal Distribution Functions for $N = 10,000$ and $R = 250$}
	\subfigure{\includegraphics[width=0.495\textwidth]{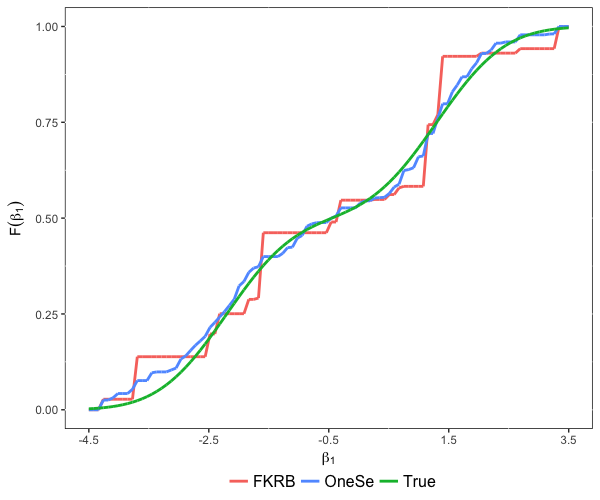}} 
	\subfigure{\includegraphics[width=0.495\textwidth]{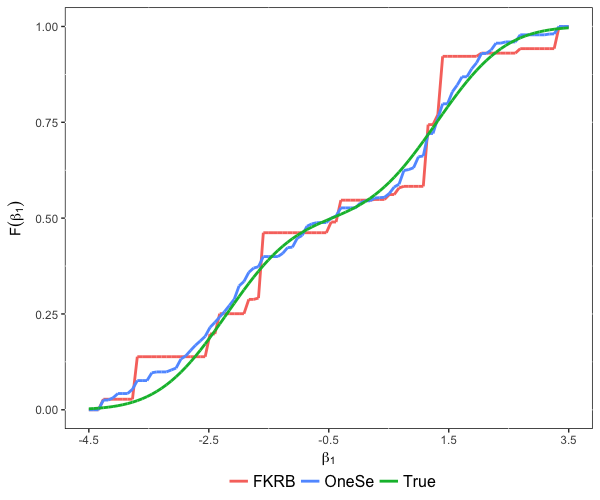}} 
	\label{fig:MCMarginalCDFs}
\end{figure}

The plots illustrate the impact of the FKRB estimator's sparse nature on the estimated marginal and joint distribution functions. Visual inspection shows that it approximates $F_0(\beta)$ through a step function with only few steps due to the small number of positive weights. In contrast, our generalized estimator provides a smooth estimate that is close to the true underlying distribution function.

\section{Application} \label{sec:application}

To study the performance of our generalized estimator with real data, we apply it to the \textit{ModeCanda} data set from the \textit{R package mlogit}. Originally, the Canadian National Rail Carrier VIA Rail assembled the data in 1989 to analyze the demand for future intercity travel in the Toronto-Montr\'eal corridor. 
The data contains information on travelers who can choose among the four intercity travel mode options car, bus, train, and air. Due to the small number of bus users (18), we follow \citeA{bhat1997} and drop  bus as an alternative. Furthermore, we only consider travelers in our analysis that can choose among all three options. Thus, the analyzed data consists of $3,593$ business travelers who can choose among airplane, train, and car. In addition to the observed choices, the data includes information on traveler's income, the trip distance, the frequency of the service, total travel cost, an indicator that is one if either the city of arrival or departure is a big city and zero otherwise, and the in- and out-of-vehicle travel time. We construct the travel time variable by summing up in-vehicle travel time and out-of-vehicle time. This is done for two reasons: first, the data on out-of-vehicle time is always zero for car users and would therefore only capture the preferences of airplane and train users. Second, we think it is plausible that individuals care more about total travel time than the travel time inside and outside of a vehicle separately.

A detailed description of the data can be found in \citeA{KPMG1990}.
Among others, the data set has been studied by \citeA{bhat1995, bhat1997covariance, bhat1997, bhat1998}, \citeA{koppelman2000}, \citeA{wen2001}. The only paper that analyzes the data with a random coefficients logit  model is the study by \citeA{hess2005}. However, they only use the explanatory variables as input for a Monte Carlo study and simulate travelers' mode choices. 

We estimate a mixed logit model with a random coefficient on the travel time and fixed coefficient on all other variables to study the preferred travel mode of business travelers. We include all the above variables into the utility specification along with mode specific constants, where we specify car as the reference alternative. To apply the fixed grid approach to a model with fixed and random coefficients, we follow the recommendation of \citeA{fox2016} and \citeA{houde2019} who suggest a two-step estimator to estimate the model with fixed and random coefficients.\footnote{We also provide an algorithm to update both the fixed and random coefficients in Appendix \ref{app:UpdatingfixedCoefficients}. The algorithm is a modification of the flexible grid estimator in \citeA{train2008}. Unfortunately, the algorithm seems to be very slow and we do not include its results in our comparison here.} In the first step, all coefficients are estimated using a parametric mixed logit.
We assume that the random coefficient is normally distributed. In the second step, the fixed variables and their estimated coefficients from the first stage are treated as data and only the random coefficient of travel time is estimated with the FKRB and elastic net estimator. 
\citeA{houde2019} justify the procedure with the argument that a mixed logit can recover the means of a distribution fairly well despite the 
incorrect assumptions on the random coefficients' distribution. Thus, the fixed coefficients can be estimated consistently with the parametric approach. 
They illustrate this property in a Monte Carlo study.
%

We center the grid of the random coefficient around the mean estimate of the travel coefficient from the first step\footnote{The estimated coefficients of the first stage are provided in Appendix \hyperref[tables:FirstStage]{A}.} and add three standard deviations to each side. We estimate the second step with different numbers of grid points. The preferred specification uses $R=100$ uniformly spread points on the range $[-0.061, 0.027]$. We choose the tuning parameter with 10-fold Cross-Validation and the one standard error rule as criterion. Figure \ref{fig:ApplicationCanada} summarizes the mass and the distribution functions estimated with the FKRB and the ridge estimator.

The elastic net estimator results in a smooth mass function whereas the FKRB exhibits the LASSO behavior. The FKRB estimator only selects five out of $100$ grid points whereas the elastic net estimator selects $75$ grid points.\footnote{We again define a weight as positive if it is greater than $10^{-3}$.} Furthermore, it can easily be seen that the estimated mass function obtained by the elastic net estimator does not seem to be normally distributed but rather looks like a mixture of two normal distributions. That is, specifying a normal or any other parametric distribution function does not seem appropriate in this example. A quite unexpected result is that there are positive weights at positive grid points implying that some people appreciate longer trips. Even though, one might argue that this might be the case if such travelers accept additional travel time for, say, additional comfort when traveling, this might also be a sign of a misspecified model. For the FKRB estimator these weights sum up $9.5 \%$ and for the elastic net to $10.1 \%$ which is lower than $12.6 \%$ for the mixed logit with normal distribution.
The weighted mean of the coefficient of travel time for the  FKRB estimator is $-0.01593$ and  $-0.01631$ for the elastic net estimator. This is roughly the same as $-0.01682$, the mean coefficient obtained from the mixed logit model with normally distributed travel time coefficient.

In addition to the estimated distributions, we report the mean (and median) over  individuals' own- and cross-travel time elasticities for the FKRB estimator, the elastic net estimator and the semiparametric mixed logit with normal distribution in Appendix \hyperref[tab:elasticities]{A}. We also calculate the ratio between elasticities estimated with the FKRB estimator and the semiparametric estimator in comparison to the elasticities estimated with the elastic net estimator. The ratios show that the estimated elasticities are up to $1.8$ times larger for the FKRB estimator and up to $4.5$ times larger for the semiparametric estimator.   

\captionsetup{position=top}
\begin{figure}[H]%
	\centering
	\caption{ Estimated Distributions of Travel Time in Mode Canada Data with $R=100$}
	\subfigure[Mass Function for FKRB]{\includegraphics[width=0.7\textwidth]{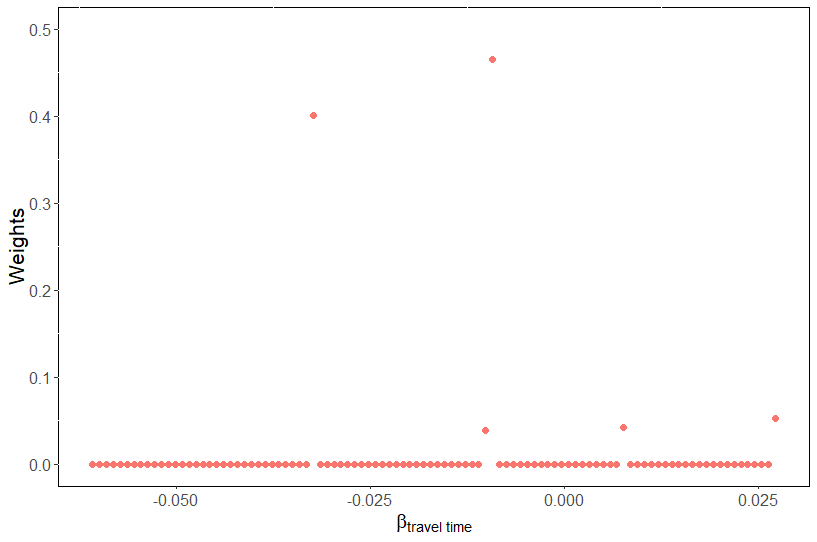}}  \quad	\subfigure[Mass Function for Elastic Net]{\includegraphics[,width=0.7\textwidth]{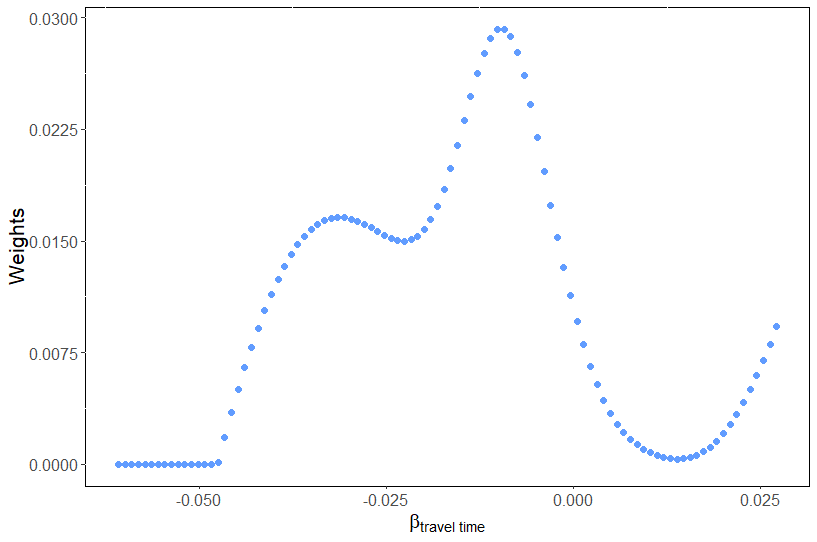}} 
	\subfigure[CDFs for FKRB (red) and Elastic Net (blue)]{\includegraphics[width=0.7\textwidth]{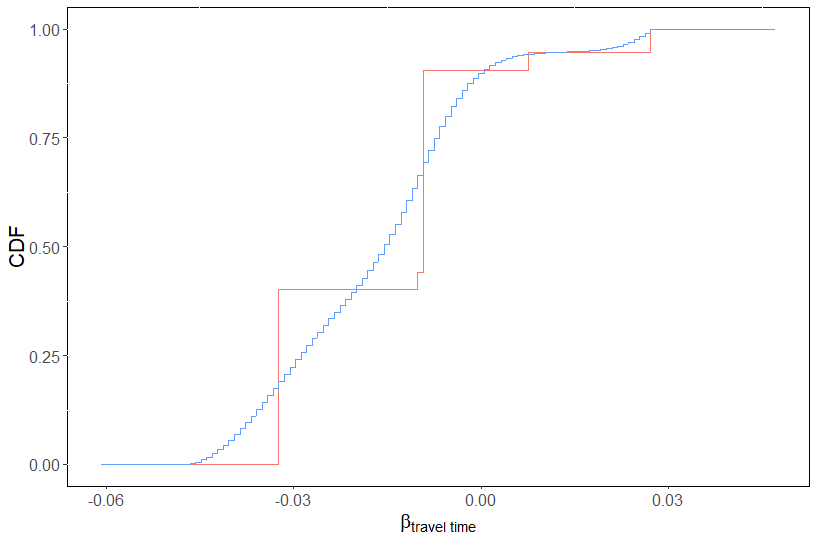}}%
	
	\label{fig:ApplicationCanada}
\end{figure}

\section{Conclusion}\label{sec:conclusion}

We extend the simple and computationally attractive nonparametric estimator of \citeA{fox2011}.  
We illustrate that their estimator is a special case of NNL, explaining its sparse solutions. The connection to NNL reveals that the estimator tends to randomly select among highly correlated grid points. This behavior gives reason to doubt the precise estimation of the true distribution through the estimator. 

To mitigate its undesirable sparsity and random selection behavior, we add a quadratic constraint on the probability weights to the optimization problem of the FKRB estimator. This simple and straightforward extension transforms the estimator to a special case of  nonnegative elastic net. The combination of the linear and quadratic constraint on the probability weights enables a more reliable selection of the relevant grid points. As a consequence, our generalized estimator provides more accurate estimates of the true underlying random coefficients' distribution without increasing computational speed and simplicity substantially.
We derive conditions for selection consistency and an error bound on the estimated distribution function to verify the improved properties of our estimator.

Two Monte Carlo studies illustrate the attractive theoretical properties of our estimator. They show that our generalized version estimates considerably more positive probability weights and recovers more grid points correctly. In addition to the improved selection consistency, the estimator provides more accurate estimates of the true underlying distributions.

Applying the FKRB and the elastic net estimator to a data set of travel choices made in the Toronto-Montr\'eal corridor confirms the sparsity of the FKRB estimator. In contrast, the elastic net estimator selects substantially more grid points, resulting in a smooth distribution function. This illustrates the fact that the elastic net estimator is able to approximate continuous distribution functions.


\newpage


\renewcommand\refname{References}

\bibliography{Literatur}

\begin{thebibliography}{}

\bibitem [\protect \citeauthoryear {%
Bhat%
}{%
Bhat%
}{%
{\protect \APACyear {1995}}%
}]{%
bhat1995}
\APACinsertmetastar {%
bhat1995}%
\begin{APACrefauthors}%
Bhat, C\BPBI R.%
\end{APACrefauthors}%
\unskip\
\newblock
\APACrefYearMonthDay{1995}{}{}.
\newblock
{\BBOQ}\APACrefatitle {A heteroscedastic extreme value model of intercity
  travel mode choice} {A heteroscedastic extreme value model of intercity
  travel mode choice}.{\BBCQ}
\newblock
\APACjournalVolNumPages{Transportation Research Part B:
  Methodological}{29}{6}{471--483}.
\PrintBackRefs{\CurrentBib}

\bibitem [\protect \citeauthoryear {%
Bhat%
}{%
Bhat%
}{%
{\protect \APACyear {1997}}%
{\protect \APACexlab {{\protect \BCnt {1}}}}}]{%
bhat1997covariance}
\APACinsertmetastar {%
bhat1997covariance}%
\begin{APACrefauthors}%
Bhat, C\BPBI R.%
\end{APACrefauthors}%
\unskip\
\newblock
\APACrefYearMonthDay{1997{\protect \BCnt {1}}}{}{}.
\newblock
{\BBOQ}\APACrefatitle {Covariance heterogeneity in nested logit models:
  econometric structure and application to intercity travel} {Covariance
  heterogeneity in nested logit models: econometric structure and application
  to intercity travel}.{\BBCQ}
\newblock
\APACjournalVolNumPages{Transportation Research Part B:
  Methodological}{31}{1}{11--21}.
\PrintBackRefs{\CurrentBib}

\bibitem [\protect \citeauthoryear {%
Bhat%
}{%
Bhat%
}{%
{\protect \APACyear {1997}}%
{\protect \APACexlab {{\protect \BCnt {2}}}}}]{%
bhat1997}
\APACinsertmetastar {%
bhat1997}%
\begin{APACrefauthors}%
Bhat, C\BPBI R.%
\end{APACrefauthors}%
\unskip\
\newblock
\APACrefYearMonthDay{1997{\protect \BCnt {2}}}{}{}.
\newblock
{\BBOQ}\APACrefatitle {An endogenous segmentation mode choice model with an
  application to intercity travel} {An endogenous segmentation mode choice
  model with an application to intercity travel}.{\BBCQ}
\newblock
\APACjournalVolNumPages{Transportation science}{31}{1}{34--48}.
\PrintBackRefs{\CurrentBib}

\bibitem [\protect \citeauthoryear {%
Bhat%
}{%
Bhat%
}{%
{\protect \APACyear {1998}}%
}]{%
bhat1998}
\APACinsertmetastar {%
bhat1998}%
\begin{APACrefauthors}%
Bhat, C\BPBI R.%
\end{APACrefauthors}%
\unskip\
\newblock
\APACrefYearMonthDay{1998}{}{}.
\newblock
{\BBOQ}\APACrefatitle {Accommodating variations in responsiveness to
  level-of-service measures in travel mode choice modeling} {Accommodating
  variations in responsiveness to level-of-service measures in travel mode
  choice modeling}.{\BBCQ}
\newblock
\APACjournalVolNumPages{Transportation Research Part A: Policy and
  Practice}{32}{7}{495--507}.
\PrintBackRefs{\CurrentBib}

\bibitem [\protect \citeauthoryear {%
Blundell%
, Gowrisankaran%
\BCBL {}\ \BBA {} Langer%
}{%
Blundell%
\ \protect \BOthers {.}}{%
{\protect \APACyear {2018}}%
}]{%
blundell2018}
\APACinsertmetastar {%
blundell2018}%
\begin{APACrefauthors}%
Blundell, W.%
, Gowrisankaran, G.%
\BCBL {}\ \BBA {} Langer, A.%
\end{APACrefauthors}%
\unskip\
\newblock
\APACrefYearMonthDay{2018}{}{}.
\newblock
\APACrefbtitle {Escalation of scrutiny: The gains from dynamic enforcement of
  environmental regulations} {Escalation of scrutiny: The gains from dynamic
  enforcement of environmental regulations}\ \APACbVolEdTR{}{\BTR{}}.
\newblock
\APACaddressInstitution{}{National Bureau of Economic Research}.
\PrintBackRefs{\CurrentBib}

\bibitem [\protect \citeauthoryear {%
Burda%
, Harding%
\BCBL {}\ \BBA {} Hausman%
}{%
Burda%
\ \protect \BOthers {.}}{%
{\protect \APACyear {2008}}%
}]{%
burda2008}
\APACinsertmetastar {%
burda2008}%
\begin{APACrefauthors}%
Burda, M.%
, Harding, M.%
\BCBL {}\ \BBA {} Hausman, J.%
\end{APACrefauthors}%
\unskip\
\newblock
\APACrefYearMonthDay{2008}{}{}.
\newblock
{\BBOQ}\APACrefatitle {A Bayesian mixed logit–probit model for multinomial
  choice} {A bayesian mixed logit–probit model for multinomial
  choice}.{\BBCQ}
\newblock
\APACjournalVolNumPages{Journal of Econometrics}{147}{2}{232--246}.
\PrintBackRefs{\CurrentBib}

\bibitem [\protect \citeauthoryear {%
Croissant%
}{%
Croissant%
}{%
{\protect \APACyear {2019}}%
}]{%
mlogit}
\APACinsertmetastar {%
mlogit}%
\begin{APACrefauthors}%
Croissant, Y.%
\end{APACrefauthors}%
\unskip\
\newblock
\APACrefYearMonthDay{2019}{}{}.
\newblock
{\BBOQ}\APACrefatitle {mlogit: Multinomial Logit Models} {mlogit: Multinomial
  logit models}{\BBCQ}\ [\bibcomputersoftwaremanual].
\newblock
\begin{APACrefURL} \url{https://CRAN.R-project.org/package=mlogit}
  \end{APACrefURL}
\newblock
\APACrefnote{R package version 0.4-1}
\PrintBackRefs{\CurrentBib}

\bibitem [\protect \citeauthoryear {%
Efron%
, Hastie%
, Johnstone%
, Tibshirani%
\BCBL {}\ \protect \BOthers {.}}{%
Efron%
\ \protect \BOthers {.}}{%
{\protect \APACyear {2004}}%
}]{%
efron2004}
\APACinsertmetastar {%
efron2004}%
\begin{APACrefauthors}%
Efron, B.%
, Hastie, T.%
, Johnstone, I.%
, Tibshirani, R.%
\BCBL {}\ \BOthersPeriod {.}\end{APACrefauthors}%
\unskip\
\newblock
\APACrefYearMonthDay{2004}{}{}.
\newblock
{\BBOQ}\APACrefatitle {Least angle regression} {Least angle regression}.{\BBCQ}
\newblock
\APACjournalVolNumPages{The Annals of statistics}{32}{2}{407--499}.
\PrintBackRefs{\CurrentBib}

\bibitem [\protect \citeauthoryear {%
El-Arini%
, Xu%
, Fox%
\BCBL {}\ \BBA {} Guestrin%
}{%
El-Arini%
\ \protect \BOthers {.}}{%
{\protect \APACyear {2013}}%
}]{%
arini2013}
\APACinsertmetastar {%
arini2013}%
\begin{APACrefauthors}%
El-Arini, K.%
, Xu, M.%
, Fox, E\BPBI B.%
\BCBL {}\ \BBA {} Guestrin, C.%
\end{APACrefauthors}%
\unskip\
\newblock
\APACrefYearMonthDay{2013}{}{}.
\newblock
{\BBOQ}\APACrefatitle {Representing Documents Through Their Readers}
  {Representing documents through their readers}.{\BBCQ}
\newblock
\BIn{} \APACrefbtitle {Proceedings of the 19th ACM SIGKDD International
  Conference on Knowledge Discovery and Data Mining} {Proceedings of the 19th
  acm sigkdd international conference on knowledge discovery and data mining}\
  (\BPGS\ 14--22).
\newblock
\APACaddressPublisher{New York, NY, USA}{ACM}.
\PrintBackRefs{\CurrentBib}

\bibitem [\protect \citeauthoryear {%
Fox%
, Kim%
, Ryan%
\BCBL {}\ \BBA {} Bajari%
}{%
Fox%
\ \protect \BOthers {.}}{%
{\protect \APACyear {2011}}%
}]{%
fox2011}
\APACinsertmetastar {%
fox2011}%
\begin{APACrefauthors}%
Fox, J\BPBI T.%
, Kim, K.%
, Ryan, S.%
\BCBL {}\ \BBA {} Bajari, P.%
\end{APACrefauthors}%
\unskip\
\newblock
\APACrefYearMonthDay{2011}{}{}.
\newblock
{\BBOQ}\APACrefatitle {A simple estimator for the distribution of random
  coefficients} {A simple estimator for the distribution of random
  coefficients}.{\BBCQ}
\newblock
\APACjournalVolNumPages{Quantitative Economics}{2}{3}{381--418}.
\PrintBackRefs{\CurrentBib}

\bibitem [\protect \citeauthoryear {%
Fox%
, Kim%
\BCBL {}\ \BBA {} Yang%
}{%
Fox%
\ \protect \BOthers {.}}{%
{\protect \APACyear {2016}}%
}]{%
fox2016}
\APACinsertmetastar {%
fox2016}%
\begin{APACrefauthors}%
Fox, J\BPBI T.%
, Kim, K.%
\BCBL {}\ \BBA {} Yang, C.%
\end{APACrefauthors}%
\unskip\
\newblock
\APACrefYearMonthDay{2016}{}{}.
\newblock
{\BBOQ}\APACrefatitle {A simple nonparametric approach to estimating the
  distribution of random coefficients in structural models} {A simple
  nonparametric approach to estimating the distribution of random coefficients
  in structural models}.{\BBCQ}
\newblock
\APACjournalVolNumPages{Journal of Econometrics}{195}{2}{236--254}.
\PrintBackRefs{\CurrentBib}

\bibitem [\protect \citeauthoryear {%
Gentle%
}{%
Gentle%
}{%
{\protect \APACyear {2007}}%
}]{%
gentle2007}
\APACinsertmetastar {%
gentle2007}%
\begin{APACrefauthors}%
Gentle, J\BPBI E.%
\end{APACrefauthors}%
\unskip\
\newblock
\APACrefYear{2007}.
\newblock
\APACrefbtitle {Matrix Algebra: Theory, Computations, and Applications in
  Statistics} {Matrix algebra: Theory, computations, and applications in
  statistics}\ (\PrintOrdinal{1st}\ \BEd).
\newblock
\APACaddressPublisher{}{Springer Publishing Company, Incorporated}.
\PrintBackRefs{\CurrentBib}

\bibitem [\protect \citeauthoryear {%
Hastie%
, Tibshirani%
\BCBL {}\ \BBA {} Friedman%
}{%
Hastie%
\ \protect \BOthers {.}}{%
{\protect \APACyear {2009}}%
}]{%
hastie09}
\APACinsertmetastar {%
hastie09}%
\begin{APACrefauthors}%
Hastie, T.%
, Tibshirani, R.%
\BCBL {}\ \BBA {} Friedman, J.%
\end{APACrefauthors}%
\unskip\
\newblock
\APACrefYear{2009}.
\newblock
\APACrefbtitle {The elements of statistical learning: data mining, inference
  and prediction} {The elements of statistical learning: data mining, inference
  and prediction}\ (\PrintOrdinal{2}\ \BEd).
\newblock
\APACaddressPublisher{}{Springer}.
\PrintBackRefs{\CurrentBib}

\bibitem [\protect \citeauthoryear {%
Hess%
, Bierlaire%
\BCBL {}\ \BBA {} Polak%
}{%
Hess%
\ \protect \BOthers {.}}{%
{\protect \APACyear {2005}}%
}]{%
hess2005}
\APACinsertmetastar {%
hess2005}%
\begin{APACrefauthors}%
Hess, S.%
, Bierlaire, M.%
\BCBL {}\ \BBA {} Polak, J\BPBI W.%
\end{APACrefauthors}%
\unskip\
\newblock
\APACrefYearMonthDay{2005}{}{}.
\newblock
{\BBOQ}\APACrefatitle {Estimation of value of travel-time savings using mixed
  logit models} {Estimation of value of travel-time savings using mixed logit
  models}.{\BBCQ}
\newblock
\APACjournalVolNumPages{Transportation Research Part A: Policy and
  Practice}{39}{2-3}{221--236}.
\PrintBackRefs{\CurrentBib}

\bibitem [\protect \citeauthoryear {%
Hoerl%
\ \BBA {} Kennard%
}{%
Hoerl%
\ \BBA {} Kennard%
}{%
{\protect \APACyear {1970}}%
}]{%
hoerl1970}
\APACinsertmetastar {%
hoerl1970}%
\begin{APACrefauthors}%
Hoerl, A\BPBI E.%
\BCBT {}\ \BBA {} Kennard, R\BPBI W.%
\end{APACrefauthors}%
\unskip\
\newblock
\APACrefYearMonthDay{1970}{}{}.
\newblock
{\BBOQ}\APACrefatitle {Ridge regression: Biased estimation for nonorthogonal
  problems} {Ridge regression: Biased estimation for nonorthogonal
  problems}.{\BBCQ}
\newblock
\APACjournalVolNumPages{Technometrics}{12}{1}{55--67}.
\PrintBackRefs{\CurrentBib}

\bibitem [\protect \citeauthoryear {%
Houde%
\ \BBA {} Myers%
}{%
Houde%
\ \BBA {} Myers%
}{%
{\protect \APACyear {2019}}%
}]{%
houde2019}
\APACinsertmetastar {%
houde2019}%
\begin{APACrefauthors}%
Houde, S.%
\BCBT {}\ \BBA {} Myers, E.%
\end{APACrefauthors}%
\unskip\
\newblock
\APACrefYearMonthDay{2019}{April}{}.
\newblock
\APACrefbtitle {Heterogeneous (Mis-) Perceptions of Energy Costs: Implications
  for Measurement and Policy Design} {Heterogeneous (mis-) perceptions of
  energy costs: Implications for measurement and policy design}\ \APACbVolEdTR
  {}{Working Paper\ \BNUM\ 25722}.
\newblock
\APACaddressInstitution{}{National Bureau of Economic Research}.
\PrintBackRefs{\CurrentBib}

\bibitem [\protect \citeauthoryear {%
Hu%
, Follmann%
\BCBL {}\ \BBA {} Miura%
}{%
Hu%
\ \protect \BOthers {.}}{%
{\protect \APACyear {2015}}%
}]{%
hu2015}
\APACinsertmetastar {%
hu2015}%
\begin{APACrefauthors}%
Hu, Z.%
, Follmann, D\BPBI A.%
\BCBL {}\ \BBA {} Miura, K.%
\end{APACrefauthors}%
\unskip\
\newblock
\APACrefYearMonthDay{2015}{}{}.
\newblock
{\BBOQ}\APACrefatitle {Vaccine design via nonnegative lasso-based variable
  selection} {Vaccine design via nonnegative lasso-based variable
  selection}.{\BBCQ}
\newblock
\APACjournalVolNumPages{Statistics in medicine}{34}{10}{1791--1798}.
\PrintBackRefs{\CurrentBib}

\bibitem [\protect \citeauthoryear {%
Illanes%
\ \BBA {} Padi%
}{%
Illanes%
\ \BBA {} Padi%
}{%
{\protect \APACyear {2019}}%
}]{%
illanes2018}
\APACinsertmetastar {%
illanes2018}%
\begin{APACrefauthors}%
Illanes, G.%
\BCBT {}\ \BBA {} Padi, M.%
\end{APACrefauthors}%
\unskip\
\newblock
\APACrefYearMonthDay{2019}{}{}.
\newblock
\APACrefbtitle {Competition, Asymmetric Information, and the Annuity Puzzle:
  Evidence from a Government-Run Exchange in Chile} {Competition, asymmetric
  information, and the annuity puzzle: Evidence from a government-run exchange
  in chile}\ \APACbVolEdTR{}{\BTR{}}.
\newblock
\APACaddressInstitution{}{Center for Retirement Research}.
\PrintBackRefs{\CurrentBib}

\bibitem [\protect \citeauthoryear {%
Jia%
\ \BBA {} Yu%
}{%
Jia%
\ \BBA {} Yu%
}{%
{\protect \APACyear {2010}}%
}]{%
jia2010}
\APACinsertmetastar {%
jia2010}%
\begin{APACrefauthors}%
Jia, J.%
\BCBT {}\ \BBA {} Yu, B.%
\end{APACrefauthors}%
\unskip\
\newblock
\APACrefYearMonthDay{2010}{}{}.
\newblock
{\BBOQ}\APACrefatitle {ON MODEL SELECTION CONSISTENCY OF THE ELASTIC NET WHEN p
  $\gg$ n} {On model selection consistency of the elastic net when p $\gg$
  n}.{\BBCQ}
\newblock
\APACjournalVolNumPages{Statistica Sinica}{20}{2}{595--611}.
\PrintBackRefs{\CurrentBib}

\bibitem [\protect \citeauthoryear {%
Koppelman%
\ \BBA {} Wen%
}{%
Koppelman%
\ \BBA {} Wen%
}{%
{\protect \APACyear {2000}}%
}]{%
koppelman2000}
\APACinsertmetastar {%
koppelman2000}%
\begin{APACrefauthors}%
Koppelman, F\BPBI S.%
\BCBT {}\ \BBA {} Wen, C\BHBI H.%
\end{APACrefauthors}%
\unskip\
\newblock
\APACrefYearMonthDay{2000}{}{}.
\newblock
{\BBOQ}\APACrefatitle {The paired combinatorial logit model: properties,
  estimation and application} {The paired combinatorial logit model:
  properties, estimation and application}.{\BBCQ}
\newblock
\APACjournalVolNumPages{Transportation Research Part B:
  Methodological}{34}{2}{75--89}.
\PrintBackRefs{\CurrentBib}

\bibitem [\protect \citeauthoryear {%
Kump%
, Bai%
, sik Chan%
, Eichinger%
\BCBL {}\ \BBA {} Li%
}{%
Kump%
\ \protect \BOthers {.}}{%
{\protect \APACyear {2012}}%
}]{%
kump2012}
\APACinsertmetastar {%
kump2012}%
\begin{APACrefauthors}%
Kump, P.%
, Bai, E\BHBI W.%
, sik Chan, K.%
, Eichinger, B.%
\BCBL {}\ \BBA {} Li, K.%
\end{APACrefauthors}%
\unskip\
\newblock
\APACrefYearMonthDay{2012}{}{}.
\newblock
{\BBOQ}\APACrefatitle {Variable selection via RIVAL (removing irrelevant
  variables amidst Lasso iterations) and its application to nuclear material
  detection} {Variable selection via rival (removing irrelevant variables
  amidst lasso iterations) and its application to nuclear material
  detection}.{\BBCQ}
\newblock
\APACjournalVolNumPages{Automatica}{48}{9}{2107--2115}.
\PrintBackRefs{\CurrentBib}

\bibitem [\protect \citeauthoryear {%
Marwick%
\ \BBA {} Koppelman%
}{%
Marwick%
\ \BBA {} Koppelman%
}{%
{\protect \APACyear {1990}}%
}]{%
KPMG1990}
\APACinsertmetastar {%
KPMG1990}%
\begin{APACrefauthors}%
Marwick, K\BPBI P.%
\BCBT {}\ \BBA {} Koppelman, F.%
\end{APACrefauthors}%
\unskip\
\newblock
\APACrefYearMonthDay{1990}{}{}.
\newblock
{\BBOQ}\APACrefatitle {Proposals for analysis of the market demand for high
  speed rail in the Quebec/Ontario corridor} {Proposals for analysis of the
  market demand for high speed rail in the quebec/ontario corridor}.{\BBCQ}
\newblock
\APACjournalVolNumPages{Submitted to Ontario/Quebec Rapid Train Task
  Force}{}{}{}.
\PrintBackRefs{\CurrentBib}

\bibitem [\protect \citeauthoryear {%
McFadden%
\ \BBA {} Train%
}{%
McFadden%
\ \BBA {} Train%
}{%
{\protect \APACyear {2000}}%
}]{%
mcfadden2000}
\APACinsertmetastar {%
mcfadden2000}%
\begin{APACrefauthors}%
McFadden, D.%
\BCBT {}\ \BBA {} Train, K.%
\end{APACrefauthors}%
\unskip\
\newblock
\APACrefYearMonthDay{2000}{}{}.
\newblock
{\BBOQ}\APACrefatitle {Mixed MNL models for discrete response} {Mixed mnl
  models for discrete response}.{\BBCQ}
\newblock
\APACjournalVolNumPages{Journal of Applied Econometrics}{15}{5}{447--470}.
\PrintBackRefs{\CurrentBib}

\bibitem [\protect \citeauthoryear {%
Nevo%
, Turner%
\BCBL {}\ \BBA {} Williams%
}{%
Nevo%
\ \protect \BOthers {.}}{%
{\protect \APACyear {2016}}%
}]{%
nevo2016}
\APACinsertmetastar {%
nevo2016}%
\begin{APACrefauthors}%
Nevo, A.%
, Turner, J\BPBI L.%
\BCBL {}\ \BBA {} Williams, J\BPBI W.%
\end{APACrefauthors}%
\unskip\
\newblock
\APACrefYearMonthDay{2016}{}{}.
\newblock
{\BBOQ}\APACrefatitle {Usage-Based Pricing and Demand for Residential
  Broadband} {Usage-based pricing and demand for residential broadband}.{\BBCQ}
\newblock
\APACjournalVolNumPages{Econometrica}{84}{2}{411--443}.
\PrintBackRefs{\CurrentBib}

\bibitem [\protect \citeauthoryear {%
{R Core Team}%
}{%
{R Core Team}%
}{%
{\protect \APACyear {2018}}%
}]{%
RCore}
\APACinsertmetastar {%
RCore}%
\begin{APACrefauthors}%
{R Core Team}.%
\end{APACrefauthors}%
\unskip\
\newblock
\APACrefYearMonthDay{2018}{}{}.
\newblock
{\BBOQ}\APACrefatitle {R: A Language and Environment for Statistical Computing}
  {R: A language and environment for statistical computing}{\BBCQ}\
  [\bibcomputersoftwaremanual].
\newblock
\APACaddressPublisher{Vienna, Austria}{}.
\PrintBackRefs{\CurrentBib}

\bibitem [\protect \citeauthoryear {%
Rossi%
, Allenby%
\BCBL {}\ \BBA {} McCulloch%
}{%
Rossi%
\ \protect \BOthers {.}}{%
{\protect \APACyear {2012}}%
}]{%
rossi2005}
\APACinsertmetastar {%
rossi2005}%
\begin{APACrefauthors}%
Rossi, P\BPBI E.%
, Allenby, G\BPBI M.%
\BCBL {}\ \BBA {} McCulloch, R.%
\end{APACrefauthors}%
\unskip\
\newblock
\APACrefYear{2012}.
\newblock
\APACrefbtitle {Bayesian statistics and marketing} {Bayesian statistics and
  marketing}.
\newblock
\APACaddressPublisher{}{John Wiley \& Sons}.
\PrintBackRefs{\CurrentBib}

\bibitem [\protect \citeauthoryear {%
Slawski%
\ \BBA {} Hein%
}{%
Slawski%
\ \BBA {} Hein%
}{%
{\protect \APACyear {2013}}%
}]{%
slawski2013}
\APACinsertmetastar {%
slawski2013}%
\begin{APACrefauthors}%
Slawski, M.%
\BCBT {}\ \BBA {} Hein, M.%
\end{APACrefauthors}%
\unskip\
\newblock
\APACrefYearMonthDay{2013}{}{}.
\newblock
{\BBOQ}\APACrefatitle {Non-negative least squares for high-dimensional linear
  models: Consistency and sparse recovery without regularization} {Non-negative
  least squares for high-dimensional linear models: Consistency and sparse
  recovery without regularization}.{\BBCQ}
\newblock
\APACjournalVolNumPages{Electron. J. Statist.}{7}{}{3004--3056}.
\PrintBackRefs{\CurrentBib}

\bibitem [\protect \citeauthoryear {%
Takada%
, Suzuki%
\BCBL {}\ \BBA {} Fujisawa%
}{%
Takada%
\ \protect \BOthers {.}}{%
{\protect \APACyear {2017}}%
}]{%
takada2017}
\APACinsertmetastar {%
takada2017}%
\begin{APACrefauthors}%
Takada, M.%
, Suzuki, T.%
\BCBL {}\ \BBA {} Fujisawa, H.%
\end{APACrefauthors}%
\unskip\
\newblock
\APACrefYearMonthDay{2017}{}{}.
\newblock
{\BBOQ}\APACrefatitle {Independently Interpretable Lasso: A New Regularizer for
  Sparse Regression with Uncorrelated Variables} {Independently interpretable
  lasso: A new regularizer for sparse regression with uncorrelated
  variables}.{\BBCQ}
\newblock
\APACjournalVolNumPages{arXiv preprint arXiv:1711.01796}{}{}{}.
\PrintBackRefs{\CurrentBib}

\bibitem [\protect \citeauthoryear {%
Tibshirani%
}{%
Tibshirani%
}{%
{\protect \APACyear {1996}}%
}]{%
tibshirani1996}
\APACinsertmetastar {%
tibshirani1996}%
\begin{APACrefauthors}%
Tibshirani, R.%
\end{APACrefauthors}%
\unskip\
\newblock
\APACrefYearMonthDay{1996}{}{}.
\newblock
{\BBOQ}\APACrefatitle {Regression shrinkage and selection via the lasso}
  {Regression shrinkage and selection via the lasso}.{\BBCQ}
\newblock
\APACjournalVolNumPages{Journal of the Royal Statistical Society: Series B
  (Methodological)}{58}{1}{267--288}.
\PrintBackRefs{\CurrentBib}

\bibitem [\protect \citeauthoryear {%
Train%
}{%
Train%
}{%
{\protect \APACyear {2008}}%
}]{%
train2008}
\APACinsertmetastar {%
train2008}%
\begin{APACrefauthors}%
Train, K.%
\end{APACrefauthors}%
\unskip\
\newblock
\APACrefYearMonthDay{2008}{}{}.
\newblock
{\BBOQ}\APACrefatitle {EM Algorithms for nonparametric estimation of mixing
  distributions} {Em algorithms for nonparametric estimation of mixing
  distributions}.{\BBCQ}
\newblock
\APACjournalVolNumPages{Journal of Choice Modelling}{1}{1}{40--69}.
\PrintBackRefs{\CurrentBib}

\bibitem [\protect \citeauthoryear {%
Train%
}{%
Train%
}{%
{\protect \APACyear {2016}}%
}]{%
train2016}
\APACinsertmetastar {%
train2016}%
\begin{APACrefauthors}%
Train, K.%
\end{APACrefauthors}%
\unskip\
\newblock
\APACrefYearMonthDay{2016}{}{}.
\newblock
{\BBOQ}\APACrefatitle {Mixed logit with a flexible mixing distribution} {Mixed
  logit with a flexible mixing distribution}.{\BBCQ}
\newblock
\APACjournalVolNumPages{Journal of Choice Modelling}{19}{}{40--53}.
\PrintBackRefs{\CurrentBib}

\bibitem [\protect \citeauthoryear {%
Wen%
\ \BBA {} Koppelman%
}{%
Wen%
\ \BBA {} Koppelman%
}{%
{\protect \APACyear {2001}}%
}]{%
wen2001}
\APACinsertmetastar {%
wen2001}%
\begin{APACrefauthors}%
Wen, C\BHBI H.%
\BCBT {}\ \BBA {} Koppelman, F\BPBI S.%
\end{APACrefauthors}%
\unskip\
\newblock
\APACrefYearMonthDay{2001}{}{}.
\newblock
{\BBOQ}\APACrefatitle {The generalized nested logit model} {The generalized
  nested logit model}.{\BBCQ}
\newblock
\APACjournalVolNumPages{Transportation Research Part B:
  Methodological}{35}{7}{627--641}.
\PrintBackRefs{\CurrentBib}

\bibitem [\protect \citeauthoryear {%
Wu%
\ \BBA {} Yang%
}{%
Wu%
\ \BBA {} Yang%
}{%
{\protect \APACyear {2014}}%
}]{%
wu2014elastic}
\APACinsertmetastar {%
wu2014elastic}%
\begin{APACrefauthors}%
Wu, L.%
\BCBT {}\ \BBA {} Yang, Y.%
\end{APACrefauthors}%
\unskip\
\newblock
\APACrefYearMonthDay{2014}{}{}.
\newblock
{\BBOQ}\APACrefatitle {Nonnegative Elastic Net and application in index
  tracking} {Nonnegative elastic net and application in index tracking}.{\BBCQ}
\newblock
\APACjournalVolNumPages{Applied Mathematics and Computation}{227}{}{541--552}.
\PrintBackRefs{\CurrentBib}

\bibitem [\protect \citeauthoryear {%
Wu%
, Yang%
\BCBL {}\ \BBA {} Liu%
}{%
Wu%
\ \protect \BOthers {.}}{%
{\protect \APACyear {2014}}%
}]{%
wu2014lasso}
\APACinsertmetastar {%
wu2014lasso}%
\begin{APACrefauthors}%
Wu, L.%
, Yang, Y.%
\BCBL {}\ \BBA {} Liu, H.%
\end{APACrefauthors}%
\unskip\
\newblock
\APACrefYearMonthDay{2014}{}{}.
\newblock
{\BBOQ}\APACrefatitle {Nonnegative-lasso and application in index tracking}
  {Nonnegative-lasso and application in index tracking}.{\BBCQ}
\newblock
\APACjournalVolNumPages{Computational Statistics and Data
  Analysis}{70}{}{116--126}.
\PrintBackRefs{\CurrentBib}

\bibitem [\protect \citeauthoryear {%
Zhao%
\ \BBA {} Yu%
}{%
Zhao%
\ \BBA {} Yu%
}{%
{\protect \APACyear {2006}}%
}]{%
zhao2006}
\APACinsertmetastar {%
zhao2006}%
\begin{APACrefauthors}%
Zhao, P.%
\BCBT {}\ \BBA {} Yu, B.%
\end{APACrefauthors}%
\unskip\
\newblock
\APACrefYearMonthDay{2006}{}{}.
\newblock
{\BBOQ}\APACrefatitle {On model selection consistency of Lasso} {On model
  selection consistency of lasso}.{\BBCQ}
\newblock
\APACjournalVolNumPages{Journal of Machine learning
  research}{7}{Nov}{2541--2563}.
\PrintBackRefs{\CurrentBib}

\bibitem [\protect \citeauthoryear {%
Zou%
\ \BBA {} Hastie%
}{%
Zou%
\ \BBA {} Hastie%
}{%
{\protect \APACyear {2005}}%
}]{%
hastie2005}
\APACinsertmetastar {%
hastie2005}%
\begin{APACrefauthors}%
Zou, H.%
\BCBT {}\ \BBA {} Hastie, T.%
\end{APACrefauthors}%
\unskip\
\newblock
\APACrefYearMonthDay{2005}{}{}.
\newblock
{\BBOQ}\APACrefatitle {Regularization and Variable Selection via the Elastic
  Net} {Regularization and variable selection via the elastic net}.{\BBCQ}
\newblock
\APACjournalVolNumPages{Journal of the Royal Statistical Society. Series B
  (Statistical Methodology)}{67}{2}{301--320}.
\PrintBackRefs{\CurrentBib}

\end{thebibliography}

\begin{appendix}

\onehalfspacing	


\newgeometry{left=1.3cm,right=1.3cm,top=2cm,bottom=2cm}
\afterpage{

\begin{landscape}
	
	
\section*{Appendix}	
	
\section[Appendix A: \\ Supplementary Tables]{Supplementary Tables}	\label{app:tables}

\setcounter{equation}{0}
\setcounter{table}{0}
\numberwithin{equation}{section}
\numberwithin{table}{section}
\numberwithin{subsection}{section}
\renewcommand{\thetable}{A.\arabic{table}}
\renewcommand{\theequation}{A.\arabic{equation}}
\renewcommand{\thesubsection}{A.\arabic{subsection}}
\captionsetup[figure]{list=no}
\captionsetup[table]{list=no}
\captionsetup[subsection]{list=no}

\onehalfspacing

\begin{table}[H]
\centering
\caption{Detailed Summary Statistics of 200 Monte Carlo Runs with Discrete Distribution.}  \label{tab:MCdiscreteLong}
\small
\tabcolsep=0.08cm
\begin{tabular}{lllccccccccccccccc}
\toprule\toprule \noalign{\smallskip} 
& & & \multicolumn{5}{c}{RMISE} & \multicolumn{5}{c}{$L_1$} & \multicolumn{4}{c}{$\mu$} & \multicolumn{1}{c}{$\rho$} \\ 
\cmidrule(l{0.5pt}r{9pt}){4-8}\cmidrule(l{0.5pt}r{9pt}){9-13} \cmidrule(l{0.5pt}r{9pt}){14-17}\cmidrule(l{0.5pt}r{9pt}){18-18}  
$N$ & $R$ & $S$ & FKRB & MSE & OneSe & LL & PredOut & FKRB & MSE & OneSe & LL & PredOut & MSE & OneSe & LL & PredOut & 3rd Qu. \\[0.2ex] 
			
\hline \\ [-2ex] 
			
1000 & 25 & 17 & 0.067 & 0.04 & 0.034 & 0.055 & 0.045 & 0.035 & 0.017 & 0.014 & 0.027 & 0.021 & 56.05 & 67.95 & 13.14 & 35.03 & 0.808 \\[0.3ex]      
1000 & 81 & 49 & 0.08 & 0.046 & 0.038 & 0.064 & 0.055 & 0.019 & 0.008 & 0.007 & 0.014 & 0.011 & 58.90 & 70.06 & 20.66 & 36.96 & 0.819 \\[0.3ex]       
1000 & 289 & 161 & 0.088 & 0.057 & 0.045 & 0.068 & 0.06 & 0.006 & 0.004 & 0.003 & 0.005 & 0.004 & 54.87 & 71.20 & 30.71 & 35.75 & 0.822 \\[0.8ex]   
10000 & 25 & 17 & 0.042 & 0.026 & 0.023 & 0.037 & 0.032 & 0.02 & 0.012 & 0.011 & 0.018 & 0.015 & 61.15 & 66.78 & 13.93 & 31.02 & 0.809 \\[0.3ex]     
10000 & 81 & 49 & 0.05 & 0.031 & 0.027 & 0.044 & 0.037 & 0.015 & 0.008 & 0.007 & 0.013 & 0.011 & 59.31 & 69.16 & 13.85 & 30.90 & 0.818 \\[0.3ex]     
10000 & 289 & 161 & 0.057 & 0.037 & 0.033 & 0.049 & 0.043 & 0.006 & 0.004 & 0.003 & 0.005 & 0.005 & 61.90 & 70.50 & 19.02 & 30.14 & 0.822 \\[0.3ex] 
			
\hline \\[-0.5ex] 
			
& & & \multicolumn{5}{c}{Pos.} & \multicolumn{5}{c}{$\%$ True Pos.} & \multicolumn{5}{c}{$\%$ Sign} \\ 
\cmidrule(l{0.5pt}r{9pt}){4-8}\cmidrule(l{0.5pt}r{9pt}){9-13} \cmidrule(l{0.5pt}r{9pt}){14-18}  
$N$ & $R$ & $S$ & FKRB & MSE & OneSe & LL & PredOut & FKRB & MSE & OneSe & LL & PredOut & FKRB & MSE & OneSe & LL & PredOut   \\[0.2ex] 
			
\hline \\[-2ex] 
			
1000 & 25 & 17 & 13.3 & 20.82 & 22.36 & 16.23 & 19.35 & 68.44 & 94.53 & 99.71 & 80.91 & 91.15 & 71.88 & 77.28 & 78.14 & 77.14 & 78.56 \\[0.5ex]     
1000 & 81 & 49 & 15.47 & 49.58 & 54.67 & 31.07 & 40.88 & 27.02 & 82.12 & 90.4 & 53.58 & 69.17 & 53.1 & 77.65 & 81.38 & 65.97 & 72.72 \\[0.5ex]      
1000 & 289 & 161 & 16.24 & 103.13 & 123.8 & 70.4 & 84.15 & 8.62 & 55.31 & 66.39 & 38.02 & 45.3 & 48.27 & 70.24 & 75.42 & 62.3 & 65.64 \\[0.5ex]     
10000 & 25 & 17 & 17.17 & 19.39 & 19.73 & 17.81 & 18.64 & 90.32 & 98.12 & 99.53 & 92.94 & 96.35 & 86.16 & 87.86 & 88.46 & 87.16 & 88.48 \\[1ex]   
10000 & 81 & 49 & 23.32 & 44.84 & 48.26 & 29.88 & 37.23 & 42.29 & 81.07 & 87.14 & 54.5 & 67.84 & 61.88 & 82.24 & 85.36 & 68.56 & 75.61 \\[0.5ex]    
10000 & 289 & 161 & 24.88 & 97.39 & 105.84 & 53.06 & 69.47 & 13.53 & 55.07 & 59.94 & 29.93 & 39.3 & 50.76 & 71.96 & 74.46 & 59.28 & 64.04 \\[0.5ex] 
			
\bottomrule\bottomrule\\[-1.8ex] 
			
\end{tabular}
\captionsetup{justification=justified,singlelinecheck=true, width=1.20\textwidth, font={footnotesize,stretch=0.8}}
\caption*{\textit{Note:} The table reports the average summary statistics over all Monte Carlo replicates for the FKRB estimator (FKRB), and for our generalized estimator with tuning parameter $\mu$ from a 10-fold cross-validation and the $MSE$ criterion (MSE), the one-standard-error rule (OneSe), the log-likelihood criterion (LL) and the number of correctly predicted binary outcomes (PredOut). The predicted binary outcome is set to one for the alternative with the highest estimated choice probability.}
\end{table}

\end{landscape}

\newpage

\newgeometry{left=1.5cm,right=1.5cm,top=2cm,bottom=2cm}

\begin{landscape}

\begin{table}[H]
\caption{Detailed Summary Statistics of 200 Monte Carlo Runs with Mixture of Two Bivariate Normals.}  \label{tab:MCcontinuouslong}
\centering
\small
\tabcolsep=0.08cm
\begin{tabular}{lllccccccccccccccc} 
		\toprule\toprule \noalign{\smallskip} 
		
		& & & \multicolumn{5}{c}{RMISE} & \multicolumn{5}{c}{Pos.} & \multicolumn{4}{c}{$\mu$} & \multicolumn{1}{c}{$\rho$}  \\ 
		\cmidrule(l{0.5pt}r{9pt}){4-8}\cmidrule(l{0.5pt}r{9pt}){9-13}\cmidrule(l{0.5pt}r{9pt}){14-17} \cmidrule(l{0.5pt}r{9pt}){18-18} 
		$N$ & $R$ & $S$ & FKRB & MSE & OneSe & LL & PredOut & FKRB & MSE & OneSe & LL & PredOut & MSE & OneSe & LL & PredOut & 3rd Qu.  \\[0.2ex]

		\hline \\ [-2ex]

		1000 & 25 & 17 & 0.085 & 0.072 & 0.055 & 0.081 & 0.066 & 9.65 & 12.94 & 17.78 & 10.38 & 14.61 & 21.2 & 74.01 & 2.69 & 34.44 & 0.823 \\ [0.2ex]      
		1000 & 50 & 33 & 0.09 & 0.068 & 0.058 & 0.081 & 0.069 & 12.57 & 26.82 & 32.4 & 17.05 & 25.52 & 48.09 & 74 & 6.65 & 34.74 & 0.82 \\ [0.2ex]          
		1000 & 100 & 67 & 0.095 & 0.07 & 0.061 & 0.084 & 0.075 & 13.65 & 47.09 & 54.85 & 24.59 & 36.9 & 58.6 & 74.37 & 11.37 & 29.57 & 0.822 \\ [0.2ex]     
		1000 & 250 & 163 & 0.102 & 0.077 & 0.063 & 0.09 & 0.078 & 14.2 & 79.28 & 103.98 & 38.05 & 64.86 & 50.5 & 74.52 & 11.94 & 31.47 & 0.824 \\ [0.8ex]   
		10000 & 25 & 17 & 0.063 & 0.061 & 0.057 & 0.063 & 0.06 & 11.65 & 12.57 & 14.89 & 11.78 & 13.41 & 18.3 & 73.94 & 1.2 & 29.19 & 0.823 \\ [0.2ex]      
		10000 & 50 & 33 & 0.058 & 0.051 & 0.047 & 0.054 & 0.051 & 17.52 & 24.94 & 28.3 & 20.03 & 24.01 & 48.71 & 73.94 & 7.74 & 32.09 & 0.82 \\ [0.2ex]     
		10000 & 100 & 67 & 0.06 & 0.048 & 0.043 & 0.054 & 0.048 & 19.87 & 39.59 & 46.7 & 27.61 & 35.56 & 51.63 & 74.04 & 10.91 & 32.55 & 0.823 \\ [0.2ex]   
		10000 & 250 & 163 & 0.063 & 0.045 & 0.04 & 0.055 & 0.048 & 21.21 & 76.43 & 87.92 & 43.45 & 61.62 & 59.98 & 74.68 & 16.05 & 36.16 & 0.824 \\ [0.2ex]

		\hline \\[-0.5ex] 
		
		& & &  \multicolumn{5}{c}{$\%$ True Pos.} & \multicolumn{5}{c}{$\%$ Sign} & & & & &  \\ 
		\cmidrule(l{0.5pt}r{9pt}){4-8}\cmidrule(l{0.5pt}r{9pt}){9-13}   
		$N$ & $R$ & $S$ & FKRB & MSE & OneSe & LL & PredOut & FKRB & MSE & OneSe & LL & PredOut & & & & &   \\[0.2ex]

		\hline \\ [-2ex]

		1000 & 25 & 17 & 49.06 & 66.35 & 88.68 & 53.09 & 74.59 & 60.12 & 70.48 & 81.48 & 62.66 & 75 & & & & &  \\ [0.3ex]     
		1000 & 50 & 33 & 33.39 & 70.33 & 84.48 & 45.65 & 67.71 & 52.93 & 73.19 & 80.72 & 60.16 & 72.34 & & & & &  \\ [0.3ex]  
		1000 & 100 & 67 & 18.01 & 63.46 & 74.1 & 33.56 & 50.03 & 43.48 & 70.95 & 77.44 & 53.38 & 63.15 & & & & & \\ [0.3ex]  
		1000 & 250 & 163 & 7.37 & 44.38 & 58.17 & 21.11 & 36.25 & 38.73 & 60.96 & 69.06 & 47.1 & 56.13 & & & & & \\ [0.8ex]  
		10000 & 25 & 17 & 57.94 & 63.35 & 77.24 & 58.68 & 68.21 & 64.2 & 67.86 & 77.48 & 64.68 & 71.12 & & & & & \\ [0.3ex]  
		10000 & 50 & 33 & 47.24 & 68.59 & 78.05 & 54.62 & 66.05 & 61.32 & 74.66 & 80.42 & 66.04 & 73.16 & & & & & \\ [0.3ex] 
		10000 & 100 & 67 & 26.57 & 54.72 & 64.84 & 37.76 & 49.11 & 48.74 & 66.73 & 73.19 & 55.98 & 63.24 & & & & & \\ [0.3ex]
		10000 & 250 & 163 & 11.39 & 43.78 & 50.56 & 24.5 & 35.12 & 41.17 & 61.32 & 65.56 & 49.37 & 55.94 & & & & & \\ [0.3ex]

		\bottomrule\bottomrule \\[-1.8ex] 
		
	\end{tabular}
	\captionsetup{justification=justified,singlelinecheck=true, width=1.19\textwidth, font={footnotesize,stretch=0.8}}
	\caption*{\textit{Note:} The table reports the average summary statistics over all Monte Carlo replicates for the FKRB estimator (FKRB), and for our generalized estimator with tuning parameter $\mu$ from a 10-fold cross-validation and the $MSE$ criterion (MSE), the one-standard-error rule (OneSe), the log-likelihood criterion (LL) and the number of correctly predicted binary outcomes (PredOut). The predicted binary outcome is set to one for the alternative with the highest estimated choice probability.}
\end{table}

\end{landscape}

}

\restoregeometry
		
	
\begin{table}[H] \centering 
		\caption{First Stage Output of Mode Canada Data: Semiparametric Estimation with Normally Distributed Random Coefficient for the Total Travel Time.}
		\label{tables:FirstStage} 
		\begin{tabular}{@{\extracolsep{5pt}}lD{.}{.}{-3} } 
			\\[-1.8ex]\hline 
			\hline \\[-1.8ex] 
			& \multicolumn{1}{c}{\textit{Dependent variable:}} \\ 
			\cline{2-2} 
			\\[-1.8ex] & \multicolumn{1}{c}{Mode Choice} \\ 
			\hline \\[-1.8ex] 
			Intercept Train & -1.641^{***} \\ 
			& (0.304) \\ 
			Intercept Air & -7.153^{***} \\ 
			& (0.913) \\ 
			Frequency & 0.077^{***} \\ 
			& (0.008) \\ 
			Cost & -0.009 \\ 
			& (0.009) \\ 
			Income Train & -0.018^{***} \\ 
			& (0.003) \\ 
			Income Air & 0.040^{***} \\ 
			& (0.005) \\ 
			Distance Train & 0.002^{*} \\ 
			& (0.001) \\ 
			Distance Air & 0.003^{***} \\ 
			& (0.001) \\ 
			Urban Train & 1.722^{***} \\ 
			& (0.163) \\ 
			Urban Air & 1.261^{***} \\ 
			& (0.194) \\ 
			Travel Time & -0.017^{***} \\ 
			& (0.003) \\ 
			sd.Travel Time & 0.015^{***} \\ 
			& (0.002) \\ 
			\hline \\[-1.8ex] 
			Observations & \multicolumn{1}{c}{3,593} \\ 
			Mc Fadden R$^{2}$ & \multicolumn{1}{c}{0.358} \\ 
			Log Likelihood & \multicolumn{1}{c}{-2,340.700} \\ 
			LR Test & \multicolumn{1}{c}{2,615.034$^{***}$ (df = 12)  (p = 0.000)} \\ 
			\hline 
			\hline \\[-1.8ex] 
		\end{tabular} 
		\captionsetup{justification=justified,singlelinecheck=true, width=0.62\textwidth, font={footnotesize,stretch=0.8}}
		\caption*{\textit{Note:} The table reports the mean estimates and standard errors (in brackets) obtained by the \textit{mlogit package} for the semiparametric mixed logit model with normally distributed travel time.\\$^{*}$p$<$0.1; $^{**}$p$<$0.05; $^{***}$p$<$0.01.}
	\end{table}


\begin{table}[H] 
\centering
\caption{Estimated Own- and Cross-Travel Time Elasticities in Mode Canada Data.}   \label{tab:elasticities}
\small
\tabcolsep=0.2cm
\begin{tabular}{lccc} 
\toprule \toprule\noalign{\smallskip} 
			
\multicolumn{4}{l}{\textbf{Elasticities estimated with FKRB:}} \\ 
 & Car & Air & Train \\ [0.2ex] 
 \hline \\ [-2ex]
		
 Car &    -0.8992 \scriptsize{(-0.8444)} & 1.3982 \scriptsize{(0.6692)} &0.1164 \scriptsize{(0.129)} \\ [0.5ex] 
Air &    0.5895 \scriptsize{(0.5943)} & -1.2267 \scriptsize{(-0.5079)} & 0.2049 \scriptsize{(0.1589)} \\ [0.5ex] 
Train &  -0.1622 \scriptsize{(0.0346)} & 0.184 \scriptsize{(0.1352)} & -0.6712 \scriptsize{(-0.8861)} \\ [0.5ex] 
\hline \\

\multicolumn{4}{l}{\textbf{Elasticities estimated with ENet:}} \\ 
 & Car & Air & Train \\ [0.2ex] 
  \hline \\ [-2ex]
 
 Car &    -0.8382 \scriptsize{(-0.7731)} & 1.4082 \scriptsize{(0.682)} & 0.1473 \scriptsize{(0.1009)} \\ [0.5ex] 
Air &   0.5312 \scriptsize{(0.5034)} & -1.2581 \scriptsize{(-0.5704)} & 0.1765 \scriptsize{(0.1339)} \\ [0.5ex] 
Train & -0.0887 \scriptsize{(0.036)} & 0.19 \scriptsize{(0.1118)} & -0.6285 \scriptsize{(-0.7691)}\\  [0.5ex] 
 \hline \\

 \multicolumn{4}{l}{\textbf{Elasticities estimated semiparametrically:}} \\ 
  & Car & Air & Train \\ [0.2ex] 
 \hline \\ [-2ex]
 
 Car &    -1.3362 \scriptsize{ (-1.2584)} & 1.366 \scriptsize{ (0.9975)} & 0.6699  \scriptsize{ (0.6846)} \\ [0.5ex] 
Air &   0.6194 \scriptsize{ (0.6093)} & -1.3744 \scriptsize{ (-1.4473)} & 0.3076 \scriptsize{ (0.2281)} \\ [0.5ex] 
Train &  0.2772 \scriptsize{ (0.1824)} & 0.3111 \scriptsize{ (0.1563)} & -1.6449 \scriptsize{ (-1.7289)}\\ [0.5ex] 
\bottomrule\bottomrule
\end{tabular}
		\captionsetup{justification=justified,singlelinecheck=true, width=0.64\textwidth, font={footnotesize,stretch=0.8}}
		\caption*{\textit{Note:} The table reports the mean and the median (in brackets) over  individuals' own- and cross-travel time elasticities for the FKRB estimator, the elastic net estimator, and the semiparametric mixed logit with normal distribution. The reported numbers  correspond to the   percentage change of the choice probability of an alternative in a column after a one percent increase in the travel time of an alternative in a row.}
\end{table}

\begin{table}[H] 
\centering
\caption{Ratio of Estimated Own- and Cross-Travel Time Elasticities in Mode Canada Data.}   \label{tab:RatioElasticities}
\small
\tabcolsep=0.2cm
\begin{tabular}{lccl} 
\toprule \toprule\noalign{\smallskip} 
			
\multicolumn{4}{l}{\textbf{Estimated Elasticities of FKRB divided by those of ENet:}} \\ 
 & Car & Air & \hphantom{0.72} Train \\ [0.2ex] 
 \hline \\ [-2ex]
		
 Car &   1.0728 \scriptsize{(1.0922)} & 0.9929 \scriptsize{(0.9813)} & 0.7908 \scriptsize{(1.2783)}   \\ [0.5ex] 
Air &   1.1099 \scriptsize{(1.1804)} & 0.975 \scriptsize{(0.8905)} & 1.1605 \scriptsize{(1.1864)}    \\ [0.5ex] 
Train & 1.8291 \scriptsize{(0.9611)} & 0.9685 \scriptsize{(1.2098)} & 1.068 \scriptsize{(1.1521)} \\ [0.5ex] 
\hline \\

\multicolumn{4}{l}{\textbf{Semiparametrically estimated Elasticities divided by those of ENet:}} \\ 
 & Car & Air & \hphantom{0.72} Train \\ [0.2ex] 
  \hline \\ [-2ex]
 
 Car &   1.5941 \scriptsize{(1.6277)} & 0.9701 \scriptsize{(1.4627)} & 4.5492 \scriptsize{(6.7854)} \\ [0.5ex] 
Air &   1.1662 \scriptsize{(1.2103)} & 1.0925 \scriptsize{(2.5375)} & 1.7425 \scriptsize{(1.7035)} \\ [0.5ex] 
Train & -3.1268 \scriptsize{(5.0686)} & 1.6379 \scriptsize{(1.398)} & 2.6173 \scriptsize{(2.2478)} \\  [0.5ex] 

\bottomrule\bottomrule
\end{tabular}
		\captionsetup{justification=justified,singlelinecheck=true, width=0.83\textwidth, font={footnotesize,stretch=0.8}}
		\caption*{\textit{Note:} The table reports the ratio of the mean and the median (in brackets) over  individuals' own- and cross-travel time elasticities reported in Table \ref{tab:elasticities} for (1) the FKRB estimator and elastic net estimator and (2) the semiparametric mixed logit with normal distribution and the elastic net estimator.}
\end{table}

\bigskip
	
\section[Appendix B: \\ Algorithm to Update Fixed and Random Coefficients]{Algorithm to Update Fixed and Random Coefficients} \label{app:UpdatingfixedCoefficients}
\onehalfspacing
	The algorithm to update the fixed coefficients uses a modification of the flexible grid estimator in \citeA{train2008}.
	
	Let $F$ denote the set of indices corresponding to the fixed coefficients and $M$ to the set of indices corresponding to the random coefficients.
	The goal is to maximize with respect to the fixed coefficients $\beta^F$ and the weights  $\theta = (\theta_1, \ldots, \theta_R)$ corresponding to $\beta^M$. Therefore, define the vector which is to be maximized as $\pi = \{\beta_F,\theta\}$.

	Then, rewrite $z_{i,j}^r$ more explicitly:
	\begin{equation} \label{eq:ZwithFixed}
z_{i,j}^r := z_{i,j}(\beta^F,\beta_r^M) = g(x_{i,j},\beta^F, \beta_r^M) = \frac{\exp\left(x_{i,j}^F \beta^F + x_{i,j}^M \beta_r^M \right)}{1 + \sum\limits_{l = 1}^J\exp\left(x_{i,l}^F \beta^F + x_{i,l}^M \beta_r^M\right)} .
\end{equation}
The likelihood criterion given in \citeA{train2008} is
\begin{equation}
L L(\beta^F,\beta^M)=  \frac{1}{N} \sum\limits_{i = 1}^N\log \left( \sum\limits_{r = 1}^R \theta_r z_{i, y_i}^r \right) =  \frac{1}{N} \sum\limits_{i = 1}^N \log \left( \sum\limits_{r = 1}^R \theta_r z_{i,y_i}(\beta^F,\beta_r^M) \right)  .
\end{equation}

The probability of agent $i$ having coefficients $\pi$ conditional on her observed choice $y_i$ and  being type $r$ is 
\begin{equation} \label{eq:hProb}
h_{i,r}\left(\pi \right)=\frac{ \theta_r z_{i,y_i}(\beta^F,\beta_r^M)}{\sum\limits_{r = 1}^R \theta_r z_{i,y_i}(\beta^F,\beta_r^M)} .
\end{equation}

Based on Equation (\ref{eq:hProb}) one can derive the iterative EM update scheme which updates $\pi^{t+1} =  \{\beta_F,\theta\}^{t+1}  = \{ \beta_F,(\theta_1, \ldots, \theta_R)\}^{t+1}  $ by using a previous estimated trial $\pi^t$ to maximize

	\begin{align} \label{eq:updateEM} \nonumber
\pi^{t+1} &=\arg \max _{\pi} Q\left(\pi | \pi^{t}\right) \\
&=\arg \max _{\pi} \sum_{i=1}^{N} \sum_{r=1}^{R} h_{i,r}\left(\pi^{t} \right)  \log \left(\theta_{r}  z_{i,y_i}(\beta^F,\beta_r^M)\right) .
\end{align}
Since $ \log \left(\theta_{r}  z_{i,j}(\beta^F,\beta_r^M)\right)  = \log (\theta_{r} ) + \log(  z_{i,y_i}(\beta^F,\beta_r^M))$ one can maximize Equation (\ref{eq:updateEM}) separately for $\beta^F$ and $\theta$. Since we use our generalized estimator given in Equation (\ref{eq7}), we only maximize Equation (\ref{eq:updateEM}) over $\beta^F$:
	\begin{align} \label{eq:updateBetaFixed}
{\{\beta^F \}}^{t+1} 
&=\arg \max _{\beta^F} \sum_{i=1}^{N} \sum_{r=1}^{R} h_{i,r}\left(\pi^{t} \right)  \log \left( z_{i,y_i}(\beta^F,\beta_r^M)\right) .
\end{align}
Plugging Equation (\ref{eq:ZwithFixed}) into Equation (\ref{eq:updateBetaFixed}) gives
	\begin{align} \label{weightedLogit}
{\{\beta^F \}}^{t+1} 
&=\arg \max _{\beta^F} \sum_{i=1}^{N} \sum_{r=1}^{R} h_{i,r}\left(\pi^{t} \right)  \log \left(  \frac{\exp\left(x_{i,y_i}^F \beta^F + x_{i,y_i}^M \beta_r^M \right)}{1 + \sum\limits_{l = 1}^J\exp\left(x_{i,l}^F \beta^F + x_{i,l}^M \beta_r^M\right)} \right)
\end{align}
or equivalently 
	\begin{align} \label{weightedLogit}
{\{\beta^F \}}^{t+1} 
&=\arg \max _{\beta^F} \sum_{i=1}^{N} \sum_{j=1}^{J} \sum_{r=1}^{R} y_{i,j} h_{i,r}\left(\pi^{t} \right)  \log \left(  \frac{\exp\left(x_{i,j}^F \beta^F + x_{i,j}^M \beta_r^M \right)}{1 + \sum\limits_{l = 1}^J\exp\left(x_{i,l}^F \beta^F + x_{i,l}^M \beta_r^M\right)} \right) .
\end{align}

This is is the formula of a weighted (standard) logit model where only the coefficients $\beta^F$ are to be maximized and the coefficients $\beta^M$ are treated as constants. The weights $ h_{i,r}\left(\pi^{t} \right)$, calculated as given in Equation (\ref{eq:hProb}), do not depend on the product $j$, but differ for different observations $i$ and grid points $r$.\\

The whole update scheme is given by the following steps

\begin{framed}
\center{{\large \textbf{\textit{Generalized Estimator of Equation (\ref{eq7}) with fixed and random coefficients}}}}\\
\begin{enumerate}
\item Estimate semi-parametric model with all regressors and store the coefficients of the fixed parameters $\beta_0^F$.
\item Choose the grid points $\beta_r^M,\,r=1,...,R$.
\item Calculate the logit kernel, $z_{i,j}(\beta_0^F,\beta_r^M)$, for each agent at each point.
\item Estimate $\theta_0$ using the Generalized Estimator in Equation (\ref{eq7}).
\item  Calculate weights for each agent at each point with $\pi_0 = \{\beta_0^F, \theta_0 \}$ as 
\begin{equation} \nonumber
h_{i,r}\left(\pi_0 \right)=\frac{ {\theta_r}_0 z_{i,y_i}(\beta_0^F,\beta_r^M)}{\sum\limits_{r = 1}^R {\theta_r}_0 z_{i,y_i}(\beta_0^F,\beta_r^M)} .
\end{equation}
\item Update the fixed coefficients $\beta_0^F = \beta_1^F$ by estimating a weighted standard logit as specified in Equation (\ref{weightedLogit}) .
\item Repeat steps 3 and 6 until convergence, using the updated coefficients $\pi_0 = \pi_1$, where $\theta_0 = \theta_1$ is updated in step 4.
\item Use these estimated weights $\widehat{\theta}$ to calculate the estimated distribution\\
\begin{equation} \nonumber
\hat{F}\left(\beta\right) = \sum\limits_{r = 1}^R \thetahat_r \ 1\left[\beta_r\leq \beta\right].
\end{equation}
\end{enumerate}
\end{framed}

\newpage

\section[Appendix C: \\ Proofs of Results in Section \ref{sec:errorbound}]{Proofs of Results in Section \ref{sec:errorbound}}

\onehalfspacing

Below, we provide the proofs of the results presented in Section \ref{sec:errorbound}. For that purpose, we first introduce some additional notation.

Let $A$ be a $m \times n$ matrix and $x$ be a $n \times 1$ vector. In the following, the $\norm{A}_\infty$ norm refers to the matrix norm induced by the maximum norm of vectors. Then
\begin{equation} \nonumber
\norm{A}_\infty := \max\limits_{\vert\vert x \vert\vert_\infty = 1} \norm{Ax}_\infty = \max\limits_{1 \leq i \leq m } \sum_{j=1}^{n} \vert a_{ij} \vert
\end{equation}
denotes the maximum row sum of matrix $A$. $\norm{x}_\infty$  refers to the largest absolute element of  vector $x$.

Similarly, $\norm{A}_2$ is defined as the matrix norm induced by the euclidean vector norm. That is,
\begin{equation} \nonumber
\norm{A}_2 := \max\limits_{\vert\vert x \vert\vert_2 = 1} \norm{Ax}_2,
\end{equation}
is called spectral norm. It can be shown that $\norm{A}_2 = \max\limits_{1 \leq i \leq n } \sqrt{\psi_i(A^T A)}$ where $\psi_i(A^T A)$ denotes the eigenvalues of $A^T A$.

\subsection{Proof of Probability Bound}

Lemma \ref{lem:PrBound} uses Hoeffding's inequality to derive a probability bound for sub-Gaussian random variables. We use the lemma in the proofs of Theorems \ref{theo:ConditionSelectionLASSO} - \ref{theo:boundDistribution}. 


%

\begin{lemma} \label{lem:PrBound} Suppose Assumption \ref{ass:Exogen} holds. Then, for $\gamma \geq 0$ 
\begin{equation}   \nonumber
	\prob\left(\norm{\frac{1}{NJ}\Ztilde^T\epsilon}_{\infty}\geq\gamma\right) \leq 2 (R-1)J \exp\left(-\frac{N \gamma^2}{2}\right)  .
	\end{equation}
\end{lemma}


\begin{proof} 
	
	Notice that 
	
	\begin{equation}\label{probAbs}
	\prob\left(\norm{\frac{1}{NJ}\Ztilde^T\epsilon}_\infty\geq\gamma\right) = \prob\left(\max\limits_{1\leq r\leq {R-1}}\left\vert \frac{1}{NJ}\sum\limits_{i = 1}^N \Ztilde_{i}^{r T} \epsilon_i \right\vert\geq\gamma\right)
	\end{equation}

	\bigskip
	where $\epsilon_i = (\epsilon_{i,1}, \ldots, \epsilon_{i,J})$ denotes a random vector of $J$ dependent variables such that Equation (\ref{probAbs}) can equivalently be written as
	
	\begin{align*}
	\prob\left(\max\limits_{1\leq r\leq {R-1}}\left\vert\frac{1}{NJ}\sum\limits_{i = 1}^N \Ztilde_{i}^{r T} \epsilon_i\right\vert\geq\gamma\right) &= \prob\left(\max\limits_{1\leq r\leq {R-1}}\left\vert\frac{1}{NJ}\sum\limits_{i = 1}^N\sum\limits_{j = 1}^J\ztilde_{i,j}^r \epsilon_{i,j}\right\vert\geq\gamma\right) \\
	&=\prob\left(\bigcup\limits_{1\leq r\leq {R-1}}\left\{\left\vert \frac{1}{NJ}\sum\limits_{i = 1}^N\sum\limits_{j = 1}^J\ztilde_{i,j}^r\epsilon_{i,j}\right\vert\geq\gamma \right\}\right) .
	\end{align*} 
	
	From $\sum_{i=1}^N\sum_{j=1}^J\ztilde_{i,j}^r\epsilon_{i,j}\leq J \max\limits_{1\leq j\leq J}\sum_{i=1}^N\ztilde_{i,j}^r\epsilon_{i,j}$, we obtain the upper bound
	
	\begin{align*}
	\prob\left(\bigcup\limits_{1\leq r\leq {R-1}}\left\{\left\vert \frac{1}{NJ}\sum\limits_{i = 1}^N\sum\limits_{j = 1}^J\ztilde_{i,j}^r\epsilon_{i,j}\right\vert\geq\gamma \right\}\right) &\leq \prob\left(\bigcup\limits_{1\leq r \leq {R-1}}\left\{J\max\limits_{1\leq j \leq J}\left\vert\frac{1}{NJ}\sum\limits_{i = 1}^N\ztilde_{i,j}^r\epsilon_{i,j}\right\vert\geq\gamma\right\}\right) \notag \\
	&\leq\sum\limits_{r=1}^{R-1}\prob\left(\max\limits_{1\leq j \leq J}\left\vert\frac{1}{N}\sum\limits_{i=1}^N\ztilde_{i,j}^r\epsilon_{i,j}\right\vert\geq\gamma\right) \notag \\
	&=\sum\limits_{r=1}^{R-1}\prob\left(\bigcup\limits_{1\leq j\leq J}\left\{\left\vert\frac{1}{N}\sum\limits_{i=1}^N\ztilde_{i,j}^r\epsilon_{i,j}\right\vert\geq\gamma\right\}\right) \notag \\
	&\leq\sum\limits_{r=1}^{R-1}\sum\limits_{j=1}^J\prob\left(\left\vert\frac{1}{N}\sum\limits_{i=1}^N\ztilde_{i,j}^r\epsilon_{i,j}\right\vert\geq\gamma\right) \notag \\
	&\leq (R-1) J  \max\limits_{\subalign{ 1 \le r &\le R-1   \\  1 \le j &\le J  }}\prob\left(\left\vert\frac{1}{N}\sum\limits_{i=1}^N\ztilde_{i,j}^r\epsilon_{i,j}\right\vert\geq \gamma\right) .
	\end{align*} 
	
	Recall from Assumption \ref{ass:Exogen} (iii) and Equation (\ref{eq:true}) that $-1\leq \ztilde_{i,j}^r\leq 1$ and $-1\leq \epsilon_{i,j}\leq 1$.
	Therefore, $\xi := (\ztilde_{1,j}^r\epsilon_{1,j}, \ldots, \ztilde_{N,j}^r\epsilon_{N,j})$ is a vector of independent uniformly bounded random variables since for every $i = 1, \ldots, N$ it holds that $-1\leq \ztilde_{i,j}^r\epsilon_{i,j}\leq 1$. It follows from the assumption of conditional exogeneity (Assumption \ref{ass:Exogen} (iv)) that $\mathop{\mathbb{E}}[\xi] = 0$. Due to the boundedness of $\xi$, its moment generating function satisfies
	
	\begin{equation*}
	\mathop{\mathbb{E}}\left[\exp(s\xi)\right] \leq \exp\left(\frac{\sigma^2s^2}{2}\right).
	\end{equation*}
	
	For any $s\in \mathbb{R}$, $\xi$ is said to be sub-Gaussian with variance proxy $\sigma^2$. Thus, using Hoeffding's inequality, 
	
	\begin{equation}\label{hoeffdingsIneq} 
	\max\limits_{\subalign{ 1 \le r &\le R-1  \\  1 \leq j &\leq J}}\prob\left(\left\vert\frac{1}{N}\sum\limits_{i=1}^N\ztilde_{i,j}^r\epsilon_{i,j}\right\vert\geq \gamma\right) \leq 2 \exp\left(-\frac{N\gamma^2}{2\sigma^2}\right) .
	\end{equation}
	
	It follows from $\xi\in[-1, 1]$ that $\sigma^2 = 1$. Therefore, 

	\begin{align}\label{boundFirstStep} \nonumber
	\prob\left(\norm{\frac{1}{NJ}\Ztilde^T\epsilon}_\infty\geq\gamma\right) &\leq (R-1)J \max\limits_{\subalign{ 1\le r &\le R-1 \\  1\leq j &\leq J}}\prob\left(\left\vert\frac{1}{N}\sum\limits_{i=1}^N\ztilde_{i,j}^r\epsilon_{i,j}\right\vert\geq \gamma\right) \\  &\leq 2 (R-1)J \exp\left(-\frac{N \gamma^2}{2}\right)  .
	\end{align}
	
\end{proof}

\subsection{Proof of Selection Consistency} 

In the following, we provide the proof of Theorem \ref{theo:ConditionSelectionLASSO}. We first derive two sufficient conditions in Lemma \ref{lem:MUandMV} that ensure that the estimated weights are equal in sign, i.e. $\thetahat =_s \thetastar $. Lemma \ref{lem:V} provides a bound on the probability of the first sufficient condition and Lemma \ref{lem:U} a bound on the probability of the second sufficient condition. Finally, we use Lemma \ref{lem:V} and Lemma \ref{lem:U} to prove Theorem \ref{theo:ConditionSelectionLASSO}. Both Lemma \ref{lem:V} and Lemma \ref{lem:U} employ Lemma \ref{lem:SVD}. 

\begin{lemma} \label{lem:SVD} It holds that
	\begin{equation} \nonumber
	\normB{ \left(\frac{1}{NJ}\Ztilde_S^T \Ztilde_S + \mu \imat_S \right)^{-1}}_\infty  \leq   \sqrt{s}   \frac{1}{\xi_{\min}^S(\mu)} .
	\end{equation}
\end{lemma}

\begin{proof}
	Using Singular Value Decomposition (SVD), rewrite $\Ztilde_S$ as
	\begin{equation} \label{SVD}
	\frac{1}{\sqrt{ NJ}} \Ztilde_S = A D M^T
	\end{equation}
	where $A$ is a $NJ \times s $ matrix with orthogonal columns, i.e. $A^T A = I_S$.\\
	$M$ is a $s \times s$ orthogonal matrix satisfying $M^T M = M M^T = I_S$. $D$ is a diagonal $s \times s$ matrix consisting of the singular values of $(1/\sqrt{ NJ}) \Ztilde_S$ on its diagonal. 	
	We apply the SVD in Equation (\ref{SVD}) to rewrite
	\begin{align} \label{svdInverse}
	\left( \frac{1}{NJ}\Ztilde_S^T\Ztilde_S + \mu \imat_S \right)^{-1} = \left(    M D^T A^T A D M^T+ \mu \imat_S \right)^{-1} &=  \nonumber
	\left(    M D^2 M^T+ \mu M  M^T  \right)^{-1} \\ &= M \left(     D^2 + \mu \imat_S \right)^{-1} M^T 
	\end{align}
	
	Therefore, 
	\begin{align}  \label{upperBoundInverse}
	\normB{ \left(\frac{1}{NJ}\Ztilde_S^T \Ztilde_S + \mu \imat_S \right)^{-1}}_\infty &= \normB{ M \left(     D^2 + \mu \imat_S \right)^{-1} M^T }_\infty 
	\leq  \sqrt{s} \normB{ M \left(     D^2 + \mu \imat_S \right)^{-1} M^T }_2 \\
	&=  \sqrt{s} \normB{  \left(     D^2 + \mu \imat_S \right)^{-1} }_2 
	=  \sqrt{s} \max\limits_{i \in S} \sqrt{\psi_i}  \nonumber \\ 
	&=  \sqrt{s}  \max\limits_{i \in S} \frac{1}{d_{ii}^2 + \mu}
	=  \sqrt{s}   \frac{1}{\min\limits_{i \in S} d_{ii}^2 + \mu} =  \sqrt{s}   \frac{1}{\xi_{\min}^S(\mu)} \nonumber 
	\end{align}
	where $\psi_i$ denotes the eigenvalues of $\left( \left(     D^2 + \mu \imat_S \right)^{-1}\right)^T \left(     D^2 + \mu \imat_S \right)^{-1} = 
	\left(     D^2 + \mu \imat_S \right)^{-2}$. Thus, $\psi_i= \left(d_{ii}^2 + \mu  \right)^{-2}$, as the eigenvalues of a diagonal matrix are its diagonal entries. The (unrestricted) eigenvalues of $1/(NJ) \Ztilde_S^T \Ztilde_S + \mu \imat_S $ are defined as $\xi^S(\mu)$. $\xi_{\min}^S(\mu)$ corresponds to the minimal eigenvalue of the matrix.
	The first inequality in Equation (\ref{upperBoundInverse}) holds by the relation of the absolute row sum norm and the spectral norm. The transformation from the first to the second line follows from the invariance of the spectral norm to orthogonal transformations \cite[pp. 130-131]{gentle2007}. The equality in the second line follows from the spectral norm. The last equality in Equation (\ref{upperBoundInverse}) holds by the relation of singular values to eigenvalues.
\end{proof}

\begin{lemma} \label{lem:MUandMV} 
	Sufficient conditions for $\thetahat =_s \theta^*$ are
	\begin{align*}
	\mathcal{M}(V) := \left\{\max\limits_{j \in S^C} V_j \leq  \lambda \right\} ,
	\end{align*}
	\begin{align*}
	\mathcal{M}(U) := \left\{\max\limits_{i \in S} \vert U_i \vert <  \rho \right\}
	\end{align*}
	where
	\begin{flalign*}
	V  &:= \frac{1}{NJ}\Ztilde_{S^C}^T\bigg[\Ztilde_S \left( \frac{1}{NJ}\Ztilde_S^T\Ztilde_S + \mu \imat_S \right)^{-1} \left(\lambda\iota_S + \mu\thetastar_S - \frac{1}{NJ}\Ztilde_S^T\epsilon \right) + \epsilon\bigg] ,\\
	U   &:= \left( \frac{1}{NJ}\Ztilde_S^T\Ztilde_S + \mu \imat_S \right)^{-1}   \frac{1}{NJ}\Ztilde_S^T\epsilon , \\
	\rho  &:= \min\limits_{i \in S} \Big\vert \left( \frac{1}{NJ}\Ztilde_S^T\Ztilde_S + \mu \imat_S \right)^{-1} \left(  \frac{1}{NJ}\Ztilde_S^T\Ztilde_S \thetastar_S  - \lambda\iota_S \right)\Big\vert.
	\end{flalign*}
\end{lemma}	
\begin{proof}
	The Lagrangian of our adjusted estimator that follows from the transformed optimization problem in Equation (\ref{eq7}) is 
	
	\begin{equation}\label{lagrangian}
	L(\theta) := \quad\frac{1}{2NJ}\vert\vert \ytilde - \Ztilde\theta\vert\vert + \lambda_n\left(\iota^T\theta - 1\right) + \frac{1}{2}\mu \ \theta^T\theta - \nu^T\theta
	\end{equation}
	
	which is minimized with respect to $\theta$, i.e. $\theta=\argmin\limits_{\theta} L(\theta)$. $\lambda$ and $\nu$ are Lagrangian multipliers that enforce that the estimated weights sum to one and that they are non-negative respectively. $\mu > 0$ is an additional tuning parameter. Note that for $\mu = 0$, Equation (\ref{lagrangian}) corresponds to the objective function of the estimator by \citeA{fox2011}. 
	
	To analyze the support recovery of our estimator, we follow the proof in \citeA{jia2010}. 
	The estimator recovers the true support of the distribution if 
	%
	every estimated probability weight $\thetahat$ has the same sign as the true weights $\thetastar$, i.e. $\thetahat =_s \thetastar$. 
	
	This is the case if the Karush-Kuhn-Tucker (KKT) conditions to the optimization problem in Equation (\ref{lagrangian}) are satisfied. The KKT conditions are given by
	
	
	
	\begin{align}
	-&\frac{1}{NJ}\Ztilde^T\left(\ytilde - \Ztilde\hat{\theta}\right) + \lambda\iota + \mu \ \hat{\theta} - \nu = 0 \label{kt1} ,\\ 
	&\lambda\left(\iota^T\thetahat - 1\right) = 0 \label{kt2} ,\\
	&\nu_r \ \thetahat_r = 0   \label{kt3} ,\\ 
	&\lambda \geq 0, \quad \nu_r \geq 0 \hspace{2.5cm} \forall \quad r = 1, \ldots, R - 1 \label{kt4} .
	\end{align}
	
	Denote the set of grid points where the true distribution has positive probability mass by $S = \{r\in\{1, \ldots, R-1\} \vert\theta_r^{*}>0\}$  and let $S^C = \{r\in\{1, \ldots, R-1\} \vert\theta_r^{*} = 0\}$ denote its complement set. The corresponding cardinalities are defined as $s := |S|$ and $ s^C :=|S^C| $.  We refer to grid points in $S$ as active grid points and to grid points in $S^C$ as inactive grid points.	
	Splitting $\thetahat$, $\Ztilde$ and $\nu$ over $S$ and $S^C$ into two blocks gives 
	
	\begin{equation} \notag
	-\frac{1}{NJ}\left[\Ztilde_S \ \Ztilde_{S^C}\right]^T\left(\ytilde - \left[\Ztilde_S \ \Ztilde_{S^C}\right]\left(
	\begin{array}{c}
	\hat{\theta}_S \\ \hat{\theta}_{S^C} \\ \end{array}\right) \right) + \lambda \iota + \mu \left(\begin{array}{c}
	\hat{\theta}_S\\ \hat{\theta}_{S^C} \\ \end{array}\right) - \left(\begin{array}{c}
	\nu_S\\ \nu_{S^C} \\ \end{array}\right) = 0.
	\end{equation}
	
	Recall that $\thetastar_r = 0$ for all grid points outside $S$, so that $\Ztilde\thetastar = \Ztilde_S\thetastar_S$. In order to recover the active grid points, it must hold that $\thetahat =_s \thetastar$ which implies $\thetahat_{S^C} = 0$. The two conditions that follow from Equation (\ref{kt1}) require 
	
	\begin{align}
	-&\frac{1}{NJ}\Ztilde_S^T\left(\ytilde - \Ztilde_S\thetahat_S\right)+ \lambda\iota_S + \mu\hat{\theta}_S - \nu_S = 0 \label{kt1pos} , \\
	-&\frac{1}{NJ}\Ztilde_{S^C}^T\left(\ytilde - \Ztilde_{S}\thetahat_{S}\right) + \lambda\iota_{S^C} - \nu_{S^C} = 0 \label{kt1zer} .
	\end{align}
	
	Note that $\thetahat_S > 0$ and $\thetahat_{S^C} = 0$ imply 
	\begin{align}
	&\nu_r = 0 \hspace{3cm} \forall \quad r \in S  \label{cond1} ,\\
	&\nu_r \geq 0 \hspace{3cm} \forall \quad r \not\in S \label{cond2} .
	\end{align}
	
	It follows from Condition (\ref{cond1}) that Condition (\ref{kt1pos}) simplifies to
	
	\begin{align}
	-&\frac{1}{NJ}\Ztilde_S^T\left(\ytilde - \Ztilde_S\thetahat_S\right)+ \lambda\iota_S + \mu\hat{\theta}_S  = 0  .
	\end{align}

	\bigskip

	Substituting the true model $\ytilde = \Ztilde\thetastar + \epsilon$, we can re-express the required conditions as
	
	\begin{equation} \label{cond3}
	-\frac{1}{NJ}\Ztilde_S^T\Ztilde_S\left(\thetastar_S - \thetahat_S\right) - \frac{1}{NJ}\Ztilde_S^T\epsilon + \lambda\iota_S + \mu\thetahat_S = 0
	\end{equation}
	
	and 
	
	\begin{equation}\label{cond4}
	-\frac{1}{NJ}\Ztilde_{S^C}^T\Ztilde_S\left(\thetastar_S - \thetahat_S\right) - \frac{1}{NJ}\Ztilde_{S^C}^T\epsilon + \lambda\iota_{S^C} - \nu_{S^C} = 0 .
	\end{equation}
	
	\bigskip

	Reformulating Condition (\ref{cond3}) gives
	\begin{equation} \label{betahat_s}
	\thetahat_S =  \underbrace{\bigg( \frac{1}{NJ}\Ztilde_S^T\Ztilde_S + \mu \imat_S \bigg)^{-1} \bigg(  \frac{1}{NJ}\Ztilde_S^T\epsilon  }_{=:U} + \frac{1}{NJ}\Ztilde_S^T\Ztilde_S \thetastar_S  - \lambda\iota_S  \bigg)  > 0 
	\end{equation}
	where the positivity constraint follows from the KKT conditions and the definition of $\thetahat_S$.\\

	Plugging Equation (\ref{betahat_s}) into Equation (\ref{cond4}) and using Condition (\ref{cond2}) yields
	\begin{equation} \label{secondKKT}
	\underbrace{\frac{1}{NJ}\Ztilde_{S^C}^T\bigg[\Ztilde_S \left( \frac{1}{NJ}\Ztilde_S^T\Ztilde_S + \mu \imat_S \right)^{-1} \left(\lambda\iota_S + \mu\thetastar_S - \frac{1}{NJ}\Ztilde_S^T\epsilon \right) + \epsilon\bigg]}_{=:V} \ \leq \lambda\iota_{S^C} .
	\end{equation}
	$U$ and $V$ are defined in Equation (\ref{betahat_s}) and Equation (\ref{secondKKT}), respectively.
	The vector $U$ consists of $s$ elements $U_i,$ $i \in S$, and is constructed from the conditions on the positive weights, and vector $V$ from the condition on the zero weights. Therefore, $V$ has $R-s$ elements $V_j$, $j \in S^C$.
	Condition (\ref{secondKKT}) is equivalent to the event
	\begin{equation} \label{MV} \notag
	\mathcal{M}(V) := \left\{\max\limits_{j \in S^C} V_j \leq  \lambda \right\}.
	\end{equation}
	
	The event $\mathcal{M}(U)$ defines a condition for the positive weights
	\begin{equation} \label{MU} \notag
	\mathcal{M}(U) := \left\{\max\limits_{i \in S} \vert U_i \vert <  \rho \right\}
	\end{equation}
	%
	where $\rho := \min\limits_{i \in S} \vert g_i \vert $  with  $g_i := \Big[ \left( \frac{1}{NJ}\Ztilde_S^T\Ztilde_S + \mu \imat_S \right)^{-1} \left(  \frac{1}{NJ}\Ztilde_S^T\Ztilde_S \thetastar_S  - \lambda\iota_S \right) \Big]_i \,$.\\
	
	Therefore, the event $ \mathcal{M}(U)$ implies
	\begin{align*} 
	0 & <  \rho -  \max\limits_{i \in S} \vert U_i \vert < \rho- \vert U_i \vert 
	< \vert g_i \vert  - \vert U_i \vert 
	< \vert  g_i +  U_i \vert = \vert  \thetahat_{S_{i}} \vert = \thetahat_{S_{i}}, \;\; \forall i \in S
	\end{align*}
	where $g_i$, $U_i$ and $\thetahat_{S_{i}}$ denote the $i$th element of the respective vectors $g$, $U$ and $\thetahat_S$. The second last equality holds by definition of $g_i$ and $U_i$ (see Equation (\ref{betahat_s})) and the last inequality by the reverse triangle inequality. Because the weights are constrained to be nonnegative by the KKT conditions, the absolute value $\vert  \thetahat_{S_{i}} \vert $ can be omitted. Consequently, $ \mathcal{M}(U)$ is a sufficient condition for Equation (\ref{betahat_s}) to hold and thus for $\thetahat_S > 0 $.\\
	
\end{proof}

\begin{lemma}  \label{lem:V} Suppose Assumption (\ref{ass:Exogen}) holds. Suppose further that the \hyperref[cond:eic]{NEIC} holds. Let $ \mathcal{M}^C(V)$ denote the complement of $ \mathcal{M}(V)$. Then, 
	\begin{equation} \nonumber
	\prob\big( \mathcal{M}^C(V) \big) \leq 2 (R-1) J \exp\left(-\frac{N \eta^2 \lambda^2 \left( \frac{\xi_{\min}^S(\mu)}{s \sqrt{s} +\xi_{\min}^S(\mu)}  \right)^2}{2}\right) .
	\end{equation}
\end{lemma}

\begin{proof}
	$V_j$ is sub-Gaussian with mean
	\begin{equation} \label{meanV} \notag
	\overline{V} := E(V) = \frac{1}{NJ}\Ztilde_{S^C}^T\Ztilde_S \left( \frac{1}{NJ}\Ztilde_S^T\Ztilde_S + \mu \imat_S \right)^{-1} \left(\lambda\iota_S + \mu\thetastar_S  \right) .  
	\end{equation}
	
	Recall the Nonnegative Elastic Net Irrepresentable Condition (\hyperref[cond:eic]{NEIC}) is
	\begin{equation} \notag
	\max\limits_{r\in S^C} \frac{1}{NJ}\Ztilde_{S^C}^T\Ztilde_S \left( \frac{1}{NJ}\Ztilde_S^T\Ztilde_S + \mu \imat_S \right)^{-1} \left(\iota_S + \frac{\mu}{\lambda}\thetastar_S  \right)  \leq 1 - \eta .
	\end{equation}
	
	Therefore, $\overline{V}_j \leq (1-\eta) \lambda$. Let $\Vtilde := \frac{1}{NJ}\Ztilde_{S^C}^T\bigg[-\Ztilde_S \left( \frac{1}{NJ}\Ztilde_S^T\Ztilde_S + \mu \imat_S \right)^{-1}   \frac{1}{NJ}\Ztilde_S^T  + \imat_{NJ} \bigg] \epsilon $ such that $V = \overline{V} + \Vtilde$.\\
	
	Consequently, it holds for the complement of $\mathcal{M}(V)$ that
	
	\begin{equation}  \nonumber
	\lambda < \max\limits_{j \in S^C} V_j = \max\limits_{j \in S^C} (\overline{V}_j + \Vtilde_j ) \leq \max\limits_{j \in S^C} \overline{V}_j + \max\limits_{j \in S^C} \Vtilde_j 
	\iff   \max\limits_{j \in S^C} \Vtilde_j >  \lambda -  \max\limits_{j \in S^C} \overline{V}_j \geq \lambda - (1-\eta) \lambda = \eta \lambda .
	\end{equation}
	
	We use the last inequality to derive an upper bound on $ \mathcal{M}^C(V)$:

	\begin{align*} \label{probV} \nonumber
	\prob\left(  \mathcal{M}^C(V) \right) & = \prob\left(   \max\limits_{j \in S^C} V_j  > \lambda  \right)
	\leq \prob\left(   \max\limits_{j \in S^C}\Vtilde_j  >  \eta \lambda  \right) 
	\leq \prob\left(   \max\limits_{j \in S^C} \vert \Vtilde_j  \vert >  \eta \lambda  \right) \\
	& =  \prob\left(   \max\limits_{j \in S^C} \left\vert \frac{1}{NJ}\Ztilde_{S^C}^T\bigg[-\Ztilde_S \left( \frac{1}{NJ}\Ztilde_S^T\Ztilde_S + \mu \imat_S \right)^{-1}   \frac{1}{NJ}\Ztilde_S^T  + I \bigg] \epsilon  \right\vert >  \eta \lambda  \right) \\
	& \leq  \prob\left(   \max\limits_{j \in S^C} \left\vert \frac{1}{NJ}\Ztilde_{S^C}^T \Ztilde_S \left( \frac{1}{NJ}\Ztilde_S^T\Ztilde_S + \mu \imat_S \right)^{-1}   \frac{1}{NJ}\Ztilde_S^T  \epsilon \right\vert + \max\limits_{j \in S^C} \left\vert  \frac{1}{NJ} \Ztilde_{S^C}^T \epsilon  \right\vert >  \eta \lambda  \right) \\
	& =  \prob\left( \normB{  \frac{1}{NJ}\Ztilde_{S^C}^T \Ztilde_S \left( \frac{1}{NJ}\Ztilde_S^T\Ztilde_S + \mu \imat_S \right)^{-1}   \frac{1}{NJ}\Ztilde_S^T  \epsilon  }_\infty + \max\limits_{j \in S^C} \left\vert  \frac{1}{NJ} \Ztilde_{S^C}^T \epsilon  \right\vert >  \eta \lambda  \right)  \\
	& \leq  \prob\left( \normB{  \frac{1}{NJ}\Ztilde_{S^C}^T \Ztilde_S}_\infty \, \normB{ \left( \frac{1}{NJ}\Ztilde_S^T\Ztilde_S + \mu \imat_S \right)^{-1}  }_\infty \, \normB{ \frac{1}{NJ}\Ztilde_S^T  \epsilon  }_\infty + \max\limits_{j \in S^C} \left\vert  \frac{1}{NJ} \Ztilde_{S^C}^T \epsilon  \right\vert >  \eta \lambda  \right) .
	\end{align*} 
	The last inequality holds due the  property of the absolute row sum norm that  $\norm{A B x}_\infty \leq \norm{A}_\infty \norm{B}_\infty \norm{x}_\infty$ for arbitrary matrices $A$, $B$ and a vector $x$.
	
	By Lemma \ref{lem:SVD} and $\norm{\frac{1}{NJ}\Ztilde_{S^C}^T \Ztilde_S}_\infty \leq s $ (since every entry in $\Ztilde$ is at most $1$ in absolute value, and thus the absolute row sum of $\frac{1}{NJ}\Ztilde_{S^C}^T \Ztilde_S$ at most $\frac{1}{NJ} s NJ = s$), we obtain
	\begingroup
	\allowdisplaybreaks
	\begin{align*} 
	\prob\left(  \mathcal{M}^C(V) \right)      & \leq  \prob\left( s \sqrt{s} \frac{1}{\xi_{\min}^S(\mu)} \max\limits_{j \in S} \left\vert \frac{1}{NJ}  \Ztilde_{S^C}^T \epsilon  \right\vert + \max\limits_{j \in S^C} \left\vert \frac{1}{NJ}  \Ztilde_{S^C}^T \epsilon  \right\vert >  \eta \lambda  \right) \\
	& \leq  \prob\left( s \sqrt{s} \frac{1}{\xi_{\min}^S(\mu)} \max\limits_{j \in R} \left\vert \frac{1}{NJ}  \Ztilde^T \epsilon  \right\vert + \max\limits_{j \in R} \left\vert \frac{1}{NJ}  \Ztilde^T \epsilon  \right\vert >  \eta \lambda  \right) \\
	& =   \prob\left( \Big(s \sqrt{s} \frac{1}{\xi_{\min}^S(\mu)}  + 1 \Big) \max\limits_{j \in R} \left\vert \frac{1}{NJ}  \Ztilde^T \epsilon  \right\vert >  \eta \lambda  \right) \\
	& \leq  \prob\left(    \max\limits_{j \in R} \left\vert \frac{1}{NJ}  \Ztilde^T \epsilon  \right\vert > \eta \lambda \frac{1}{s \sqrt{s} \frac{1}{\xi_{\min}^S(\mu)} + 1 }  \right) .
	\end{align*}
	Applying Hoeffding's inequality with $\gamma =  \eta \lambda \frac{1}{s \sqrt{s} \frac{1}{\xi_{\min}^S(\mu)} + 1 } $ as outlined in Lemma \ref{lem:PrBound} gives
	\begin{align*}                
	\prob\left(  \mathcal{M}^C(V) \right)                     & \leq  2 (R-1) J \exp\left(-\frac{N \left(\eta \lambda \frac{1}{s \sqrt{s} \frac{1}{\xi_{\min}^S(\mu)}+1}  \right)^2}{2\sigma^2}\right) \\
	& =  2 (R-1) J \exp\left(-\frac{N \left(\eta \lambda \frac{\xi_{\min}^S(\mu)}{s \sqrt{s} +\xi_{\min}^S(\mu)}  \right)^2}{2\sigma^2}\right) \\
	&=  2 (R-1) J \exp\left(-\frac{N \eta^2 \lambda^2 \left( \frac{\xi_{\min}^S(\mu)}{s \sqrt{s} +\xi_{\min}^S(\mu)}  \right)^2}{2}\right) .
	\end{align*}
	\endgroup
\end{proof}


\begin{remark}
	The above calculations can be simplified to for the baseline estimator, i.e. if $\mu = 0$. Assume that the NIC condition for LASSO holds (\hyperref[cond:eic]{NEIC} with $\mu = 0 $). Additionally, note that it holds for $\mu \geq 0$ that 
	\begin{equation} \nonumber
	\left(\frac{1}{NJ}\Ztilde_S^T\Ztilde_S + \mu \imat_S \right)^{-1}  \Ztilde_S^T = \Ztilde_S^T \left( \frac{1}{NJ}\Ztilde_S \Ztilde_S^T + \mu \imat_N \right)^{-1}.   
	\end{equation}
	Using the above equality for $\mu = 0$, we obtain
	\begin{align*} 
	\prob\left(   \max\limits_{j \in S^C} V_j  > \lambda  \right)
	& \leq \prob\left(   \max\limits_{j \in S^C}\Vtilde_j  >  \eta \lambda  \right) 
	\leq \prob\left(   \max\limits_{j \in S^C} \vert \Vtilde_j  \vert >  \eta \lambda  \right) \\
	& =  \prob\left(   \max\limits_{j \in S^C} \left\vert \frac{1}{NJ}\Ztilde_{S^C}^T\bigg[-\Ztilde_S \left( \frac{1}{NJ}\Ztilde_S^T\Ztilde_S \right)^{-1}   \frac{1}{NJ}\Ztilde_S^T  + \imat_S \bigg] \epsilon  \right\vert >  \eta \lambda  \right) \\
	& =  \prob\left(   \max\limits_{j \in S^C} \left\vert \frac{1}{NJ}\Ztilde_{S^C}^T\bigg[- \frac{1}{NJ} \Ztilde_S \Ztilde_S^T \left( \frac{1}{NJ}\Ztilde_S \Ztilde_S^T  \right)^{-1}     + \imat_S \bigg] \epsilon  \right\vert >  \eta \lambda  \right) \\
	& =  \prob\left(   \max\limits_{j \in S^C} \left\vert \frac{1}{NJ}\Ztilde_{S^C}^T\bigg[- \imat_S    + \imat_S \bigg] \epsilon  \right\vert >  \eta \lambda  \right) \\
	& =  \prob\left(  0 >  \eta \lambda  \right) = 0 
	\end{align*}
	since $\eta \lambda > 0$.
\end{remark}

\hspace{2cm}

\begin{lemma} 
	\label{lem:U} Suppose Assumption (\ref{ass:Exogen}) holds. Let $ \mathcal{M}^C(U)$ denote the complement of $ \mathcal{M}(U)$. Then,
	\begin{equation} \nonumber
	\prob\big( \mathcal{M}^C(U) \big) \leq 2 s J \exp\left(-\frac{N    \xi_{\min}^S(\mu)^2  \rho^2  }{2 s}\right) .
	\end{equation}
\end{lemma}
\begin{proof}
	Because $U$ is sub-Gaussian with mean 0, the probability of the complement of $\mathcal{M}(U)$ corresponds to
	
	
	
	\begin{align*} \label{probU} \nonumber
	\prob\left(  \mathcal{M}^C(U) \right)  & =
	\prob\left(    \max\limits_{i \in S} \vert U_i \vert \geq \rho \right) \\
	& =  \prob\left(    \max\limits_{i \in S}  \left( \frac{1}{NJ}\Ztilde_S^T\Ztilde_S + \mu \imat_S \right)^{-1}   \frac{1}{NJ}\Ztilde_S^T\epsilon   \geq \rho \right) \\
	& \leq  \prob\left(   \normB{ \left( \frac{1}{NJ}\Ztilde_S^T\Ztilde_S + \mu \imat_S \right)^{-1} }_\infty \, \normB{ \frac{1}{NJ}\Ztilde_S^T\epsilon }_\infty  \geq \rho \right) .
	\end{align*}
	In the next step Lemma \ref{lem:SVD} is applied again.
	\begin{align*}  
	\prob\left(  \mathcal{M}^C(U) \right)    & \leq  \prob\left(    \sqrt{s}   \frac{1}{\xi_{\min}^S(\mu)} \normB{ \frac{1}{NJ}\Ztilde_S^T\epsilon }_\infty  \geq \rho \right) \\
	& \leq  \prob\left(    \normB{ \frac{1}{NJ}\Ztilde_S^T\epsilon }_\infty  \geq     \xi_{\min}^S(\mu) \frac{1}{\sqrt{s}} \rho \right) \\
	& \leq  2 s J \exp\left(-\frac{N \left(   \xi_{\min}^S(\mu) \frac{1}{\sqrt{s}} \rho  \right)^2}{2\sigma^2}\right) 
	=  2 s J \exp\left(-\frac{N    \xi_{\min}^S(\mu)^2  \rho^2  }{2 s \sigma^2}\right)  \\
	& =  2 s J \exp\left(-\frac{N    \xi_{\min}^S(\mu)^2  \rho^2  }{2 s}\right) 
	\end{align*}
	where the last inequality follows from  Hoeffding's inequality in Lemma \ref{lem:PrBound} with $\gamma =   \xi_{\min}^S(\mu) \frac{1}{\sqrt{s}} \rho $.
\end{proof}

We use the above lemmata to prove Theorem \ref{theo:ConditionSelectionLASSO}.

\bigskip
\paragraph{Proof of Theorem \ref{theo:ConditionSelectionLASSO}.} \label{app:ProofOfSelectionConsistency}
\begin{proof}[\unskip\nopunct] It holds that
	\begin{equation} \nonumber
	\prob\left( \thetahat =_s \theta \right) \geq \prob\big(  \mathcal{M}(V) \cap \mathcal{M}(U) \big)
	\end{equation}
	since $\mathcal{M}(U)$ is a sufficient condition for the selection of the true weights according to Lemma \ref{lem:MUandMV}.

Under the condition that \hyperref[ass:ConditionSelectionLASSO]{RCDG} holds,  applying Lemma \ref{lem:V} and  Lemma \ref{lem:U}  gives  $ \lim\limits_{N\rightarrow\infty} \prob\left(  \mathcal{M}^C(V)  \right) = 0 $ and  $\lim\limits_{N\rightarrow\infty} \prob\left(  \mathcal{M}^C(U)  \right) = 0 $.\\

	Thus,
	\begin{align*} \nonumber
	\lim\limits_{N\rightarrow\infty} \prob\left( \thetahat =_s \theta \right) &\geq \lim\limits_{N\rightarrow\infty} \prob\big(  \mathcal{M}(V) \cap \mathcal{M}(U) \big) \\
	&\geq  \lim\limits_{N\rightarrow\infty} \left\{ 1 - \prob\left(  \mathcal{M}^C(V)  \right) - \prob\left(  \mathcal{M}^C(U) \right) \right\} \\
	&= 1 .	 
	\end{align*}
\end{proof}

\hspace{3cm}

\subsection{Proof of Error Bounds}

In the following, we first provide the proof of the error bound of the estimated weights presented in Theorem \ref{theo:l2normweights} and the proof of Corollary \ref{core:l2normweights}. We then use the derived bound to proof the error bound of the estimated random coefficients' distribution in Theorem \ref{theo:boundDistribution}. In the proofs of Theorem \ref{theo:l2normweights} and Theorem \ref{theo:boundDistribution}, we apply Lemma \ref{lem:PrBound}. 

\paragraph{Proof of Theorem \ref{theo:l2normweights}.}  \label{app:ProofErrorBoundWeights}
\begin{proof}[\unskip\nopunct]
	Note that if $\thetahat$ is the solution to the Lagrangian in Equation (\ref{lagrangian}), it must hold that it minimizes  (\ref{lagrangian}), i.e.  $L(\thetahat) \leq L(\theta)$ for any  $\theta$. Thus, it holds that $L(\thetahat) \leq L(\thetastar)$ where $\thetastar$ are the true weights. Applying this to the objective function in (\ref{lagrangian}), we obtain 
	
	\begin{equation*}
	\frac{1}{2NJ}\norm{\ytilde - \Ztilde\thetahat}^2_2 + \lambda\left(\iota^T\thetahat - 1\right) + \frac{\mu}{2}\thetahat^T\thetahat \leq \frac{1}{2NJ}\norm{\ytilde - \Ztilde\thetastar}^2_2 + \lambda\left(\iota^T\thetastar - 1\right) + \frac{\mu}{2}\thetastarT \thetastar .
	\end{equation*}
	
	Substituting the true model $\ytilde = \Ztilde\thetastar + \epsilon$ into the above condition and simplifying gives
	
	\begin{equation*}
	\frac{1}{2NJ}\norm{\Ztilde\left(\thetastar - \thetahat\right) + \epsilon}^2_2 + \lambda\left(\iota^T\thetahat - 1\right) + \frac{\mu}{2}\thetahat^T\thetahat \leq \frac{1}{2NJ}\norm{\epsilon}^2_2 + \lambda\left(\iota^T\thetastar - 1\right) + \frac{\mu}{2}\thetastarT \thetastar .
	\end{equation*}
	
	Taking into account that 
	
	$$\norm{ \Ztilde(\thetastar - \thetahat) + \epsilon}^2_2 = \norm{ \Ztilde(\thetastar - \thetahat)}^2_2 + \norm{\epsilon}^2_2 + 2\epsilon^T(\Ztilde(\thetastar - \thetahat)) \notag$$ 
	
	we obtain 
	
	\begin{align}\label{eq11ErrorBound}
	&\frac{1}{2NJ}\norm{\Ztilde\left(\thetastar - \thetahat\right)}_2^2 + \lambda\left(\iota^T\thetahat -1 \right) + \frac{\mu}{2}\thetahat^T\thetahat \leq \notag \\
	& \frac{1}{NJ}\epsilon^T \Ztilde \left(\thetahat - \thetastar\right) + \lambda\left(\iota^T\thetastar - 1\right) + \frac{\mu}{2}\thetastarT \thetastar .
	\end{align}
	
	Note that $\epsilon^T\Ztilde(\thetahat - \thetastar)\leq \norm{ \Ztilde^T\epsilon}_\infty \norm{\thetahat - \thetastar}_1$. 
	
	Applying  Lemma \ref{lem:PrBound} with  $\gamma \equiv \gamma(N,\delta) := \sqrt{2\log\left(\frac{2(R-1)J}{\delta}\right)\bigg/N}$  we obtain
	
	\begin{align} \label{probGamma}
	\prob\left(\norm{\frac{1}{NJ}\Ztilde^T\epsilon}_\infty\geq \gamma \right) &\leq 2 (R-1)J \exp\left(-N\left(\sqrt{\frac{2\log\left(\frac{2(R-1)J}{\delta}\right)}{N}}\right)^2\bigg/2\right) \notag \\
	&= 2(R-1)J\exp\left(\log\left(\left(\frac{2(R-1)J}{\delta}\right)^{-1}\right)\right) \notag \\
	&= \delta .
	\end{align}
	
	In the following, we assume  that   $\{(1/(NJ))\vert\vert \Ztilde^T\epsilon\vert\vert_\infty\leq\gamma\}$, which happens with probability at least $1-\delta$ according to Equation (\ref{probGamma}). Therefore, the rest of the proof holds with probability $1-\delta$. Using that the event $\{(1/(NJ))\vert\vert \Ztilde^T\epsilon\vert\vert_\infty\leq\gamma\}$ occurs, we can bound the the right hand side in Equation (\ref{eq11ErrorBound}) from above by
	
	\begin{equation}\label{eq12ErrorBound}
	\frac{1}{2NJ}\norm{\Ztilde\left(\thetastar - \thetahat\right)}^2 + \lambda\left(\iota^T\thetahat -1\right) + \frac{\mu}{2}\thetahat^T\thetahat \leq \gamma\norm{\thetahat - \thetastar}_1 + \lambda\left(\iota^T\thetastar - 1\right) + \frac{\mu}{2} \thetastarT  \thetastar .
	\end{equation}
	
	We split $\thetahat$, $\Ztilde$ and $\nu$ over $S$ and $S^C$ into two blocks, whereby $S$ again denotes the set of relevant grid points for which the true weights $\thetastar > 0$ and $S^C$ the set of points for which $\thetastar = 0$. It follows that 
	
	$$\iota^T\theta = \iota^T_S\theta_S + \iota^T_{S^C}\theta_{S^C} = \vert\vert\theta_S\vert\vert_1 + \vert\vert\theta_{S^C}\vert\vert_1\notag$$ 
	
	and 
	
	$$\theta^T\theta = \theta_S^T\theta_S + \theta_{S^C}^T\theta_{S^C}\notag .$$ 
	
	Thus, we can reformulate Equation (\ref{eq12ErrorBound}) as 
	
	\begin{align*}
	&\frac{1}{2NJ}\norm{\Ztilde\left(\thetastar - \thetahat\right)}^2_2 + \lambda\left(\norm{\thetahat_S}_1 + \norm{\thetahat_{S^C}}_1 - 1\right) + \frac{\mu}{2}\left(\thetahat^T_S\thetahat_S + \thetastarT_{S^C}\thetastar_{S^C}\right) \leq \notag\\
	&\gamma\norm{\thetahat - \thetastar}_1 + \lambda\left(\normb{\thetastar_S}_1 + \normb{\thetastar_{S^C}}_1 - 1\right) + \frac{\mu}{2}\left(\thetastarT_S \thetastar + \thetastarT_{S^C}\thetastar_{S^C}\right) .
	\end{align*}
	
	It follows from $\thetastar_{S^C} = 0$ that $\vert\vert\thetahat - \thetastar\vert\vert_1 = \vert\vert\thetahat_S - \thetastar_S\vert\vert_1 + \vert\vert\thetahat_{S^C}\vert\vert_1$ such that after some simple manipulations we obtain 
	
	\begin{align}\label{eq13ErrorBound}
	&\frac{1}{2NJ}\norm{\Ztilde\left(\thetastar - \thetahat\right)}_2^2 + \lambda\left(\norm{\thetahat_S}_1 + \norm{\thetahat_{S^C}}_1 - 1\right) + \frac{\mu}{2}\left(\thetahat^T_S\thetahat_S - \thetastarT_S\thetastar_S + \thetahat_{S^C}^T\thetahat_{S^C}\right) \leq \notag \\
	&\gamma\norm{\thetahat - \thetastar}_1  + \lambda\left(\normb{\thetastar_S}_1 - 1\right) .
	\end{align}
	
	Note that the terms in (\ref{eq13ErrorBound}) that are multiplied by the Langrangian parameter $\lambda$ drop out. Recall that by the definition of a linear probability model, $\vert\vert\thetastar_S\vert\vert_1 - 1 = 0$. With respect to the second term, $\lambda(\vert\vert\thetahat_S\vert\vert_1 + \vert\vert\thetahat_{S^C}\vert\vert_1 - 1)$, there are two different cases to be considered due to the inequality constraint $\sum_{r=1}^R\theta_r \leq 1$: (1) the estimated probability weights sum to one (the constraint is binding), and (2) the sum of the estimated probability weights is less than one (the constraint is not binding). In the former case, $\vert\vert\thetahat_S\vert\vert_1 + \vert\vert\thetahat_{S^C}\vert\vert_1 - 1 = 0$. In the latter case, the KKT conditions require $\lambda=0$. Thus, Condition (\ref{eq13ErrorBound}) simplifies to 
	
	\begin{equation}\label{cond7}
	\frac{1}{2NJ}\norm{\Ztilde\left(\thetastar - \thetahat\right)}^2_2 + \frac{\mu}{2}\left(\thetahat^T_S\thetahat_S - \thetastarT_S\thetastar_S + \thetahat_{S^C}^T\thetahat_{S^C}\right)\leq \gamma\norm{\thetahat - \thetastar}_1 .
	\end{equation}

	It follows from $\vert\vert\thetahat_S - \thetastar_S\vert\vert_2^2 = \thetahat_S^T\thetahat_S - 2\thetastarT_S\thetahat_S + \thetastarT_S\thetastar_S$ that

	\begin{equation}
	\thetahat_S^T\thetahat_S - \thetastarT_S\thetastar_S + \thetahat_{S^C}^T\thetahat_{S^C}= \norm{\thetahat_S - \thetastar_S}_2^2 + 2\thetastarT_S\thetahat_S - 2\thetastarT_S\thetastar +  \norm{\thetahat_{S^C}}_2^2 \notag 
	\end{equation}
	
	and from  $\thetastar_{S^C} = 0$ that $\vert\vert\thetahat_{S^C}\vert\vert_p = \vert\vert\thetahat_{S^C} - \thetastar_{S^C}\vert\vert_p$ for $p = 1, 2$.
	
	Consequently, we can collect the terms over the index sets $S$ and $S^C$ to $\vert\vert\thetahat_{S} - \thetastar_S\vert\vert_1 + \vert\vert\thetahat_{S^C}\vert\vert_1 = \vert\vert\thetahat - \thetastar\vert\vert_1$ and $\vert\vert\thetahat_S - \thetastar_S\vert\vert_2^2 + \vert\vert\thetahat_{S^C}\vert\vert_2^2 = \vert\vert\thetahat - \thetastar\vert\vert_2^2$.
	
	This yields
	\begin{equation}
	\thetahat_S^T\thetahat_S - \thetastarT_S\thetastar_S + \thetahat_{S^C}^T\thetahat_{S^C}= \norm{\thetahat - \thetastar}_2^2 + 2\thetastarT_S\thetahat_S - 2\thetastarT_S\thetastar  \notag .
	\end{equation}

	Therefore, Equation (\ref{cond7}) can be equivalently expressed as 	
	
	\begin{align}\label{cond8}
	&\frac{1}{2NJ}\norm{\Ztilde\big(\thetastar - \thetahat\big)}_2^2 + \frac{\mu}{2}\norm{\thetahat - \thetastar}^2_2 \leq \notag \\
	&\gamma\norm{\thetahat - \thetastar}_1 + \frac{\mu}{2}\bigg( 2\thetastarT_S\thetastar_S - 2\thetastarT_S\thetahat_S \bigg) .
	\end{align}

	Next, because $\thetastar_S > 0$ and $\vert\vert \thetahat_{S} - \thetastar_{S}\vert\vert_1  \leq \sqrt{s} \vert\vert\thetahat_{S} - \thetastar_{S}\vert\vert_2 $ it holds that
	
	\begin{align}\label{cond10}
	\thetastarT_S\left(\thetastar_S - \thetahat_S\right)  \leq   \thetastarT_S \left\vert \thetahat_{S} - \thetastar_{S} \right\vert 
	\leq  \normb{\thetastar_S}_\infty\norm{\thetahat_{S} - \thetastar_{S}}_1 
	\leq \sqrt{s} \normb{\thetastar_S}_\infty\norm{\thetahat_{S} - \thetastar_{S}}_2 
	\end{align}
	where $\vert \thetahat_{S} - \thetastar_{S} \vert $ takes the absolute value of each element of the vector $\thetahat_{S} - \thetastar_{S}$.

	Substituting Condition (\ref{cond10}) back into the error bound in Equation (\ref{cond8}) and using the the fact that $\vert\vert\thetahat - \thetastar\vert\vert_1\leq\sqrt{(R-1)} \ \vert\vert\thetahat - \thetastar\vert\vert_2$, for $\gamma \leq k\lambda$, we can rewrite Equation (\ref{cond8}) as

	\begin{align}\label{cond11}
	&\frac{1}{2NJ}\norm{\Ztilde\big(\thetastar - \thetahat\big)}_2^2 + \frac{\mu}{2}\norm{\thetahat - \thetastar}^2_2 \leq 
	k\lambda\sqrt{(R-1)}\norm{\thetahat - \thetastar}_2 + \mu\sqrt{s}\normb{\thetastar_S}_\infty\norm{\thetahat_S - \thetastar_S}_2 .
	\end{align}
	
	Recall that 
	\begin{equation*}
	\norm{\Ztilde\big(\thetahat - \thetastar\big)}^2_2 = \big(\thetahat - \thetastar\big)^T\Ztilde^T\Ztilde\big(\thetahat - \thetastar\big)
	\end{equation*}
	
	and that the left-hand-side in Condition (\ref{cond11}) can be summarized as
	

	\begin{align}\label{cond12}
	&\frac{1}{2}\big(\thetahat - \thetastar\big)^T\bigg[\frac{1}{NJ}\Ztilde^T\Ztilde + \mu\imat\bigg]\big(\thetahat - \thetastar\big)  \leq 
	\bigg(k\lambda\sqrt{(R-1)} + \mu\sqrt{s}\normb{\thetastar_S}_\infty\bigg)\norm{\thetahat - \thetastar}_2 .
	\end{align}

	%
	%
	%
	%
	%
	Recall that $\xi_{\min}(\mu)$ defines the minimum eigenvalue of the real symmetric matrix $1/(NJ)\Ztilde^T\Ztilde + \mu\imat$ over the set of vectors $\mathcal{H}$ (see Subsection (\ref{subsec:bound})).
	
	It holds that $\xi_{\min}(\mu)>0$ if $\mu > 0 $ and that $\xi_{\min}\geq 0$ if $\mu = 0 $. In the following, we assume $\xi_{\min}(\mu)>0$.
	
	Thus, multiplying the left-hand-side in Condition (\ref{cond12}) by $\vert\vert\thetahat-\thetastar\vert\vert^2_2 / \vert\vert\thetahat - \thetastar\vert\vert^2_2$ and using the restricted minimum eigenvalue definition gives the upper $\ell_2$-error bound between the estimated and true probability weights:
	
	\begin{align*}\label{cond13} 
	&\frac{\xi_{\min}(\mu)}{2}\norm{\thetahat - \thetastar}^2_2 \leq\notag \bigg(k\lambda\sqrt{(R-1)} + \mu\sqrt{s}\norm{\thetastar_S}_\infty\bigg)\norm{\thetahat - \thetastar}_2 \notag \\
	\Rightarrow\quad&\norm{\thetahat - \thetastar}_2 \leq \frac{2\sqrt{(R-1)} \ k\lambda + 2\mu\sqrt{s} \norm{\thetastar_S}_\infty}{\xi_{\min}(\mu)} .
	\end{align*}
\end{proof}	

\paragraph{Proof of Corollary \ref{core:l2normweights}.}  \label{app:ProofErrorBoundCorrollary}
\begin{proof}[\unskip\nopunct]
	By assumption, it holds that
	\begin{align*}
	\left(\sqrt{(R-1)} \ k\lambda + \mu\sqrt{s} \norm{\thetastar_S}_\infty \right)\xi_{\min}(0) & \leq  \sqrt{(R-1)} \ k\lambda \xi_{\min}(0) +  \mu \sqrt{(R-1)} \ k\lambda \\
	& =  \sqrt{(R-1)} \ k\lambda (\xi_{\min}(0) + \mu ) . 
	\end{align*}
	Using $ \xi_{\min}(\mu) = \xi_{\min}(0) + \mu$ gives
	\begin{align*}
	\left(\sqrt{(R-1)} \ k\lambda + \mu\sqrt{s} \norm{\thetastar_S}_\infty \right)\xi_{\min}(0) \leq \sqrt{(R-1)} \ k\lambda \xi_{\min}(\mu)
	\end{align*}
	which is equivalent to
	\begin{align*}
	\frac{2 \sqrt{(R-1)} \ k\lambda + 2 \mu\sqrt{s} \norm{\thetastar_S}_\infty }{\xi_{\min}(\mu)} \leq \frac{2 \sqrt{(R-1)} \ k\lambda}{ \xi_{\min}(0)} .
	\end{align*}
\end{proof}

\paragraph{Proof of Theorem \ref{theo:boundDistribution}.}  \label{app:ProofErrorBoundDistribution}
\begin{proof}[\unskip\nopunct]

	It holds that the difference of $\hat{F}\left(\beta\right)$ and $ F^*(\beta)$ in any point $\beta \in \mathbb{R}^K$ can be bounded by
	\begin{align*}
	\left\vert \hat{F}\left(\beta\right) - F^*(\beta) \right\vert  & =   \left\vert  \sum\limits_{r = 1}^R \thetahat_r \ 1\left[\beta_r\leq \beta\right] - \sum\limits_{r = 1}^R\thetastar_r \ 1\left[\beta_r\leq \beta\right] \right\vert  \\
	& \leq \sup_{\beta} \left\vert  \sum\limits_{r = 1}^R \left(\thetahat_r -\thetastar_r \right)  \ 1\left[\beta_r\leq \beta\right]  \right\vert \\
	& \leq  \sum\limits_{r = 1}^R \left\vert  \thetahat_r - \thetastar_r \right\vert =   \sum\limits_{r = 1}^{R-1} \left\vert  \thetahat_r - \thetastar_r \right\vert +  
	\left\vert  \thetahat_R - \thetastar_R \right\vert 
	\end{align*}
	where the last inequality holds by the triangle inequality.
	
	Then, 
	\begin{align*}
	\left\vert \hat{F}\left(\beta\right) - F^*(\beta) \right\vert  & \leq  \sum\limits_{r = 1}^{R-1} \left\vert  \thetahat_r - \thetastar_r \right\vert +  
	\Big\vert  1- \sum\limits_{r = 1}^{R-1}   \thetahat_r - 1 + \sum\limits_{r = 1}^{R-1} \thetastar_r \Big\vert \\
	&= \sum\limits_{r = 1}^{R-1} \Big\vert  \thetahat_r - \thetastar_r \Big\vert +  
	\Big\vert   \sum\limits_{r = 1}^{R-1}  \left( \thetastar_r - \thetahat_r  \right) \Big\vert 
	\leq 2 \sum\limits_{r = 1}^{R-1} \left\vert  \thetahat_r - \thetastar_r \right\vert  \\
	& =  2 \norm{\thetahat - \thetastar  }_1  \leq 2 \sqrt{(R-1)} \norm{\thetahat - \thetastar  }_2 ,
	\end{align*}
	which, by Theorem \ref{theo:l2normweights}, can be bounded by
	
	\begin{align*}
	\vert \hat{F}\left(\beta\right) - F^*(\beta) \vert  \leq  2 \sqrt{(R-1)}   \; \frac{2\sqrt{(R-1)} \ k\lambda + 2\mu\sqrt{s} \norm{\thetastar_S}_\infty}{\xi_{\min}(\mu)} .
	\end{align*}

\end{proof}

\end{appendix}

\end{document}